\documentclass[prl]{revtex4}
\usepackage{amsmath,enumerate}
\usepackage{amssymb}
\usepackage{graphicx}
\usepackage{amsthm}
\newtheorem{theorem}{Theorem}
\newtheorem{lemma}{Lemma}
\newtheorem{cor}{Corollary}
\newtheorem{prop}{Proposition}

\newcommand{\be}{\begin{equation}}
\newcommand{\ee}{\end{equation}}
\newcommand{\bra}[1]{\left\langle #1\right|}
\newcommand{\ket}[1]{\left|#1\right\rangle}
\newcommand{\braket}[2]{\left\langle #1,#2\right\rangle}
\newcommand{\ketbra}[2]{|#1\rangle\langle #2|}
\newcommand{\pure}[1]{\ketbra{#1}{#1}}
\newcommand{\Tr}{{\mathop{\mathrm{Tr}}}}
\providecommand{\one}{\leavevmode\hbox{\small1\kern-3.8pt\normalsize1}}
\renewcommand\Im{\operatorname{Im}}
\newcommand{\Z}{\mathbb{Z}}
\newcommand{\err}{\check{\epsilon}}

\newcommand{\DX}{{\cal D}_X}
\newcommand{\DY}{{\cal D}_Y}
\newcommand{\gap}{\gamma}
\newcommand{\A}{\mathcal{A}}
\newcommand{\N}{\mathbb{N}}
\newcommand{\R}{\mathbb{R}}
\newcommand{\Sa}{\mathcal{S_{\alpha}}}
\newcommand{\diam}{\operatorname{diam}}
\makeatletter
\def\imod#1{\allowbreak\mkern10mu({\operator@font mod}\,\,#1)}
\makeatother

\begin{document}

\title{Quantization of Hall Conductance For Interacting Electrons Without Averaging Assumptions}
\author{Matthew B. Hastings}
\email{mahastin@microsoft.com}
\affiliation{Microsoft Research Station Q, CNSI Building, University of California, Santa Barbara, CA,
93106, USA}
\author{Spyridon Michalakis}
\email{spiros@lanl.gov}
\affiliation{T-4 and CNLS, LANL - Los Alamos, NM, 87544, USA}

\begin{abstract}

We consider two-dimensional Hamiltonians on a torus with finite range,
finite strength interactions and a unique ground state with a
non-vanishing spectral gap, and a conserved local charge, as defined
precisely in the text.  Using the local charge operators,
we introduce a boundary magnetic flux in the horizontal and vertical
direction and evolve the
ground state quasi-adiabatically around a square of size one magnetic
flux, in flux space.
At the end of the evolution we obtain a trivial Berry phase, which we
compare, via a method reminiscent of Stokes' Theorem, to the Berry
phase obtained from an evolution around a small loop near the origin.
As a result, we prove, without any averaging assumption, that the Hall
conductance for interacting electron systems is quantized in integer
multiples of $e^2/h$ up to small corrections bounded by a function
that decays as a stretched exponential in the linear size $L$.
Finally, we discuss extensions to the fractional case under an additional
topological order assumption to describe the multiple degenerate ground states.
\end{abstract}
\maketitle

At low temperature, the Hall conductance of a quantum system can be
quantized to remarkable precision.  While this is an experimental fact,
and the essential ingredients of our intuitive understanding of it were
provided by Laughlin\cite{laughlin}, there is still no fully satisfactory
mathematical proof of why this happens.  The main approaches theoretically
are either the Chern number approach\cite{avron,niu}, which relies on
an additional averaging assumption, and the noncommutative geometry approach\cite{ncg} which is only applicable to non-interacting electrons.
In this paper, we present an approach which avoids averaging and is applicable
to interacting electrons, and can even be extended to handle fractional
quantization.

Before we formulate our main result, we discuss our setup of the Quantum Hall Effect.
We will consider a discrete, tight-binding model of interacting electrons with orbitals centered at sites on a torus $T$. In particular, $T$ is obtained by joining the boundaries of a finite $[1,L]\times [1,L]$ subset of $\Z^2$. At each site $s\in T$, we introduce the charge operator $q_s$ with eigenvalues $0$, $1, \dots , q_{\max}$
representing the states with the respective charge occupying the site.
Let $\A_Z$ be the algebra of observables associated with $Z\subset T$. 
From a general point of view, we are interested in properties of the ground state
of the Hamiltonian $H_0 = \sum_{Z\subset T} \Phi(Z)$, where the following assumptions are satisfied:
\begin{enumerate}
\item The interactions $\Phi(Z)$ have finite strength and finite range.
Formally,  $\exists\,  J, R > 0$ with $L>2R$ such that:
\begin{enumerate}
\item $\Phi(Z) = \Phi(Z)^{\dagger}\in \A_Z$, 
\item $\sup_{s\in T} \sum_{Z\ni s} \| \Phi(Z) \| \leq J$, 
\item $\forall Z\subset T,\quad \diam(Z) > R \implies \Phi(Z) = 0,$
\end{enumerate}
where
\be
\diam(Z) = \max_{\{ s_1,s_2\in Z\}} \{d(s_1,s_2)\}, \quad d(s_1,s_2) = |x(s_1)-x(s_2)| \imod{L} + |y(s_1)-y(s_2)| \imod{L},\nonumber
\ee
with $x(s)$, $y(s)$ denoting the $x$ and $y$ coordinates of site $s$, respectively.
\item The Hamiltonian $H_0$ has a unique,
normalized ground state, which we will denote by $\ket{\Psi_0}$, and a
spectral gap $\gap > 0$ to the first excited state.  We use $P_0=|\Psi_0\rangle \langle \Psi_0|$ to denote the projector onto the ground state.
\item Finally, the total charge $Q = \sum_{s\in T} q_s$ is conserved, so that $[Q,H_0]=0$.
\end{enumerate}

Our main result is the following:
\begin{theorem}
Let $H_0$ be a Hamiltonian satisfying the above assumptions.
Then, 
the Hall conductance $\sigma_{xy}$ of $\ket{\Psi_0}$, defined in (\ref{def:cond}) for a system of interacting particles described by $H_0$, satisfies the quantization condition 
\be
\label{mainres}
\left|\sigma_{xy} - n \cdot\frac{e^2}{h}\right| \le C \left(q_{\max} R^2 \frac{J}{\gamma} L\right)^{5/2}\, G_{R,J,\gamma}^{5/2}(L) \, e^{-G^2_{R,J,\gamma}(L)/6}, \quad G_{R,J,\gamma}(L)= \left(\frac{\gamma}{4 v R}\right) \left(\frac{L}{48(2\pi) q_{\max} \ln^3 L}\right)^{1/5},
\ee
for some $n \in \N$, for all $L$ such that the following inequality is satisfied:
\be
\label{Lrequire}
G^2_{R,J,\gamma}(L)\geq C' \ln\left(q_{max} R^2 \frac{J}{\gamma} L\right).
\ee
In (\ref{mainres}),
$e^2/h$ denotes the square of the electron charge divided by Planck's constant.
The quantities $C,C'$ are
numeric constants of order unity, and $v = 132e\, (1+R)^3 J<360 (1+R)^3 J$. 
\end{theorem}
This implies that, for fixed $J,R,\gamma$ and $q_{\max}$, the Hall conductance is exactly quantized to
an integer multiple of $e^2/h$ in the thermodynamic ($L\rightarrow \infty$) limit.  Note
that, for fixed $R,J,\gamma,q_{max}$, the inequality (\ref{Lrequire})
is satisfied for all sufficiently large $L$.

Before giving the proof, we discuss the applicability of our assumptions to physical experimental systems.
The first assumption includes only finite range interactions.   While there are long range Coulomb interactions in real
experimental systems,
the screening of Coulomb interactions may justify this assumption. Moreover, the case in which the
torus is not exactly square, but has an aspect ratio of order unity, can be
handled by combining several sites in one direction into a single site to
make the aspect ratio unity; in fact, with minor changes in the proof,
a polynomial aspect ratio can be handled at the cost of polynomial changes in the prefactor of (\ref{mainres}).
We omit this case for simplicity.
Our techniques can also be extended to the case of exponentially decaying interactions, which we also omit.

We note here that sites $s, s' \in T$ with $y(s)=y(s')$ and $x(s) = L, x(s')=1$ are neighboring sites on $T$. The same is true for $s, s' \in T$ with $x(s)=x(s')$ and $y(s) = L, y(s')=1$. Thus, interactions $\Phi(X)$ in $H_0$ on subsets $X$ which cross the boundaries $x= 0$ and $y=0$ are not necessarily vanishing.  
Finally, we note that without loss of generality, we may assume that each term $\Phi(Z)$ of the initial Hamiltonian commutes with the total charge $Q$, by
substituting $\Phi(Z)$ with the averaged interaction which has the same range $R$ and strength still bounded by $J$:
$$\frac{1}{2\pi}\int_0^{2\pi} e^{i \theta Q}\, \Phi(Z)\, e^{-i \theta Q}\, d\theta.$$

\subsection*{Sketch of the main argument.}
Our proof is closely related to the Chern number approach\cite{avron,niu}, so we review it and then contrast our technique.
In the Chern number approach, one computes the Hall conductance of the above system of interacting electrons by introducing magnetic fluxes through twists $\theta_x, \theta_y$ at the boundary of the torus.  A time-dependent flux $\theta_x$ creates an electric
field in the $x$ direction and the flux $\theta_y$ measures the current in the $y$ direction, and thus the Hall
conductance can be identified with a certain current-current correlation evaluated at the ground state $\ket{\Psi_0}$.  In the Chern number approach, one
must
assume that there is a unique ground state for all $\theta_x,\theta_y$.  
If the ground state is adiabatically transported
around an infinitesimal loop near $\theta_x=\theta_y=0$, the state acquires a phase proportional to the area of the loop multiplied
by the Hall conductance.  That is, the Hall conductance is equal to the curvature of the connection given by parallel transport of
the ground state.  The average of the Hall conductance over the torus is then equal to an integer, the Chern number.
Unfortunately, this argument  requires a non-vanishing gap for all $\theta_x,\theta_y$ and
only provides a quantization result for the averaged Hall conductance and does not provide
a quantization result for $\theta_x=\theta_y=0$.

Our proof avoids both of these assumptions, while only requiring the presence of a gap at $\theta_x=\theta_y=0$.
Roughly, our argument proceeds as follows.  Given any state and any path in parameter space, we define the
\emph{quasi-adiabatic evolution}, introduced in \cite{hast-lsm}, of that state along the path. 
Various different quasi-adiabatic evolutions have been defined in \cite{hast-lsm,hast-qad,osborne}; the evolution we choose here
is yet a different one. 
The essential features of quasi-adiabatic evolution are that the evolution of a state be described by an approximately local operator
acting on that state and {\it so long as the path is in a region of the parameter space with a sufficiently large spectral gap},
the ground state of the Hamiltonian at the start of the path is quasi-adiabatically evolved into a state close to the ground state of
the Hamiltonian at the end of the path, up to multiplication by a phase.

Our proof then rests on four results. First, in corollary
\ref{cor:adiabatic_phase_3},
we show that the quasi-adiabatic
evolution of the ground state
around a small (but non-infinitesimal) square loop near the origin in parameter space ($\theta_x=\theta_y=0$) gives a state which
is close to the ground state multiplied by a phase, which we relate to
the Hall conductance at $\theta_x=\theta_y=0$.

Second, we show that quasi-adiabatic evolution of the ground state following the path starting at $(0,0)$, going to $(0,\theta_y)$, then to $(\theta_x,\theta_y)$,
evolving around a small loop at the given
$(\theta_x,\theta_y)$, and then returning to $(0,0)$ again gives a state
close to the ground state up to
a phase.  We do {\it not} require any lower bound on the gap on the path in this step, and if the gap does become small then the
quasi-adiabatic evolution around this path may not approximate the adiabatic evolution.  Rather, the proof that we return to the
ground state depends on locality estimates.  Crucially, we show 
in lemma \ref{lem:translation}
that the phase produced by evolving around this path is approximately
the same for any $\theta_x,\theta_y$.
Third, we show, in a procedure reminiscent of Stokes' theorem, that the quasi-adiabatic evolution around a large loop in parameter space,
$$\Lambda: (0,0) \rightarrow (2\pi,0) \rightarrow (2\pi,2\pi) \rightarrow (0, 2\pi) \rightarrow (0,0),$$ can be exactly decomposed
into the product of a large number of quasi-adiabatic evolutions around small loops as is described at the beginning of this paragraph (see also, Fig.~\ref{fig:decomposition}).
Fourth, we show in corollary 
\ref{cor:gs_evol}
that quasi-adiabatic evolution of the ground state around the large loop returns to a state which is close to
the ground state with a trivial overall phase.  

Combining these four steps, the trivial phase produced by quasi-adiabatic evolution around the large loop is close to
the product of the phases around each small loop.  Since the phase produced by evolution around each small loop is approximately the same,
we can then show that the phase around a small loop near the origin, raised to a power equal to the number of small loops, is close to unity.
Since this phase is related to the Hall conductance, this gives the desired proof of quantization.

The most difficult technical step is the second one.  The first step relies only on the fact that quasi-adiabatic continuation
for systems with a large gap approximates adiabatic evolution.  The fourth step
is similar to the proof of the Lieb-Schultz-Mattis theorem in higher dimensions\cite{hast-lsm}.  We use
energy estimates and
additional virtual fluxes $\phi_x=-\theta_x,\phi_y=-\theta_y$ introduced at $x= L/2+1$ and $y=L/2+1$ to show that the
quasi-adiabatic evolution of the ground state along any of the four sides of the square  (for example, from $\theta_x=\theta_y=0$ to $\theta_x=2\pi,\theta_y=0$)
gives a state which is close to the ground state up to a phase. Then, we show that these phases cancel between evolutions along opposite sides of the square.

Lieb-Robinson bounds play a critical role in our proof.  These bounds
were first introduced in \cite{lr}, and extended in \cite{hast-koma,ns}. Nachtergaele and Sims, in \cite{ns}, gave
the important extension to general lattices including exponentially decaying interactions.  The most recent, tightest, and most general estimates, which we use,
are from \cite{loc-estimates}.  The need to use Lieb-Robinson bounds does
currently limit us to lattice Hamiltonians; the extension to fermions moving
in $\mathbb{R}^2$ would require using unbounded interactions and an infinite
dimensional Hilbert space and is currently beyond our techniques, although
work such as \cite{anh} is an important step towards this.

We use $C$ throughout the text to refer to various numeric constants of order unity.  We use the phrase
``for sufficiently large $L$" to indicate that $L$ is large enough that Eq.~(\ref{Lrequire}) is satisfied.
There are three reasons that we have a lower bound on $L$.  The first reason is
that in many places, such as Eq.~(\ref{sets_omega}), we construct subsets of the lattice whose size
is a fraction of $L$; this first reason in fact only requires that $L$ be at least a constant times $R$.
The second reason is that many of the error estimates are a sum of error terms which have different behaviors in
$L$, and we want to require
in each case that the term in the
error estimate which is dominant in the asymptotic ($L\rightarrow \infty$) limit
is larger than the other error terms.  This second reason also only requires that $L$ be at least a constant times $R$.  The third reason, which in fact gives the most
stringent requirement on $L$ and leads to Eq.~(\ref{Lrequire}), is that our choice of $r$ in corollary (\ref{cor:adiabatic_phase_3}) requires an upper bound on $r$, which depends on $L$ as given in Eq.~(\ref{rchoice}) and so requires a lower bound on $\alpha\gamma$, and hence a lower bound on $G_{R,J,\gamma}(L)$.
For simplicity, we summarize all these requirements as ``for sufficientlly large".

While the final estimate (\ref{mainres}) relies on many estimates
along the way, we can give a very brief explanation of where the scaling
with $L^{2/5}$ arises: in the section on Lieb-Robinson bounds for interactions
$S_\alpha^k$, we are able to derive meaningful bounds for $\alpha$ scaling
as $L^{1/5}$ (up to log factors), and then this scaling of $\alpha$
determines the resulting error estimates, since the dominant terms elsewhere
scale exponentially in the square of $\alpha$.
After giving the proof, we give various extensions.  First, we give a corollary, using the Hall conductance as an obstruction
to finding certain paths in parameter space connecting certain Hamiltonians.
We then discuss extensions to the fractional quantum Hall case, in which we need an additional assumption discussed below which
describes  the presence of topological order; in the fractional case, there are multiple ground states, in which case the phases need not exactly cancel for evolution around the large loop. The assumption of topological order deals with this issue, as discussed in the section on the Fractional Hall Effect.

\section{Introducing the twists}
The evolution we are interested in will be generated by Hamiltonians of the form
\begin{equation}\label{def:H_S}
H(\theta_x, \phi_x, \theta_y, \phi_y)=\sum_{Z\subset T} \Phi(Z;\theta_x, \phi_x,\theta_y,\phi_y).
\end{equation}

In order to specify the interaction terms in the Hamiltonian with twists we introduce, for any operator $A$, the following periodic operator twists: 
\begin{align}\label{def:charges}
&R_X (\theta_x, A) = e^{i \theta_x Q_X} A e^{-i \theta_x Q_X} = R_X (\theta_x+2\pi, A), &
&Q_X = \sum_{1\le x(s)\leq L/2} q_s&\\
&R_Y (\theta_y, A) = e^{i \theta_y Q_Y} A e^{-i \theta_y Q_Y} = R_Y (\theta_y+2\pi, A), &
&Q_Y = \sum_{1\le y(s)\leq L/2} q_s.&
\end{align}

Then, the flux-twisted interaction terms $\Phi(Z;\theta_x, \phi_x,\theta_y, \phi_y)$ are defined according to the following prescription:
\begin{enumerate}[{X}-1.]
\item If $\exists s \in Z : |x(s) - 1| < R$, then $\Phi(Z;\theta_x, \phi_x, \theta_y, \phi_y) = R_X(\theta_x, \Phi(Z;0,0,\theta_y,\phi_y))$.
\item If $\exists s \in Z : |x(s) - (L/2+1)| < R$, then $\Phi(Z;\theta_x, \phi_x,\theta_y, \phi_y) = R_X(-\phi_x, \Phi(Z;0,0,\theta_y,\phi_y))$.
\item Otherwise, $\Phi(Z;\theta_x, \phi_x,\theta_y, \phi_y) = \Phi(Z;0,0,\theta_y,\phi_y)$.
\end{enumerate}
Continuing, the terms $\Phi(Z;0,0,\theta_y,\phi_y)$ are defined as follows:
\begin{enumerate}[{Y}-1.]
\item If $\exists s \in Z : |y(s) - 1| < R$, then $\Phi(Z;0,0,\theta_y, \phi_y) = R_Y(\theta_y, \Phi(Z))$.
\item If $\exists s \in Z : |y(s) - (L/2+1)| < R$, then $\Phi(Z;0,0,\theta_y, \phi_y) = R_Y(-\phi_y, \Phi(Z))$.
\item Otherwise, $\Phi(Z;0,0,\theta_y, \phi_y) = \Phi(Z)$.
\end{enumerate}

Note that since the interactions $\Phi(Z)$ commute with the total charge $Q$ and have finite range $R$, the twists only affect $2$ horizontal and $2$ vertical strips, each of width $2R$, centered on the lines $x=1, x = L/2+1$ and $y=1, y= L/2+1$ (the boundaries of the sets on which $Q_X$ and $Q_Y$ are supported), as shown in Fig.~\ref{fig:interactions}.
Moreover, for each $\Phi(Z)$ at most two ``perpendicular" twists act non-trivially to produce the corresponding twisted interaction. In particular, there are $4$ sets of interactions that see both twists simultaneously, namely those that satisfy conditions X-$1$ $\wedge$ Y-$1$, X-$1$ $\wedge$ Y-$2$, X-$2$ $\wedge$ Y-$1$ or X-$2$ $\wedge$ Y-$2$. For example, interactions $\Phi(Z)$ that satisfy conditions X-$1$ $\wedge$ Y-$1$ transform into $R_Y(\theta_y, R_X(\theta_x,\Phi(Z)))$ (the order in which we apply the twists is irrelevant, since $[Q_X, Q_Y] = 0$).

\begin{figure}
\centering
\includegraphics[width=150px]{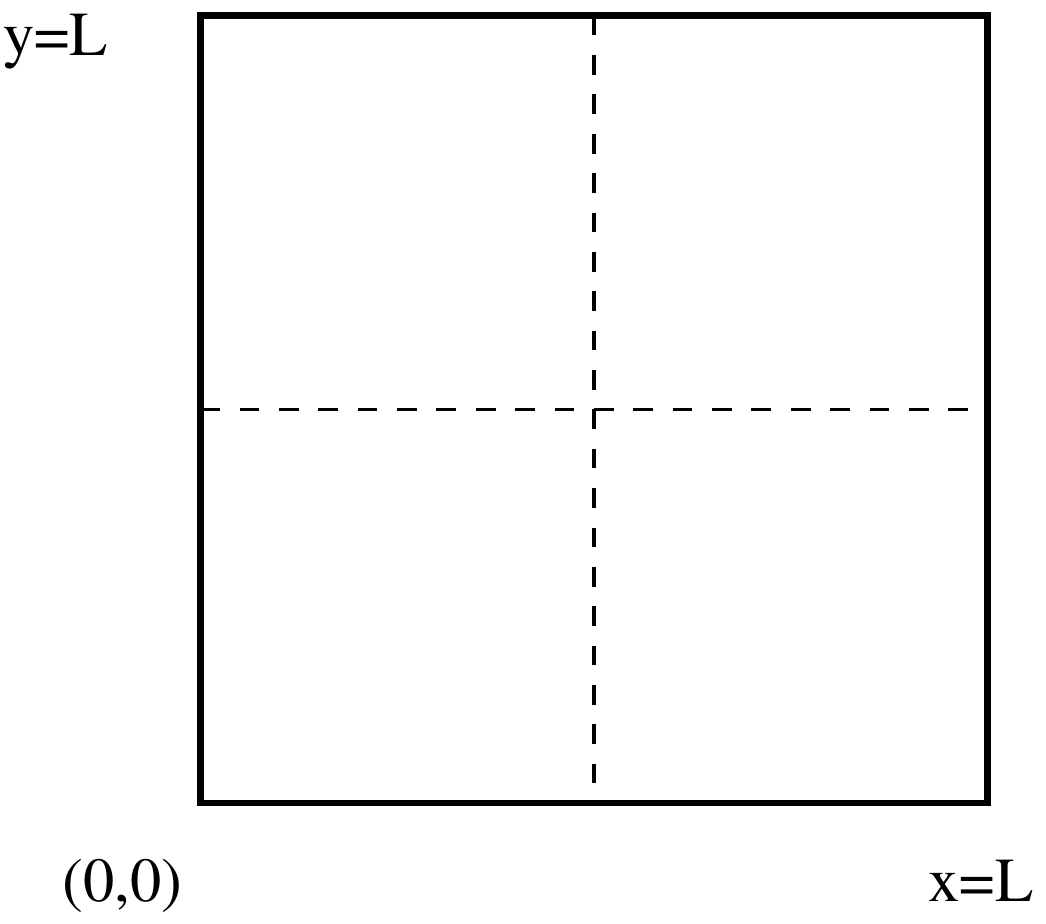}
\caption{{\small{Lines illustrating how the twists are defined on the torus.  The twists $\phi_x,\phi_y$ affect interactions
close to the vertical and horizontal dashed lines, respectively, while the twists $\theta_x,\theta_y$ affect interactions
close to the vertical and horizontal solid lines.
}}}
\label{fig:interactions}
\end{figure}

We will need bounds later on the partial derivatives of the Hamiltonian
with respect to the twists.
There are $L$ sites in $T$ (the lines $x=1$ and $y=1$, respectively) on which the individual terms in $\partial_{\theta} H(\theta,0,\theta_y,\phi_y)|_{\theta=\theta_x}$ and $\partial_{\theta} H(\theta_x,\phi_x,\theta,0)|_{\theta=\theta_y}$ may be grouped to act. Each set of interactions acting on a site $s$ has norm bounded as follows:
\begin{equation}\label{bound:terms}
\sum_{Z\ni s} \|\partial_{\theta} \Phi(Z;\theta,0,\theta_y,\phi_y)|_{\theta=\theta_x}\| = \sum_{Z\ni s} \|[Q_{Z\cap X}, \Phi(Z;\theta_x,0,\theta_y,\phi_y)]\| \le q_{\max} |Z\cap X| \, \sum_{Z\ni s} \|\Phi(Z)\| \le Q_{\max}\,J,
\end{equation}
where $Q_{Z\cap X}$ is the charge in $Z$ that is also contained in $Q_X$, and we defined:
\be
Q_{\max} \equiv  q_{\max} \left(\frac{R^2}{4}\right).
\ee
We also used that $A \ge 0 \implies \|[A,B]\| = \|[A- (\|A\|/2)\one,B]\| \le 2 \|A- (\|A\|/2)\one\| \|B\| = \|A\|\|B\|$ and $|Z\cap X| \le |Z| \le R^2/4$. Similarly, $\sum_{Z \ni s}\|\partial_{\theta} \Phi(Z;\theta_x,\phi_x,\theta,0)|_{\theta=\theta_y}\| \le Q_{\max}\, J$. The previous bounds imply:
\begin{eqnarray}\label{bound:rot-H_X}
\|\partial_{\theta} H(\theta,0,\theta_y,\phi_y)|_{\theta=\theta_x}\| &\le& Q_{\max} J L \\
\|\partial_{\theta} H(\theta_x,\phi_x,\theta,0)|_{\theta=\theta_y}\| &\le& Q_{\max} J L.\label{bound:rot-H_Y}
\end{eqnarray}

Finally, remembering that $\Phi(Z)$ commutes with the total charge $Q$, one may verify using the definitions of twisted interactions, that:
\begin{eqnarray}\label{eq:unitary_equiv_X}
R_X(-\theta, H(\theta_x,\phi_x,\theta_y,\phi_y)) = H(\theta_x-\theta,\phi_x+\theta,\theta_y,\phi_y)\\
R_Y(-\theta, H(\theta_x,\phi_x,\theta_y,\phi_y)) = H(\theta_x,\phi_x,\theta_y-\theta,\phi_y+\theta)\label{eq:unitary_equiv_Y}
\end{eqnarray}
which implies that the twist/anti-twist Hamiltonians (i.e. $\phi_x  = -\theta_x$ and $\phi_y = -\theta_y$) are unitarily equivalent to the original gapped Hamiltonian. In particular, we have:
\begin{equation}\label{twist-anti-twist}
R_X(\theta_x, R_Y(\theta_y, H_0)) = H(\theta_x, -\theta_x, \theta_y, -\theta_y).
\end{equation}

\section{The quasi-adiabatic evolution}
An important ingredient of cyclic adiabatic evolution is the persistence of the spectral gap in the Hamiltonians driving the evolution. The
quasi-adiabatic evolution defined below will approximate the
adiabatic evolution so long as the gap is sufficiently large.  Often times, though, we simply don't know how the spectral gap behaves during the evolution, but
using virtual fluxes we will
be able to show that a state quasi-adiabatically continued around certain cycles has energy that is very close to the ground state energy of the original gapped Hamiltonian. This observation is enough to allow us to deduce that the evolved state is close to the initial non-degenerate ground state, up to a phase factor.
\subsection{Approximating the adiabatic evolution}
In order to approximate the adiabatic evolution of the ground state of a family of gapped Hamiltonians, we first need to write down a differential equation describing the desired evolution.
To do this, we note that if for some $s \ge 0$, we have that $\{H(\theta)\}_{\theta \in [0,s]}$ is a differentiable family of gapped Hamiltonians with a differentiable family of unique (up to a phase) ground states $\ket{\Psi_0(\theta)}$ and ground state energies $E_0(\theta)$, then differentiating $(H(\theta)-E_0(\theta))\,\ket{\Psi_0(\theta)}=0$ yields:
\begin{equation}\label{gs:evol_0}
(\one - P_0(\theta))\partial_{\theta} \ket{\Psi_0(\theta)} = -\frac{\one-P_0(\theta)}{H(\theta)-E_0(\theta)}\, \partial_\theta \, H(\theta) \, \ket{\Psi_0(\theta)},
\end{equation}
where $P_0(\theta) = \pure{\Psi_0(\theta)}$.
Since the phase of $\ket{\Psi_0(\theta)}$ is arbitrary, we may choose
\be
\ket{\Psi'_0(\theta)} = e^{-\int_0^{\theta} \braket{\Psi_0(t)}{\partial_t \Psi_0(t)} dt}\ket{\Psi_0(\theta)}.
\ee
which is also a ground state of $H(\theta)$ and satisfies the {\it parallel transport} condition:
\be\label{parallel-transport}
\braket{\Psi'_0(\theta)}{\partial_{\theta} \Psi'_0(\theta)} =0.
\ee
Combining the above condition with (\ref{gs:evol_0}) we have:
\begin{equation}\label{gs:evol}
\partial_{\theta} \ket{\Psi'_0(\theta)} = -\frac{\one-P_0(\theta)}{H(\theta)-E_0(\theta)}\, \partial_\theta \, H(\theta) \, \ket{\Psi'_0(\theta)}.
\end{equation}
From here on, we will assume that whenever we {\it adiabatically} transport a unique ground state along a path in parameter
space, we do this to satisfy the {\it parallel transport} condition. Moreover, if $H(\theta)$ has a spectral gap $\Delta(\theta) > 0$, then we have $\Delta(\theta) \cdot (\one - P_0(\theta)) \le H(\theta) - E_0(\theta)$ which combined with (\ref{gs:evol}) implies the bound:
\be\label{bound:partial}
\|\partial_{\theta} \ket{\Psi'_0(\theta)}\| \le \frac{\|\partial_\theta H(\theta)\|}{\Delta(\theta)} 
\ee
Moreover, we get after expanding $\partial^2_\theta [(H(\theta)-E_0(\theta)) \ket{\Psi'_0(\theta)}] = 0$ and rearranging terms:
\begin{eqnarray}
(\one-P_0(\theta)) \partial^2_\theta \ket{\Psi'_0(\theta)} &=& 
-\frac{\one-P_0(\theta)}{H(\theta)-E_0(\theta)}\, [2\partial_\theta \, H(\theta) \, (\partial_\theta \ket{\Psi'_0(\theta)})+ (\partial^2_\theta H(\theta)) \ket{\Psi'_0(\theta)}] \implies \nonumber\\
\partial^2_\theta \ket{\Psi'_0(\theta)} + (\braket{\partial_\theta \Psi'_0(\theta)}{\partial_\theta\Psi'_0(\theta)})\ket{\Psi'_0(\theta)} &=&
-\frac{\one-P_0(\theta)}{H(\theta)-E_0(\theta)}\, [2\partial_\theta \, H(\theta) \, (\partial_\theta \ket{\Psi'_0(\theta)}) + (\partial^2_\theta H(\theta)) \ket{\Psi'_0(\theta)}]\implies \nonumber\\
\|\partial^2_\theta \Psi'_0(\theta)\| &\le& \|\partial_\theta \Psi'_0(\theta)\|^2 + \frac{2\|\partial_\theta \, H(\theta)\| \cdot \|\partial_\theta \Psi'_0(\theta)\| + \|\partial^2_\theta H(\theta)\|}{\Delta(\theta)} \nonumber\\
&\le& \frac{3\|\partial_\theta \, H(\theta)\|^2}{\Delta^2(\theta)} + \frac{\|\partial^2_\theta H(\theta)\|}{\Delta(\theta)}\label{bound:partial_2}
\end{eqnarray}
where we used condition (\ref{parallel-transport}) to get $\braket{\Psi'_0(\theta)}{\partial^2_\theta\Psi'_0(\theta)}= -\braket{\partial_\theta \Psi'_0(\theta)}{\partial_\theta\Psi'_0(\theta)}$ from $\partial_\theta\braket{\Psi'_0(\theta)}{\partial_\theta\Psi'_0(\theta)} = 0$, for the second line and used (\ref{bound:partial}) to get the final estimate.

It will be useful to have an estimate on the gap $\Delta(\theta)$ as $\theta$ changes from $0$. 
\begin{lemma}\label{lem:gap-estimate}
If for some $s \ge 0$, we have that $\{H(\theta)\}_{\theta \in [0,s]}$ is a differentiable path of Hamiltonians each with spectral gap $\Delta(\theta)$ and ground state $\ket{\Psi_0(\theta)}$, then for $\theta\in[0,s]$:
\be\label{bound:gap-estimate}
\Delta(\theta) \ge \Delta(0) - 2|\theta| \cdot \sup_{r\in[0,\theta]} \|\partial_r H(r)\|
\ee
\begin{proof}
Let $E_0(\theta)$ be the ground state energy of $H(\theta)$.
The spectrum of $H(0)$ is contained in $\{E_0(0)\}\cup [E_0(0)+\Delta(0),\infty)$.  
Therefore, a well-known result in functional analysis implies that the spectrum of $H(\theta)$ is contained in
$[E_0(0)-\|H(\theta)-H(0)\|,E_0(0)+\|H(\theta)-H(0)\|]\cup [E_0(0)+\Delta(0)-\|H(\theta)-H(0)\|,\infty)$.
We have $\| H(\theta)-H(0) \| \leq \int_0^\theta \| \partial_r H(r) \| dr\leq |\theta|\cdot \sup_{r\in[0,\theta]}\|\partial_r H(r)\|$.
Thus, $\Delta(\theta) \ge \Delta(0) - 2|\theta| \cdot \sup_{r\in[0,\theta]} \|\partial_r H(r)\|$.
\end{proof}
\end{lemma}
Therefore, the gap $\Delta(\theta)$ remains non-vanishing as long as
$|\theta| \cdot \sup_{\eta\in[0,\theta]} \|\partial_\eta H(\eta)\| < \Delta(0)/2$.

\subsection{The super-operator $\Sa(H,A)$}
Now, combining ideas from \cite{hast-lsm} and \cite{osborne}, we introduce the following super-operator that will allow us to approximate the adiabatic evolution given by (\ref{gs:evol}):
\begin{equation}
\Sa(H,A)=\int_{-\infty}^{\infty} s_{\alpha}(t) \left[\int_0^t \tau_u^{H}(A) du \right] dt,
\end{equation}
with $\tau_u^{H}(A) = e^{iuH} A e^{-iuH}$ and 
\begin{equation}
s_{\alpha}(t) = \frac{1}{\alpha \sqrt{2\pi}}\, e^{-\frac{t^2}{2\alpha^2}}.
\end{equation} 
The free parameter $\alpha$ can be fine tuned later to give the best approximation possible.
Note that the following useful bound holds:
\be \label{naive_bnd:sa}
\|\Sa(H,A)\| \le 2 \int_0^{\infty} s_{\alpha}(t) \left(\int_0^t \left\|\tau_{u}^H(A)\right\| du\right) dt
= 2 \|A\|\, \int_{0}^{\infty} s_{\alpha}(t)\,t\, dt = \frac{2\alpha}{\sqrt{2\pi}}\, \|A\| \int_0^{\infty} e^{-s} ds = \frac{2\alpha}{\sqrt{2\pi}}\, \|A\|,
\ee
where we made the substitution $s = t^2/(2\alpha^2)$ in the previous to last equality.
Note also that for Hermitian $H$ and $A$, the operator $\Sa(H,A)$ is also Hermitian (which will be essential in error estimates later).

We define now the following state as an approximation to the true ground state of $H(\theta)$:
\begin{equation}\label{eq:partial_psi}
\partial_{\theta} \ket{\Phi_{\alpha}(\theta)}= i\, \Sa(H(\theta),\partial_{\theta} H(\theta)) \, \ket{\Phi_{\alpha}(\theta)},\quad \ket{\Phi_{\alpha}(0)} = \ket{\Psi_0(0)},
\end{equation}
where $\ket{\Psi_0(0)}$ is the ground state of $H(0)$, assuming that such a state is unique.
Defining the unitary $U_{\alpha}(\theta)$ by 
\begin{equation}\label{def:unitary_alpha}
\partial_{\theta} U_{\alpha}(\theta) = i\, S_a(H(\theta),\partial_\theta H(\theta)) \, U_{\alpha}(\theta), \quad U_{\alpha}(0) = \one,
\end{equation}
we may also write:
\begin{equation}\label{eq:phi_alpha}
\ket{\Phi_{\alpha}(\theta)}= U_{\alpha}(\theta) \ket{\Psi_0(0)}.
\end{equation}
We say that the state $\ket{\Phi_{\alpha}(\theta)}$ is the {\it quasi-adiabatic} evolution of $\ket{\Psi_{0}(0)}$ along the simple path $0 \rightarrow \theta$.

The following lemma gives precise bounds on the deviation of the above approximation from the true ground state of $H(\theta)$:
\begin{lemma}[Adiabatic approximation]\label{lem:quasi} Assume that for some $s \ge 0$, $\{H(\theta)\}_{\theta \in [0,s]}$ is a differentiable path of gapped Hamiltonians, each with spectral gap $\Delta(\theta) >0$. For $\ket{\Phi_{\alpha}(\theta)}$ defined in (\ref{eq:partial_psi}) and $\ket{\Psi_0(\theta)}$ the ground state of $H(\theta)$ satisfying (\ref{parallel-transport}), the following bound holds for $\theta \in [0,s]$:
\begin{equation}\label{delta-small}
\left\|\ket{\Phi_{\alpha}(\theta)}-\ket{\Psi_0(\theta)}\right\| \le 2\, |\theta| \,\sup_{\eta \in [0,\theta]} \|\partial_{\theta'} H(\theta')|_{\theta'=\eta}\| \, \frac{e^{-\alpha^2\Delta^2(\theta)}}{\Delta(\theta)}.
\end{equation}
Moreover, we have the bound:
\begin{eqnarray}\label{delta-theta}
\|i\Sa(H(0),\partial_\theta H(\theta)_{\theta=0}) \ket{\Psi_0(0)} -\partial_{\theta} \ket{\Psi_0(\theta)}|_{\theta=0}\| \le \|\partial_{\theta} H(\theta)|_{\theta=0}\| \, \frac{e^{-\alpha^2\Delta^2(0)}}{\Delta(0)}.
\end{eqnarray}
\begin{proof}
The above estimates follow from considering the norm of the vector 
$$\ket{\delta_{\alpha}(\theta)}=\ket{\Phi_{\alpha}(\theta)}-\ket{\Psi_0(\theta)},$$
which satisfies $\ket{\delta_{\alpha}(0)}=0$ and
\begin{eqnarray}\label{delta}
\partial_{\theta} \ket{\delta_{\alpha}(\theta)} &=& i \, \Sa(\theta) \ket{\delta_{\alpha}(\theta)} + i \left(\Sa(\theta) - \mathcal{T}(\theta)\right)\,\ket{\Psi_0(\theta)},
\end{eqnarray}
where
\be
\Sa(\theta)=\Sa(H(\theta),\partial_{\theta} H(\theta)), \quad
i\, \mathcal{T}(\theta) = - \frac{\one - P_0(\theta)}{H(\theta)-E_0(\theta)}\, \partial_{\theta} H(\theta).
\ee
and the definition of $i\, \mathcal{T}(\theta)$ follows from (\ref{gs:evol}).

We note that for $\{\ket{\Psi_n(\theta)}\}_{n \ge 0}$ an eigenbasis of $H(\theta)$, we have:
\begin{eqnarray}
\braket{\Psi_n(\theta)}{i\,\Sa(\theta) \Psi_0(\theta)} &=& \left(i\, \int_{-\infty}^{\infty} s_{\alpha}(t) \left( \int_0^t e^{iu(E_n(\theta)-E_0(\theta))} du \right) dt\right) \braket{\Psi_n(\theta)}{\partial_{\theta} H(\theta) \Psi_0(\theta)} \nonumber\\
&=& \left(\int_{-\infty}^{\infty} s_{\alpha}(t) \left(e^{it(E_n(\theta)-E_0(\theta))} -1 \right) dt\right) \frac{\braket{\Psi_n(\theta)}{\partial_{\theta} H(\theta) \Psi_0(\theta)}}{E_n(\theta)-E_0(\theta)}\quad n\ge 1,\\
\braket{\Psi_0(\theta)}{i\, \Sa(\theta) \Psi_0(\theta)} &=& \left(i\, \int_{-\infty}^{\infty} s_{\alpha}(t) \, t \,dt\right) \braket{\Psi_0(\theta)}{\partial_{\theta} H(\theta) \Psi_0(\theta)} = 0,
\end{eqnarray}
where the last equality follows from $s_{\alpha}(t)$ being an even function. Now, setting 
\begin{equation}
\ket{\Psi'(\theta)} = i(\Sa(\theta)-\mathcal{T}(\theta))\ket{\Psi_0(\theta)},
\end{equation}
we have that :
\begin{eqnarray}
\braket{\Psi_n(\theta)}{\Psi'(\theta)} &=& \left(\int_{-\infty}^{\infty} s_{\alpha}(t) \, e^{it(E_n(\theta)-E_0(\theta))} dt\right) \frac{\braket{\Psi_n(\theta)}{\partial_{\theta} H(\theta) \Psi_0(\theta)}}{E_n(\theta)-E_0(\theta)}, \quad n \ge 1,\\
\braket{\Psi_0(\theta)}{\Psi'(\theta)} &=& 0,
\end{eqnarray}
where $E_n(\theta)$ is the eigenvalue of $\ket{\Psi_n(\theta)}$.
Recalling that $E_n(\theta)-E_0(\theta) \ge \Delta(\theta)$, for $n \ge 1$, the above equations and the Fourier transform of the Gaussian $s_{\alpha}(t)$ readily imply:
\begin{equation}\label{psi-small}
\|\ket{\Psi'(\theta)}\| \le \|\partial_{\theta'} H(\theta')|_{\theta'=\theta}\|\, \frac{e^{-\alpha^2\Delta^2(\theta)}}{\Delta(\theta)},
\end{equation}
where we used the identity $\|A \ket{\Psi}\|^2 = \sum_{n \ge 0} |\braket{\Psi_n(\theta)}{A \Psi}|^2$ and the norm inequality $\|A \ket{\Psi}\| \le \|A\|$ for unit vectors $\ket{\Psi}$.
Finally, setting $$\Delta_{\alpha}(\theta) = \pure{\delta_{\alpha}(\theta)},$$ and recalling the definition of $U_{\alpha}(\theta)$ in (\ref{def:unitary_alpha}) we have:
\begin{equation}
\Delta_{\alpha}(\theta) = U_{\alpha}(\theta) \left(\int_0^\theta U^{\dagger}_{\alpha}(t) \left(\ketbra{\Psi'(t)}{\delta_{\alpha}(t)} + \ketbra{\delta_{\alpha}(t)}{\Psi'(t)}\right) U_{\alpha}(t) \, dt \right) U^{\dagger}_{\alpha}(\theta),
\end{equation}
which one may verify by differentiating with respect to $\theta$ and then using (\ref{delta}). Now,
taking the trace on both sides, we get:
\begin{eqnarray}
\|\delta_{\alpha}(\theta)\|^2 &=& \int_0^{\theta} \left(\braket{\delta_{\alpha}(t)}{\Psi'(t)}+\braket{\Psi'(t)}{\delta_{\alpha}(t)} \right) dt\\
&\le& 2 |\theta| \, \sup_{t \in [0,\theta]} \|\delta_{\alpha}(t)\| \, \sup_{t\in[0,\theta]} \|\Psi'(t) \|
\end{eqnarray}
Choosing $t_{\theta} \in [0,\theta]: \|\delta_{\alpha}(t_{\theta})\| = \sup_{t \in [0,\theta]} \|\delta_{\alpha}(t)\|$, the above bound implies:
\begin{equation}
 \|\delta_{\alpha}(t_{\theta})\|^2 \le 2 |t_{\theta}| \, \|\delta_{\alpha}(t_{\theta})\| \, \sup_{t\in[0,t_{\theta}]} \|\Psi'(t) \| \implies \|\delta_{\alpha}(\theta)\| \le \|\delta_{\alpha}(t_{\theta})\| \le 2 |t_{\theta}| \, \sup_{t\in[0,t_{\theta}]} \|\Psi'(t) \|
\end{equation}
Combining the above bound with (\ref{psi-small}), we immediately get (\ref{delta-small}).
Now, note that (\ref{delta}) combined with (\ref{psi-small}) imply (\ref{delta-theta}) as follows:
\begin{equation}
\|\partial_{\theta} \ket{\delta_{\alpha}(\theta)}|_{\theta=0}\| = \|\ket{\Psi'(0)}\| \le \|\partial_{\theta} H(\theta)|_{\theta=0}\| \, \frac{e^{-\alpha^2\Delta^2(0)}}{\Delta(0)},
\end{equation}
where we used (\ref{naive_bnd:sa}) in the estimate above. This completes the proof.
\end{proof}
\end{lemma}
\subsection{Quasi-adiabatic unitaries and loop operators}
Since we will be looking at quasi-adiabatic evolutions corresponding to Hamiltonians $H(\theta_x,\theta_y)$ with domain embedded in $\R^2$ (actually, $(\theta_x,\theta_y) \in 2\pi \times 2\pi$), we introduce here the unitaries describing the quasi-adiabatic evolution of $\ket{\Psi_0}=\ket{\Psi_0(0,0)}$ around a closed loop in flux-space.
We begin by defining below the generators of the quasi-adiabatic dynamics: 
\begin{eqnarray}\label{def:approx_x}
\DX(\theta_x,\theta_y) &=& \Sa(H(\theta_x,\theta_y),\; \partial_{\theta} H(\theta,\theta_y)|_{\theta=\theta_x})\\
\DY(\theta_x,\theta_y) &=& \Sa(H(\theta_x,\theta_y),\; \partial_{\theta} H(\theta_x,\theta)|_{\theta=\theta_y}).\label{def:approx_y}
\end{eqnarray}
Note that for notational purposes we suppress the dependence of these operators on $\alpha$.
We continue with the unitaries corresponding to evolution through changes in the $\theta_x$ and $\theta_y$ fluxes, respectively:
\begin{eqnarray}\label{def:unitary_X}
\partial_{r} U_{X}(\theta_x,\theta_y,r) &=& i \, \DX(\theta_x+r,\theta_y)\, U_{X}(\theta_x,\theta_y,r),\quad U_{X}(\theta_x,\theta_y,0) = \one\\
\partial_{r} U_{Y}(\theta_x,\theta_y,r) &=& i \, \DY(\theta_x,\theta_y+r)\, U_{Y}(\theta_x,\theta_y,r),\quad U_{Y}(\theta_x,\theta_y,0) = \one.\label{def:unitary_Y}
\end{eqnarray}
It will be useful to note here the following composition rule for $U_X(\theta_x,\theta_y,r)$, $U_Y(\theta_x,\theta_y,r)$:
\begin{eqnarray}
U_X(\theta_x+r,\theta_y,s) \, U_X(\theta_x,\theta_y,r) &=& U_X(\theta_x,\theta_y,r+s),\\
U_Y(\theta_x,\theta_y+r,s) \, U_Y(\theta_x,\theta_y,r) &=& U_Y(\theta_x,\theta_y,r+s),
\end{eqnarray}
which is easily verified upon differentiating both sides with respect to $s$ and is equivalent to evolving for ``time" $r+s$ in the $X$ and $Y$ directions, respectively, by evolving first for ``time'' $r$ and then for ``time'' $s$.

Using the above unitaries, we construct the following useful evolution operators:
\begin{eqnarray}
V[(\theta_x,\theta_y)\rightarrow (\phi_x,\phi_y)] &=& U_X(\theta_x,\phi_y,\phi_x-\theta_x)\, U_Y(\theta_x,\theta_y,\phi_y-\theta_y)\label{def:evol_ur}\\
W[(\theta_x,\theta_y)\rightarrow (\phi_x,\phi_y)] &=& U_Y(\phi_x,\theta_y,\phi_y-\theta_y)\, U_X(\theta_x,\theta_y,\phi_x-\theta_x)\label{def:evol_ru}\\
V_{\circlearrowleft}(\theta_x, \theta_y ,r) &=& V^{\dagger}[(\theta_x,\theta_y)\rightarrow(\theta_x+r, \theta_y+r)]\,W[(\theta_x,\theta_y)\rightarrow (\theta_x+r,\theta_y+r)] \label{def:unitary_loop_R}
\end{eqnarray}
Note that $V[(\theta_x,\theta_y),(\phi_x,\phi_y)]$ can be thought of as evolving a state along the path $\Gamma_V:(\theta_x,\theta_y) \rightarrow (\theta_x,\phi_y) \rightarrow (\phi_x,\phi_y)$ in parameter space, while $W[(\theta_x,\theta_y),(\phi_x,\phi_y)]$ would evolve a state along the path $\Gamma_W:(\theta_x,\theta_y) \rightarrow (\phi_x,\theta_y) \rightarrow (\phi_x,\phi_y)$. Finally, $V_{\circlearrowleft}(\theta_x, \theta_y ,r)$ is equivalent to evolving counter-clockwise around a square of side $r$ and origin $(\theta_x,\theta_y)$.

We introduce now the following important family of states:
\begin{eqnarray}\label{def:loop-states}
\ket{\Psi_{\circlearrowleft}(r)} = V_{\circlearrowleft}(0, 0 , r) \ket{\Psi_0(0,0)} = U^{\dagger}_{Y}(0, 0 , r)\,U^{\dagger}_{X}(0, r , r)\, U_{Y}(r , 0 , r)\, U_{X}(0, 0 , r) \ket{\Psi_0(0,0)}
\end{eqnarray}
These states describe the quasi-adiabatic evolution of the initial ground state $\ket{\Psi_0(0,0)}$ around a square of size $r$, starting at the origin in flux-space. 
The estimate (\ref{delta-small}) directly implies the following phase estimate for an adiabatic evolution around a closed path 
\be\label{def:loop}
\Lambda(r) = (0,0)\rightarrow (r,0)\rightarrow(r,r)\rightarrow(0,r)\rightarrow(0,0),
\ee
where the arrows represent straight lines between two points in flux-space: 
\begin{cor}\label{cor:phase-estimate}
Assume that the two-parameter path of differentiable Hamiltonians $\{H(\theta_x,\theta_y)\}_{(\theta_x,\theta_y) \in \Lambda(r)}$ maintains a uniform lower bound $\Delta > 0$ on the spectral gap.
Then, for $\ket{\Psi_{\circlearrowleft}(r)}$ defined in (\ref{def:loop-states}) and $\ket{\Psi_{0}(\theta_x,\theta_y)}$ the ground state of $H(\theta_x,\theta_y)$, the following phase estimate holds:
\be\label{phase-estimate}
\left|\braket{\Psi_{0}(0,0)}{\Psi_{\circlearrowleft}(r)} - e^{i\phi(r)} \right| \le 8\, r \, D_H(r) \, \frac{e^{-\alpha^2\Delta^2}}{\Delta},
\ee
where
\begin{eqnarray}
\phi(r) &=& \int_0^r d\theta_x \, \Im \left\{\braket{\Psi_0(\theta_x,r)}{\partial_{\theta_x} \Psi_0(\theta_x,r)} - \braket{\Psi_0(\theta_x,0)}{\partial_{\theta_x} \Psi_0(\theta_x,0)}\right\} \nonumber\\
&-& \int_0^r d\theta_y \, \Im \left\{\braket{\Psi_0(r,\theta_y)}{\partial_{\theta_y} \Psi_0(r,\theta_y)} - \braket{\Psi_0(0,\theta_y)}{\partial_{\theta_y} \Psi_0(0,\theta_y)}\right\} \label{eq:phase_1}\\
&=& 2 \int_0^{r} d\theta_x \int_0^r d\theta_y \, \Im\left\lbrace \braket{\partial_{\theta_y} \Psi_0(\theta_x,\theta_y)}{\partial_{\theta_x} \Psi_0(\theta_x,\theta_y)}\right\rbrace, \label{eq:phase_2}\\
D_H(r) &=& \sup_{(\theta_x,\theta_y) \in \Lambda(r)}\left\{\|\partial_{\theta} H(\theta,\theta_y)|_{\theta=\theta_x}\|,\|\partial_{\theta} H(\theta_x,\theta)|_{\theta=\theta_y}\|\right\}. \label{def:D_H}
\end{eqnarray}
and $\Im\{\cdot\}$ stands for the imaginary part of a complex number.
\begin{proof}
Define the following four vectors, which correspond to deviations of the quasi-adiabatic evolution at the end of each leg in $\Lambda(r)$:
\begin{eqnarray}
\ket{\delta_\alpha(r,0)} &=& U_{X}(0, 0 , r) \ket{\Psi_0(0,0)} -  e^{i\phi_1(r)}\ket{\Psi_0(r,0)} ,\quad i\phi_1(r) = -\int_0^{r} \braket{\Psi_0(\theta_x,0)}{\partial_{\theta_x} \Psi_0(\theta_x,0)} d\theta_x\\
\ket{\delta_\alpha(r,r)} &=& U_{Y}(r, 0 , r) \ket{\Psi_0(r,0)} -  e^{i\phi_2(r)}\ket{\Psi_0(r,r)} ,\quad i\phi_2(r) = -\int_0^{r} \braket{\Psi_0(r,\theta_y)}{\partial_{\theta_y} \Psi_0(r,\theta_y)} d\theta_y\\
\ket{\delta_\alpha(0,r)} &=& U^{\dagger}_{X}(0, r , r) \ket{\Psi_0(r,r)} -  e^{i\phi_3(r)}\ket{\Psi_0(0,r)} ,\quad i\phi_3(r) = \int_0^{r} \braket{\Psi_0(\theta_x,r)}{\partial_{\theta_x} \Psi_0(\theta_x,r)} d\theta_x\\
\ket{\delta_\alpha(0,0)} &=& U^{\dagger}_{Y}(0, 0 , r) \ket{\Psi_0(0,r)} -  e^{i\phi_4(r)}\ket{\Psi_0(0,0)} ,\quad i\phi_4(r) = \int_0^{r} \braket{\Psi_0(0,\theta_y)}{\partial_{\theta_y} \Psi_0(0,\theta_y)} d\theta_y.
\end{eqnarray}
Then, we have the following decomposition:
\begin{eqnarray}
\ket{\Psi_{\circlearrowleft}(r)} 
&=& e^{i\phi_1(r)+i\phi_2(r)+i\phi_3(r)+i\phi_4(r)} \ket{\Psi_0(0,0)} + e^{i\phi_1(r)+i\phi_2(r)+i\phi_3(r)} \ket{\delta_\alpha(0,0)} \nonumber\\
&+& e^{i\phi_1(r)+i\phi_2(r)} U^{\dagger}_{Y}(0, 0 , r) \ket{\delta_\alpha(0,r)} + e^{i\phi_1(r)} U^{\dagger}_{Y}(0, 0 , r)\,U^{\dagger}_{X}(0, r , r)\,\ket{\delta_\alpha(r,r)} \nonumber\\
&+& U^{\dagger}_{Y}(0, 0 , r)\,U^{\dagger}_{X}(0, r , r)\,U_{Y}(r, 0 , r)\,\ket{\delta_\alpha(r,0)},
\end{eqnarray}
which implies:
\be
\left|\braket{\Psi_0(0,0)}{\Psi_{\circlearrowleft}(r)} - e^{i\phi_1(r)+i\phi_2(r)+i\phi_3(r)+i\phi_4(r)}\right| \le \|\delta_\alpha(0,0)\|+\|\delta_\alpha(0,r)\|+\|\delta_\alpha(r,r)\|+\|\delta_\alpha(r,0)\|.
\ee
Now, from Stokes' Theorem we have:
\be
i\,\phi(r) = i\phi_1(r)+i\phi_2(r)+i\phi_3(r)+i\phi_4(r) = 2i\int_0^{r} d\theta_x \int_0^r d\theta_y \, \Im\left\lbrace \braket{\partial_{\theta_y} \Psi_0(\theta_x,\theta_y)}{\partial_{\theta_x} \Psi_0(\theta_x,\theta_y)}\right\rbrace,
\ee
so it remains to show that the vectors $\delta_a$ have small norm. But, this follows immediately from (\ref{delta-small}) and the fact that the phases $\{\phi_i(r)\}_{1\le i \le 4}$ are chosen so that the respective ground states satisfy the parallel transport condition (\ref{parallel-transport}). In particular,  applying (\ref{delta-small}) successively to the gapped Hamiltonians $H(r,0),\, H(r,r), H(0,r)$ and $H(0,0)$, we have the bounds:
\begin{eqnarray}
\|\delta_\alpha(r,0)\| &\le& 2\, r \,\sup_{\eta \in [0,r]} \|\partial_{\theta} H(\theta,0)|_{\theta=\eta}\| \, \frac{e^{-\alpha^2\Delta^2}}{\Delta}\\
\|\delta_\alpha(r,r)\| &\le& 2\, r \,\sup_{\eta \in [0,r]} \|\partial_{\theta} H(r,\theta)|_{\theta=\eta}\| \, \frac{e^{-\alpha^2\Delta^2}}{\Delta}\\
\|\delta_\alpha(r,0)\| &\le& 2\, r \,\sup_{\eta \in [0,r]} \|\partial_{\theta} H(\theta,r)|_{\theta=\eta}\| \, \frac{e^{-\alpha^2\Delta^2}}{\Delta}\\
\|\delta_\alpha(0,0)\| &\le& 2\, r \,\sup_{\eta \in [0,r]} \|\partial_{\theta} H(0,\theta)|_{\theta=\eta}\| \, \frac{e^{-\alpha^2\Delta^2}}{\Delta}.
\end{eqnarray}
which completes the proof.
\end{proof}
\end{cor}
Note that the quantity $\Im \left \lbrace \braket{\partial_{\theta_y} \Psi_0(\theta_x,\theta_y)}{\partial_{\theta_x} \Psi_0(\theta_x,\theta_y)}\right \rbrace$ is gauge-invariant; in particular, it remains constant under phase changes $\ket{\Psi'_0(\theta_x,\theta_y)} = e^{i\,f(\theta_x,\theta_y)} \ket{\Psi_0(\theta_x,\theta_y)}$, with $f(\theta_x,\theta_y)$ any real, differentiable function of $\theta_x$ and $\theta_y$.

We now take a closer look at the phase $\phi(r)$ accumulated during the adiabatic evolution of $\ket{\Psi_0(0,0)}$ along the closed path $\Lambda(r)$ defined in (\ref{def:loop}).
\begin{cor}\label{cor:phi(r)}
Let $\phi(r)$ be the phase defined in (\ref{eq:phase_1}-\ref{eq:phase_2}). Then, the following bound holds for $r > 0$:
\begin{eqnarray}
\label{bound:phase_partials}
&&\left|\phi(r)/r^2 - 2 \Im \left\lbrace\braket{\partial_{\theta_y}\Psi_0(0,\theta_y)_{\theta_y=0}}{\partial_{\theta_x}\Psi_0(\theta_x,0)_{\theta_x=0}}\right\rbrace\right| \nonumber \\
&\le& 
\left[\sup_{s_3\in[0,r]} \left(\frac{1}{2}\|[\partial_s g(s,s)]_{s=s_3}\| + \frac{r}{4!}\cdot \sup_{s_4\in[0,r]}\|\partial^2_s [g(s,s_4)+g(s_4,s)]_{s=s_3}\|\right)\right] \cdot r,
\end{eqnarray}
with $g(s,s') = 2\Im \left\lbrace \braket{\partial_{s''} \Psi_0(s,s'')_{s''=s'}}{\partial_{s''} \Psi_0(s'',s')_{s''=s}}\right\rbrace$.
\begin{proof}
The above bound follows from the Taylor expansion of $\phi(r)$ up to order $3$. We have from (\ref{eq:phase_2}) that $\phi(0) = \partial_r \phi(r)|_{r=0} = 0$ and $\partial^2_r \phi(r)|_{r=0} = 4 \Im \left\lbrace\braket{\partial_{\theta_y}\Psi_0(0,\theta_y)}{\partial_{\theta_x}\Psi_0(\theta_x,0)}_{\theta_x=\theta_y=0}\right\rbrace$. This follows from the Taylor expansion of the general integral operator $K(r) = \int_0^r d\theta_x \int_0^r d\theta_y \, g(\theta_x,\theta_y)$, where $g(\theta_x,\theta_y)$ is doubly differentiable for $\theta_x, \theta_y \in [0,r]$. In particular, we have 
\be
\left[\partial^3_s K(s)\right]_{s=s_3} = 3 [\partial_s g(s,s)]_{s=s_3} + \int_0^{s_3} ds_4 \, [\partial^2_s (g(s,s_4)+g(s_4,s))]_{s=s_3}
\ee
which follows from higher order partials of $\left(\partial_s K(s)\right)_{s=s_1} =  \int_0^{s_1} ds [g(s,s_1)+g(s_1,s)]$. Then, expanding around $r=0$ up to third order, we get:
\be
K(r) = g(0,0)\, r^2 + \int_0^r ds_1\, \int_0^{s_1} ds_2\, \int_0^{s_2} ds_3 \, \left[3 (\partial_s g(s,s))_{s=s_3} + \int_0^{s_3} ds_4 \, [\partial^2_s (g(s,s_4)+g(s_4,s))]_{s=s_3}\right].
\ee
The above expansion implies
\be
|K(r)/r^2 - g(0,0)| \le \left[\sup_{s_3\in[0,r]} \left(\frac{1}{2}\|[\partial_s g(s,s)]_{s=s_3}\| + \frac{r}{4!}\cdot \sup_{s_4\in[0,r]}\|\partial^2_s [g(s,s_4)+g(s_4,s)]_{s=s_3}\|\right)\right] \cdot r.
\ee
In our case, $g(s,s') = 2\Im \left\lbrace \braket{\partial_{s''} \Psi_0(s,s'')_{s''=s'}}{\partial_{s''} \Psi_0(s'',s')_{s''=s}}\right\rbrace$.
\end{proof}
\end{cor}
\section{Localizing the operator $\Sa(H,A)$}
We turn our focus in showing that the operator $\Sa(H, A_Z)$, with $H = \sum_{X\subset T} \Phi(X)$ and $A_Z \in \A_Z$, is effectively localized around $Z$; the only requirement for the interactions $\Phi(X)$ of the Hamiltonian $H$ is that they have finite range $R$ and finite strength $J$. For $M \ge R$, let us write:
\begin{equation}\label{def:H_M}
H_M(Z) = \sum_{X\subset Z(M)} \Phi(X),\quad \mbox{with}\quad Z(M) = \{ X\subset T: d(X,Z) < M-R\},
\end{equation}
for the interactions restricted to within a radius $M$ of the set $Z$. Then, for any operator $A_Z$ with support on $Z\subset T$, we define:
\begin{equation}\label{defn:evolution_dif}
\Delta_k^{u}(H,A_Z) = \tau_u^{H_k(Z)}(A_Z) - \tau_u^{H_{k-1}(Z)}(A_Z), \quad \mbox{for } k > R, \quad
\Delta_R^{u}(H,A_Z) = A_Z.
\end{equation}
Note that the action of $\Delta_k^{u}(H,A_Z)$ depends on the support $Z$ of its second argument, but we suppress this dependence for notational purposes.
It follows that for $R \le M \le L$ we have:
$$\tau_u^{H_M(Z)}(A_Z) = \sum_{k = R}^M \Delta_k^{u}(H,A_Z).$$
From the above decomposition, we have for $A_Z \in \A_Z$:
\begin{equation}\label{sa_decomposition}
\Sa^{(M)}(H,A_Z) \equiv \Sa(H_M(Z), A_Z) = \sum_{k = R}^M \Sa^k(H,A_Z), \quad \Sa^k(H,A_Z) = \int_{-\infty}^{\infty}s_{\alpha}(t) \left[\int_0^t \Delta_k^{u}(H,A_Z) du \right] dt
\end{equation}
It will also be useful to define the following operator here for $A = \sum_{Z\subset T} A_Z$:
\begin{equation}\label{sa_decomposition_2}
\Sa^{(M)}(H,A) \equiv \sum_{Z\subset T} \Sa^{(M)}(H,A_Z).
\end{equation}
Now, for any $A_Z\in \A_Z$ with $\diam(Z) \le R$ and $\Delta_M(Z) = H_{M+1}(Z)-H_{M}(Z)$, the following Lieb-Robinson bound holds:
\begin{eqnarray}
\left\|\left[\tau_u^{H_{M+1}(Z)}(A_Z),\Delta_M(Z)\right]\right\| \le
2 |Z| \cdot \|\Delta_{M}(Z)\| \|A_Z\| \, e^{v_a\,|u|- a \cdot d(A_Z, \Delta_M(Z))} \le C_{R,J} \|A_Z\| \, M e^{v|u|- M/R} \label{bound:LR}
\end{eqnarray}
where $C_{R,J} = 2 R^2 J$ and $v = 132e\, (1+R)^{3} J$, in the last bound. The first inequality, known as a Lieb-Robinson bound, and the constants involved are obtained in \cite{loc-estimates}. The second bound follows from the observations below:
\begin{enumerate}
\item The distance between the terms $A_Z$ and $\Delta_M(Z)$ is calculated using the definition of $H_M(Z)$ in (\ref{def:H_M}): 
\begin{equation*}
d(A_Z, \Delta_M(Z)) = d(A_Z, H_{M+1}(Z)-H_{M}(Z)) = M-R.
\end{equation*}
\item For Hamiltonians embedded in the torus $T$, we have the bound:
\be
|Z| \cdot \|H_{M+1}(Z)-H_{M}(Z)\| \le (R/2)^2 \cdot (4JM) \le R^2 J\, M.\nonumber
\ee
\item The Lieb-Robinson velocity $v$ is given by $v_a = 2 C_0(1) \|\Phi\|_{1}^a$, for $a= 1/R$, with $C_0(1)$ defined in (\ref{c_epsilon}) and $\|\Phi\|_{1}^a$ defined in (\ref{def:phi_a}), with $ f(s_1,s_2) = \sum_{Z \ni s_1, s_2} \|\Phi(Z)\| \le J$. We bound $\|\Phi\|_{1}^a$ by choosing $d(s_1,s_2) =R$, since each interaction $\Phi(Z)$ has range at most $R$.
\end{enumerate}
Moreover, using the above Lieb-Robinson bound, we have the estimate:
\begin{eqnarray*}
\left\|\Delta_{M+1}^{u}(H,A_Z)\right\| &=& \left\|\int_0^u \partial_s \left( \tau_{-s}^{H_M(Z)}\left(\tau_{s}^{H_{M+1}(Z)}(A_Z)\right)\right)_{s=t} \, dt \right\|\\
&\le& \int_0^{|u|} \left\|\left[\tau_t^{H_{M+1}(Z)}(A_Z), H_{M+1}(Z)-H_{M}(Z)\right]\right\| dt\\
&\le& C_{R,J} \|A_Z\|\, M\,e^{-M/R} \int_0^{|u|} e^{v|t|} dt\\
&\le&  \frac{1}{2R} \|A_Z\|\, M e^{-M/R + v |u|}.
\end{eqnarray*}
The equality in the first line of the above equation follows because the integral over $t$
is equal to $\tau_{-u}^{H_M(Z)}\left(\tau_{u}^{H_{M+1}(Z)}(A_Z)\right)-A_Z$, which is unitarily equivalent to
$\tau_{u}^{H_{M+1}(Z)}(A_Z)-\tau_{u}^{H_M(Z)}(A_Z)$.  The inequality in the second line follows from
$\partial_s \left(\tau_{-s}^{H_M(Z)}\left(\tau_{s}^{H_{M+1}(Z)}(A_Z)\right)\right)=
i\tau_{-s}^{H_M(Z)}\left(\left[H_{M+1}(Z)-H_{M}(Z),\tau_{s}^{H_{M+1}}(A_Z)\right]\right)$ and using a triangle inequality.
Of course, for all $M \ge R$ we also have the naive upper bound:
\be
\|\Delta_{M+1}^{u}(H,A_Z)\| \le \left\|\tau_u^{H_{M+1}(Z)}(A_Z)\right\| + \left\|\tau_u^{H_{M}(Z)}(A_Z)\right\| \le 2\|A_Z\|, 
\ee
which is useful for larger times $u$.

We now estimate the error we make by using $\Sa^{(M)}(H,A_Z)$, defined in (\ref{sa_decomposition}), to effectively localize the operator $\Sa(H,A_Z)$ to within a radius $M$ of the support of $A_Z$.
\begin{lemma}[Localization of $\Sa(H,A)$]\label{lem:sa_approx}
For any operator $A_Z \in \A_Z$, with $\diam(Z) \le R$, and $H$ a Hamiltonian with finite range $R$  and finite strength $J$ interactions, the following bound holds:
\begin{equation}
\|\Sa(H,A_Z) - \Sa^{(M)}(H,A_Z)\| \le \sum_{N \ge M} \|\Sa^{N+1}(H,A_Z)\| \le \frac{2\alpha}{\sqrt{2\pi}}\, \|A_Z\| \,  g_\alpha(M/R),
\end{equation}
where for $\sigma \equiv 2\alpha v \ge \sqrt{2\pi}$:
\[g_\alpha(x) = \left\{ 
\begin{array}{l l}
 8R \sqrt{2\pi}\, \sigma\, e^{- \frac{x^2}{2\sigma^2}} & \quad \mbox{for  } 1 \le x < \sigma^2 \bigskip\\
 \frac{32 R}{\sigma^2}\, x\, e^{- \frac{x}{2}} & \quad \mbox{for  } x \ge \sigma^2
 \end{array} \right. \]
\begin{proof}
We begin by bounding the norm of $\Sa^{N+1}(H,A_Z)$, for $N \ge M$, as defined in (\ref{sa_decomposition}), since 
\begin{equation}
\Sa(H,A_Z) - \Sa^{(M)}(H,A_Z) = \sum_{N\ge M}^L \Sa^{N+1}(H,A_Z).
\end{equation}
Using the estimates on $\|\Delta_{N+1}^u(H,A_Z)\|$ derived above, and noting that $s_{\alpha}(t) \le \frac{1}{\alpha \sqrt{2\pi}}, \, \forall \, t \in \R$, we have:
\begin{eqnarray}
\| \Sa^{N+1}(H,A_Z)\| &\le& 2 \int_0^{\infty} s_{\alpha}(t) \left(\int_0^t \|\Delta_{N+1}^{u}(H,A_Z)\| du\right) dt\\
&\le& 2 (2\|A_Z\|)\, \int_{T_N}^{\infty} s_{\alpha}(t)\,t\, dt  + 2  \frac{1}{2R}\,\|A_Z\| N e^{- N/R} \int_0^{T_N} \frac{1}{\alpha \sqrt{2\pi}} \left(\int_0^t  e^{v |u|} du\right) dt\\
&\le& \frac{4\alpha}{\sqrt{2\pi}}\, \|A_Z\|\,\int_{(T_N/\alpha)^2/2}^{\infty} e^{-t} \, dt  + 
\frac{2}{\sqrt{2\pi}} \|A_Z\|  \frac{1}{2R} \frac{N}{\alpha v} e^{- N/R} \int_0^{T_N} e^{v t} dt\\
&=& \frac{4\alpha}{\sqrt{2\pi}}\, \|A_Z\|\, e^{-(T_N/\alpha)^2/2}+ 
\frac{2\alpha}{\sqrt{2\pi}}\|A_Z\|  \frac{1}{2R} \frac{N}{\alpha^2 v^2} e^{- N/R+vT_N}\\
&=& \frac{4\alpha}{\sqrt{2\pi}}\, \|A_Z\| \left(1 +  \frac{1}{4R} \cdot \frac{N}{\alpha^2 v^2} \right) e^{- \epsilon(N/R) \frac{N}{R} }\\
&\le& \frac{4\alpha}{\sqrt{2\pi}}\, \|A_Z\| \left(1 + \sigma^{-2} \frac{N}{R} \right) e^{- \epsilon(N/R) \frac{N}{R}} \label{exp_decay}
\end{eqnarray}
where we set
\begin{equation}\label{def:epsilon}
\sigma = 2\alpha v,\quad  \epsilon(x) = 1 - \frac{2}{1 + \sqrt{1 + 8 \frac{x}{\sigma^2}}}, 
\end{equation}
after choosing $T_N$ so that $(T_N/\alpha)^2/2 + v T_N - N/R =0$.
We note here that $\epsilon(x)$ is a positive, increasing function of $x \ge 0$, with the following properties:
\begin{equation}\label{bound:alpha_small}
x \ge \sigma^2 \implies \epsilon(x) \ge 1/2.
\end{equation}
Moreover,
\begin{equation}\label{bound:alpha_large}
x \le \sigma^2 \implies \epsilon(x) = \frac{8 \frac{x}{\sigma^2}}{\left(1 + \sqrt{1 + 8 \frac{x}{\sigma^2}}\right)^2} \ge \frac{x}{2\sigma^2}.
\end{equation}
Using the above observations, (\ref{exp_decay}) implies for $M \ge \sigma^2 R \ge 2\pi R$:
\begin{eqnarray}
\sum_{N \ge M} \left\|\Sa^{N+1}(H,A_Z) \right\| \le \frac{4\alpha}{\sqrt{2\pi}}\, \|A_Z\| \sum_{N \ge M} \frac{2}{\sigma^2}\frac{N}{R} e^{- \frac{N}{2R}}
\le  \frac{16\alpha}{\sqrt{2\pi}}\, \|A_Z\| \, \frac{4R}{\sigma^2} \left(\frac{M}{R} \right) e^{- \frac{M}{2R}},
\label{sa_approx_0}
\end{eqnarray}
where, noting that $N e^{- \frac{N}{2R}}$ is decreasing for $N \ge 2R$, we used a integral approximation for the following sum:
\begin{eqnarray*}
&&\sum_{N \ge M} N e^{- \frac{N}{2R}} = e^{- \frac{M}{2R}} \left(M \sum_{N \ge 0} e^{- \frac{N}{2R}} + \sum_{N \ge 0} N e^{- \frac{N}{2R}}\right)
\le \frac{M}{1-e^{-\frac{1}{2R}}} e^{- \frac{M}{2R}} + \left(-\partial_s \int_0^\infty e^{-s x} \,dx\right)_{s= \frac{1}{2R}} e^{- \frac{M}{2R}}\\
&&\le 4\left(MR+R^2\right) e^{- \frac{M}{2R}} \le 8R^2 \left(\frac{M}{R}\right) e^{- \frac{M}{2R}} .
\end{eqnarray*}
To get the last inequality we used the bound $1- e^{-x} \ge x/2$ for $0 \le x \le \ln 2$, noting that $R \ge 1 \implies 1/(2R) < \ln 2$.
For $R \le M < \sigma^2 R$:
\begin{eqnarray}
\sum_{N \ge M} \left\|\Sa^{N+1}(H,A_Z) \right\| &\le& \frac{4\alpha}{\sqrt{2\pi}}\, \|A_Z\| \left(\sum_{N \ge M}^{\sigma^2 R} 2 e^{- \frac{N^2}{2\sigma^2 R^2}} + \sum_{N \ge \sigma^2 R} 2\frac{N}{R \sigma^2} e^{- \frac{N}{2R}}\right)\nonumber\\
&\le&  \frac{4\alpha}{\sqrt{2\pi}}\, \|A_Z\| \, \left[\left(2 + \sqrt{2\pi}\, \sigma R\right) e^{- \frac{M^2}{2\sigma^2 R^2}}+ 16 R \, e^{- \frac{\sigma^2}{2}}\right] \nonumber\\
&\le&  \frac{16\alpha}{\sqrt{2\pi}}\, R \|A_Z\| \, \left(\sqrt{2\pi}\, \sigma \right) e^{- \frac{M^2}{2\sigma^2 R^2}}
\label{sa_approx}
\end{eqnarray}
where we assumed $\sigma \ge \sqrt{2\pi}$ to get the last inequality.
Combining all of the previous observations with the bounds (\ref{sa_approx_0}) and (\ref{sa_approx}), we get the desired behavior of $g_\alpha(x)$, as claimed in the statement of the Lemma.
\end{proof}
\end{lemma}

\subsection{Lieb-Robinson bounds for interactions $\Sa^k(H,A)$}
We are interested in Lieb-Robinson bounds for systems evolving under Hamiltonians with interaction terms
$\Sa^k(H,A_Z)$ defined in (\ref{sa_decomposition}), where $k \ge R$ and $A_Z \in \A_Z$. We note that each term has support on $Z(k)$, the set of radius $k$ centered on $Z$, defined in (\ref{def:H_M}). 
It will be useful to define the set of supports for all non-zero interactions:
\be
\ell = \{Z \subset T: \|A_Z\| \neq 0\}
\ee 
For sites $s_1,s_2 \in T$, setting $d(Z;s_1,s_2) = \max_{i=1,2}d(Z,s_i)$ and $S(M;s_1,s_2)= \{Z\in \ell: d(Z;s_1,s_2)=M\}$, we have:
\be\label{set_implications}
S(M;s_1,s_2) \neq \emptyset \implies M \ge d(s_1,s_2)/2, \quad Z(k) \ni s_1,s_2 \implies k \ge d(Z;s_1,s_2),
\ee
where the first implication follows from $2M \ge d(s_1,Z)+d(s_2,Z)$ and the second one follows from its contrapositive statement.
Finally, we note that for all $s_1,s_2$, the sets $S(M;s_1,s_2)$ and $S(M';s_1,s_2)$ are disjoint for
$M\neq M'$ and
\be\label{ell_decomposition}
\bigcup_{M\ge d(s_1,s_2)/2} S(M;s_1,s_2)=\ell.
\ee
For a quantum system embedded in $T$, the Lieb-Robinson velocity is bounded above by $2 C_0(\epsilon) \|\Phi\|_{\epsilon}^a$ \cite{loc-estimates}, where for $\epsilon \le 1$:
\begin{equation}\label{c_epsilon}
C_0(\epsilon) = \sup_{s_1,s_2 \in T} \sum_{z \in T} \left(\frac{1+d(s_1,s_2)}{(1+d(s_1,z)) (1+d(z,s_2))}\right)^{2+\epsilon} \le 34 + \frac{32}{\epsilon}
\ee
and for $a > 0$:
\be\label{def:phi_a}
\|\Phi\|_{\epsilon}^a = \sup_{s_1,s_2 \in T} (1+d(s_1,s_2))^{2+\epsilon}\, e^{a d(s_1,s_2)} \,  f(s_1,s_2),
\ee
with
\begin{eqnarray}
 f(s_1,s_2) &\equiv& \sum_{Z\in \ell} \min \left\{ \|\Sa(H,A_Z)\|, \left\|\sum_{k\ge d(Z;s_1,s_2)} \Sa^{k}(H,A_Z)\right\| \right\}\\
 &=& \sum_{M \ge d(s_1,s_2)/2} \left(\sum_{Z\in S(M;s_1,s_2)}\min \left\{ \|\Sa(H,A_Z)\|, \left\|\Sa(H,A_Z)- \Sa^{(M)}(H,A_Z)\right\| \right\}\right),
\end{eqnarray}
where we used (\ref{set_implications}) and (\ref{ell_decomposition}) to get the second line.

The bounds for $C_0(\epsilon)$ are derived as follows:
\begin{eqnarray}
\sum_{z \in T} \left(\frac{1+d(s_1,s_2)}{(1+d(s_1,z)) (1+d(z,s_2))}\right)^{2+\epsilon} 
&\leq& 2+ \sum_{z\neq s_1,s_2} \left(\frac{1+d(s_1,s_2)}{1+d(s_1,s_2) + d(s_1,z)\cdot d(z,s_2)}\right)^{2+\epsilon}\\
&\le& 2+ \sum_{z\neq s_1, s_2} \left(\frac{2}{\min\{d(s_1,z),d(z,s_2)\}}\right)^{2+\epsilon}\\
&\le& 2+ \sum_{m \ge 1} \left(C_2\, m\right)\cdot \left(2/m \right)^{2+\epsilon} \\
&\le& 2+  2^{4+\epsilon} \left(1+ \frac{1}{\epsilon}\right),
\end{eqnarray}
where $C_2=4$ is the surface area of the unit ball in $\Z^{2}$, the first inequality follows from setting $y=\min\{d(s_1,z), d(s_2,z)\}$, $x=\max\{d(s_1,z), d(s_2,z)\}$, $w= d(s_1,s_2)$ and noting that:
\be
(y \le x) \, \wedge\, (w \le 2x) \implies (1+w-2x) (y-2) \le 4x \implies (1+w)y \le 2 (1+w +xy) \implies \frac{1+w}{1+w+xy} \le \frac{2}{y}
\ee
and the last bound is derived from $\sum_{m \ge 1} m^{-(1+\epsilon)} \le 1 + \int_1^\infty x^{-(1+\epsilon)} \,dx$.

Now, in order to bound $\|\Phi\|_{\epsilon}^a$, we estimate an upper bound for $ f(s_1,s_2)$ using (\ref{naive_bnd:sa}) and Lemma \ref{lem:sa_approx} as follows:
\begin{eqnarray}
 f(s_1,s_2) &\le& \frac{2\alpha}{\sqrt{2\pi}}\, \sum_{M \ge d(s_1,s_2)/2} \min \{1, g_\alpha(M/R)\} \left(\sum_{Z\in S(M;s_1,s_2)}  \|A_Z\| \right) \\
&\le& \frac{2\alpha}{\sqrt{2\pi}}\, \left(\sup_{M\ge 0}  \sum_{Z\in S(M;s_1,s_2)}  \|A_Z\| \right) \sum_{M \ge d(s_1,s_2)/2} \min \{1, g_\alpha(M/R)\}\\
&\le& \frac{2\alpha}{\sqrt{2\pi}}\, j_{\max} \cdot s_{\max} \sum_{M \ge d(s_1,s_2)/2} \min \{1, g_\alpha(M/R)\},
\end{eqnarray}
where we set 
\be\label{s_max}
j_{\max} = \left(\sup_{s \in T} \sum_{Z\ni s} \|A_Z\| \right), \quad s_{\max} = \sup_{s_1,s_2 \in T} \left(\sup_{M \ge 0} |S(M;s_1,s_2)|\right),
\ee
with $|S|$ denoting the number of sites in the set $S$ (considered as a union of the sites within its subsets).
We turn our attention to $h_{\alpha}(d) \equiv \sum_{M \ge d/2} \min \{1, g_\alpha(M/R)\}$, which we estimate in three separate cases, using the definition of $g_\alpha(x)$ in Lemma \ref{lem:sa_approx} and the fact that $1 \le g_{\alpha}(x)$ for $x \le \sqrt{2\ln (8 R \sigma^2)} \sigma$. Setting 
\be\label{def:d_0}
d_0 = \sqrt{2\ln (8 \sigma^2 R)} \,(2\sigma R),
\ee
we have:
\[ h_\alpha(d) \le \left\{ 
\begin{array}{l l}
(d_0 - d)/2 + 4\pi R & \quad \mbox{for  } 0 \le d < d_0 \medskip\\
 32\pi \, \sigma^2 R^2\, e^{-\frac{d^2}{2 (2\sigma R)^2}} & \quad \mbox{for  } d_0 \le d < 2\sigma^2 R  \bigskip \\
128 \frac{d R}{\sigma^2} \,e^{-\frac{d}{4R}} & \quad \mbox{for  } d \ge 2\sigma^2 R
 \end{array} \right. \]
where we used integral approximations to bound $h_{\alpha}(d)$, as in the proof of Lemma \ref{lem:sa_approx}.
Putting everything together, we have:
\be
\|\Phi\|_a^\epsilon \le  \frac{2\alpha}{\sqrt{2\pi}}\, j_{\max} \cdot s_{\max} \max_{0 \le d \le L} \left \{(1+d)^{2+\epsilon}\, e^{a d}\, h_\alpha(d)\right\}
\ee
Differentiating with respect to $d$ the expression in brackets after choosing the appropriate bound for $h_\alpha(d)$, we obtain the following uniform upper bound for the given value of the parameter $a$:
\[ \max_{0 \le d \le L} \left\{(1+d)^{2+\epsilon}\, e^{a d}\, h_\alpha(d)\right\} \le
\begin{array}{l l} d^{3+\epsilon}_{0} & \quad \mbox{with } a =\frac{1}{d_0}. \end{array} \]

Finally, setting $\epsilon = (\ln d_0)^{-1}$, we get:
\begin{eqnarray}
\|\Phi\|_a^\epsilon &\le& \frac{2e}{\sqrt{2\pi}}\, j_{\max}\cdot s_{\max} \,\alpha\, d^{3}_{0}.
\end{eqnarray}
Applying the bound in \cite{loc-estimates} to a unitary $U(\theta)$ satisfying:
\be
\partial_\theta' \, U(\theta')_{\theta=\theta'} = i \sum_{Z \in \ell} \Sa(H(\theta), \partial_\theta \Phi(Z;\theta))\, U(\theta), \quad U(0)=\one,
\ee
with $j_{\max} \le Q_{\max}\, J$, using (\ref{bound:terms}), and $s_{\max} \le 16 R^2$, from the geometry of the set $\ell$, we have shown that for $a = d_0^{-1}$ and $\epsilon = (\ln d_0)^{-1}$, with $d_0$ given in (\ref{def:d_0}):
\be
2 C_{0}(\epsilon) \|\Phi\|_a^{\epsilon} \le 4 Q_{\max} R^{-2} d_0^4 \sqrt{\ln d_0},
\ee
where we used $2C_0(\epsilon) \le 132/\epsilon$, for $\epsilon \le 1$, to get 
$$2 C_{0}(\epsilon) \left(\frac{2e}{\sqrt{2\pi}}\, j_{\max}\cdot s_{\max} \,\alpha \right) \le 4 Q_{\max} \,R^{-2}  d_0 \sqrt{\ln d_0},$$
recalling that $v = 132e\, (1+R)^3 J$ and noting that $\frac{\sqrt{2\pi}}{2} \sqrt{\ln d_0} \le \sqrt{2\ln(8\sigma^2 R)}$.
Finally, for  $A \in \A_X$, $B \in \A_Y$ and $X,Y \subset T$, we get:
\be\label{lr-bound:adiabatic}
\|[U(\theta) \,A \,U^{\dagger}(\theta), B]\| \le 2 \|A\|\|B\| \min \{|X|,|Y|\} \, e^{- \frac{d(X,Y)}{d_0} + v(Q_{\max}, d_0) |\theta|},
\ee
where $d_0$ is given in (\ref{def:d_0}) and we define:
\be \label{velocity_bound}
\quad v(Q_{\max}, d_0) \equiv 4 Q_{\max} R^{-2} d_0^4 \sqrt{\ln d_0}.
\ee
\subsection{Application to evolutions based on $U_X(0,\theta_y,\theta_x)$.}
We end this section with an important lemma relating localized versions of the unitaries $U_X(0,\theta_y,\theta_x)$ and $U_X(0,0,\theta_x)$ defined in (\ref{def:unitary_X}) with $H(\theta_x,\theta_y) = H(\theta_x,0,\theta_y,0)$. First, we introduce the unitaries $U_\Omega(\theta_x,\theta_y,\theta)$ satisfying $\forall \,\theta_x , \theta_y \in [0,2\pi]$:
\be\label{unitary_omega}
\partial_{\theta'} U_\Omega(\theta_x,\theta_y,\theta')_{\theta'=\theta} = -i\, \Sa^{(L/24)}(H(\theta_x-\theta,\theta_y), \partial_{\theta'} H_\Omega(\theta',\theta_y)_{\theta'=\theta_x-\theta})\,U_\Omega(\theta_x,\theta_y,\theta), \quad U_\Omega(\theta_x,\theta_y,0) =\one,
\ee
with $H_\Omega(\theta_x,\theta_y)= \sum_{Z\subset \Omega} \Phi(Z;\theta_x,0,\theta_y,0)$ and
\be\label{set_omega}
\Omega = \left\{s \in T: |y(s)| \le \frac{5}{24}L-R \right\}.
\ee
The following Hamiltonian decomposition will be useful:
\be\label{hamiltonian_decomposition}
H(\theta_x,\theta_y) = H_\Omega(\theta_x,\theta_y) + H_{\Omega^c}(\theta_x,\theta_y), \quad H_{\Omega^c}(\theta_x,\theta_y) = \sum_{Z\not \subset \Omega} \Phi(Z;\theta_x,0,\theta_y,0).
\ee
Finally, since the unitary $U_\Omega(\theta_x,\theta_y,\theta)$ acts trivially outside of the set $\Omega_Y$ defined in (\ref{sets_omega}) below, the following is true:
\be\label{rotation_omega}
R_Y(\theta_y, U_\Omega(\theta_x,0,\theta)) = U_\Omega(\theta_x,\theta_y,\theta)
\ee
The Lieb-Robinson bound we have obtained on the velocity of propagation for interactions of the form $\Sa^k(H,A_Z)$, has a direct application to estimating the following important upper bounds:
\begin{lemma}[Twisting Lemma]\label{lem:local_approximation}
For $U_X(0,\theta_y,\theta_x)$, $U_X(0,0,\theta_x)$ defined in (\ref{def:unitary_X}) with $H(\theta_x,\theta_y) = H(\theta_x,0,\theta_y,0)$ and $U_\Omega(\theta_x,\theta_y,\theta)$ defined in (\ref{unitary_omega}), we have for sufficiently large $L$ and $\alpha = \frac{1}{4vR} \left(\frac{L \xi}{48 \ln^3 L}\right)^{\frac{1}{5}}$, with $\xi = \frac{(R/2)^2}{2\pi Q_{\max}}$:
\begin{eqnarray}
\|U^{\dagger}_X(0,\theta_y,\theta_x) \,A_{\Omega_0}\, U_X(0,\theta_y,\theta_x) - U_\Omega(\theta_x,\theta_y,\theta_x) \,A_{\Omega_0}\, U^{\dagger}_\Omega(\theta_x,\theta_y,\theta_x)\| \le |\theta_x| \frac{2\alpha}{\sqrt{2\pi}} Q_{\max}\, J \|A_{\Omega_0}\|\, L^3 e^{- \left(\frac{L}{48 R}\right)^{4/5}} \\
\|U^{\dagger}_X(0,0,\theta_x) \,A_{\Omega_0}\, U_X(0,0,\theta_x)-U_\Omega(\theta_x,0,\theta_x) \,A_{\Omega_0}\, U^{\dagger}_\Omega(\theta_x,0,\theta_x)\| \le |\theta_x| \frac{2\alpha}{\sqrt{2\pi}} Q_{\max}\, J \|A_{\Omega_0}\|\, L^3 e^{- \left(\frac{L}{48 R}\right)^{4/5}},
\end{eqnarray}
for all $\theta_x,\theta_y \in [0,2\pi]$ and $A_{\Omega_0} \in \A_{\Omega_0}$ with 
\be\label{omega_0}
\Omega_0 = \left\{s \in T: |x(s)| \le L/8-R \mbox{ and } |y(s)| \le L/8-R\right\}.
\ee
Moreover, for all $\theta_x,\theta_y \in [0,2\pi]$ the evolved operator $U_\Omega(\theta_x,\theta_y,\theta_x) \,A_{\Omega_0}\, U^{\dagger}_\Omega(\theta_x,\theta_y,\theta_x)$ has support strictly within the set $\Omega_X \cap \Omega_Y$, where:
\be \label{sets_omega}
{\Omega_X} = \{s \in T: |x(s)| \le L/4 \}, \quad \Omega_Y = \{s \in T: |y(s)| \le L/4 \}.
\ee
\begin{proof}
We only show the first bound, since the second one follows from a similar argument.
First, note that $U_\Omega(\theta_x,\theta_y,\theta)$ satisfies:
\begin{eqnarray*}
\partial_{\theta'} U_\Omega(\theta_x,\theta_y,\theta_x-\theta')_{\theta'=\theta} = -\partial_{\theta''} U_\Omega(\theta_x,\theta_y,\theta'')_{\theta''=\theta_x-\theta}
= i\, \Sa^{(L/24)}(H(\theta,\theta_y), \partial_{\theta'} H_\Omega(\theta',\theta_y)_{\theta'=\theta})\,U_\Omega(\theta_x,\theta_y,\theta_x-\theta),
\end{eqnarray*}
where we made the substitution $\theta'' = \theta_x-\theta'$ to get the first equality and used (\ref{unitary_omega}) to get the second equality.
Set
$$\Delta_U(\theta_x,\theta_y) \equiv U^{\dagger}_X(0,\theta_y,\theta_x) \,A_{\Omega_0}\, U_X(0,\theta_y,\theta_x)-U_\Omega(\theta_x,\theta_y,\theta_x) \,A_{\Omega_0}\, U^{\dagger}_\Omega(\theta_x,\theta_y,\theta_x)$$ 
and
\begin{eqnarray*}
\Delta_S(\theta,\theta_y) &\equiv& \Sa(H(\theta,\theta_y), \partial_{\theta'} H(\theta',\theta_y)_{\theta'=\theta})-\Sa^{(L/24)}(H(\theta,\theta_y), \partial_{\theta'} H_\Omega(\theta',\theta_y)_{\theta'=\theta}) \\
&=& \Delta^{(L/24)}_S(\theta,\theta_y) + \Sa^{(L/24)}(H(\theta,\theta_y), \partial_{\theta'} H_{\Omega^c}(\theta',\theta_y)_{\theta'=\theta}),
\end{eqnarray*}
for 
\begin{eqnarray*}
\Delta^{(L/24)}_S(\theta,\theta_y) &\equiv& \Sa(H(\theta,\theta_y), \partial_{\theta'} H_{\Omega}(\theta',\theta_y)_{\theta'=\theta}) -\Sa^{(L/24)}(H(\theta,\theta_y), \partial_{\theta'} H_\Omega(\theta',\theta_y)_{\theta'=\theta}) \\
&+& \Sa(H(\theta,\theta_y), \partial_{\theta'} H_{\Omega^c}(\theta',\theta_y)_{\theta'=\theta})-\Sa^{(L/24)}(H(\theta,\theta_y), \partial_{\theta'} H_{\Omega^c}(\theta',\theta_y)_{\theta'=\theta}).
\end{eqnarray*}
Now, we have
\begin{eqnarray}
\|\Delta_U(\theta_x,\theta_y)\| &=& \left\|\int_0^{\theta_x}  \partial_{\theta'} \left(U^{\dagger}_X(0,\theta_y,\theta')\,U_\Omega(\theta_x,\theta_y,\theta_x-\theta') \,A_{\Omega_0}\, U^{\dagger}_\Omega(\theta_x,\theta_y,\theta_x-\theta') \,U_X(0,\theta_y,\theta')\right)_{\theta'=\theta}\, d\theta \right\| \nonumber\\
&\le& |\theta_x| \sup_{\theta \in [0,\theta_x]} \|[U_\Omega(\theta_x,\theta_y,\theta) \,A_{\Omega_0}\, U^{\dagger}_\Omega(\theta_x,\theta_y,\theta), \Delta_S(\theta,\theta_y)]\|\nonumber\\
&\le& 2 |\theta_x| \|A_{\Omega_0}\| \sup_{\theta \in [0,\theta_x]} \left\|\Delta^{(L/24)}_S(\theta,\theta_y)\right\| \nonumber\\
&+& |\theta_x| \sup_{\theta \in [0,\theta_x]} \left\|[U_\Omega(\theta_x,\theta_y,\theta) \,A_{\Omega_0}\, U^{\dagger}_\Omega(\theta_x,\theta_y,\theta), \Sa^{(L/24)}(H(\theta,0), \partial_{\theta'} H_{\Omega^c}(\theta',0)_{\theta'=\theta})]\right\|\label{bound:twist}
\end{eqnarray}
Applying the Lieb-Robinson bound (\ref{lr-bound:adiabatic}), we have for $\theta \in [0,2\pi]$:
\be\label{bound:lr_twist}
\left\|\left[U_\Omega(\theta_x,\theta_y,\theta) \,A_{\Omega_0}\, U^{\dagger}_\Omega(\theta_x,\theta_y,\theta), \Sa^{(L/24)}(H(\theta,0), \partial_{\theta} H_{\Omega^c}(\theta,0))\right]\right\| \le  \frac{4\alpha}{\sqrt{2\pi}} Q_{\max}\, J \,L |{\Omega_0}| \|A_{\Omega_0}\| e^{- \frac{L}{24 d_0} + 2\pi \,v(Q_{\max}, d_0)},
\ee
where $d_0$ is given in (\ref{def:d_0}), $v(Q_{\max}, d_0)$ is defined in (\ref{velocity_bound}) and we made use of (\ref{naive_bnd:sa}) with the bound (\ref{bound:terms}), to get:
$$\|\Sa^{(L/24)}(H(\theta,0), \partial_{\theta} H_{\Omega^c}(\theta,0))\| \le \frac{2\alpha}{\sqrt{2\pi}} Q_{\max}\, J \,L.$$ 
Moreover, $|{\Omega_0}| \le L^2/64$. Now, setting 
\be
\xi = \frac{R^2}{4 (2\pi) Q_{\max}}
\ee
and assuming
$d_0 \le \left(\frac{L \xi}{48 \sqrt{\ln L}}\right)^{\frac{1}{5}}$
implies the bounds $\frac{L}{48 d_0} \ge 2\pi \,v(Q_{\max}, d_0)$ and, hence, $\frac{L}{24 d_0} - 2\pi \,v(Q_{\max}, d_0) \ge \frac{L}{48 d_0} \ge \left(\frac{L}{48}\right)^{4/5} \left(\frac{\sqrt{\ln L}}{\xi}\right)^{1/5} \ge \left(\frac{L}{48 R}\right)^{4/5}$ for the exponent in (\ref{bound:lr_twist}). Going back to (\ref{bound:lr_twist}), we get:
\be\label{bound:lr_twist2}
\left\|\left[U_\Omega(\theta_x,\theta_y,\theta) \,A_{\Omega_0}\, U^{\dagger}_\Omega(\theta_x,\theta_y,\theta), \Sa^{(L/24)}(H(\theta,0), \partial_{\theta} H_{\Omega^c}(\theta,0))\right]\right\| \le  \frac{\alpha}{16\sqrt{2\pi}} Q_{\max}\, J \|A_{\Omega_0}\|\, L^3 e^{- \left(\frac{L}{48 R}\right)^{4/5}},
\ee
Finally, using Lemma \ref{lem:sa_approx}, after noting that $\frac{L}{24R}\ge (d_0/2R)^2 \ge \sigma^2$, we have the bound 
\be\label{bound:local_twist}
2 \|A_{\Omega_0}\| \sup_{\theta \in [0,\theta_x]} \left\|\Delta^{(L/24)}_S(\theta,\theta_y)\right\| \le \frac{32\alpha}{3\sqrt{2\pi}} Q_{\max}\, J\,\|A_{\Omega_0}\| \frac{L^2}{\sigma^2} \,e^{-\frac{L}{48R}}.
\ee
Noting that (\ref{bound:lr_twist2}) dominates (\ref{bound:local_twist}) in (\ref{bound:twist}), completes the proof of the first bound in the Lemma. To justify the choice of the global parameter $\alpha$ in the statement of this Lemma, we recall that we have assumed $d_0 = 2\sigma R \sqrt{\ln (8\sigma^2 R)} \le \left(\frac{L \xi}{48 \sqrt{\ln L}}\right)^{\frac{1}{5}}$ and since $\sigma \equiv 2\alpha v$, we choose 
\be
\alpha = \frac{1}{4vR} \left(\frac{L \xi}{48 \ln^3 L}\right)^{\frac{1}{5}}
\ee
to satisfy the bound on $d_0$ while maximizing $\alpha$.
\end{proof}
\end{lemma}
\section{Computing the Hall Conductance}
We use the Kubo formula from linear response theory as applied to the setting of a torus pierced by two solenoids carrying magnetic fluxes $\theta_x$ and $\theta_y$ in the $x$ and $y$ directions, respectively\cite{thouless}. We denote the Hall conductance at the origin in flux-space by $\sigma_{xy}$, given by the formula:
\begin{equation}\label{def:cond}
\sigma_{xy} = 2 \Im \left\lbrace\braket{\frac{\partial}{\partial \theta_y}\Psi_0(0,\theta_y)_{\theta_y=0}}{\frac{\partial}{\partial \theta_x}\Psi_0(\theta_x,0)_{\theta_x=0}}\right\rbrace \cdot \left(2\pi \, \frac{e^2}{h}\right)
\end{equation}
where $\ket{\Psi_0(0,\theta_y)}$ and $\ket{\Psi_0(\theta_x,0)}$ are the ground states of $H(0,0,\theta_y,0)$ and $H(\theta_x,0,0,0)$, respectively. We note here that the flux angles $\theta_x$ and $\theta_y$ measure the flux in units of $\hbar/e$, the unit of magnetic flux divided by $2\pi$, so that $H(\theta_x,\theta_y)$ is periodic in both angles with period $2\pi$. With this in mind, we may also write for dimensionless $\theta_x,\theta_y$:
 
To show quantization of $\sigma_{xy}$, it is sufficient to prove that $\left|1- e^{2\pi i\,\sigma_{xy}/(e^2/h)}\right|$ is bounded above by a function that decays as a stretched exponential in $L$. It follows then, that $\exists \,n \in \N$ such that:
\be\label{main_bound}
\left|\sigma_{xy}-n \frac{e^2}{h}\right| \le \frac{\sqrt{2} e^2}{2\pi h} \left(
\left|\braket{\Psi_0}{\Psi_{\circlearrowleft}(r)}^{\left(\frac{2\pi}{r}\right)^2} - e^{i\,\sigma_{xy}\frac{2\pi h}{e^2}}\right| + \left|1-\braket{\Psi_0}{\Psi_{\circlearrowleft}(2\pi)}\right|
+ \left|\braket{\Psi_0}{\Psi_{\circlearrowleft}(2\pi)}-\braket{\Psi_0}{\Psi_{\circlearrowleft}(r)}^{\left(\frac{2\pi}{r}\right)^2}\right| \right)
\ee
where we used
\be
\left|\sigma_{xy}\frac{2\pi h}{e^2}-2\pi n\right| \le {\sqrt{2}} \left|1- e^{i\,\left(\sigma_{xy}\frac{2\pi h}{e^2}-2\pi n\right)}\right|,\nonumber
\ee
which follows from assuming $\left|1-e^{i\,\theta}\right|\le 1\implies |1-e^{i\theta}|^2 = 2(1-\cos\theta) \le 1\implies |\theta| \in [0,\pi/3]$, up to integer multiples of $2\pi$. In that case, $\cos|\theta|$ is monotonically decreasing from $1$ to $1/2$ in that range. In particular, we have $$\left|1-e^{i|\theta|}\right|^2 = 2(1-\cos|\theta|) = 2\int_0^{|\theta|} d\phi \int_0^{\phi} \cos\eta \, d\eta \ge 2\int_0^{|\theta|}  {\phi} \cos\phi \,d\phi \ge |\theta|^2 \cos|\theta| \ge |\theta|^2/2.$$ Since $\left|1-e^{i|\theta|}\right|=\left|1-e^{i\theta}\right|$, we get $|\theta| \le \sqrt{2} \left|1-e^{i\theta}\right|$.

At this point, we are ready to derive bounds for each of the three terms in (\ref{main_bound}). The first one emerges from the close relationship between the Hall conductance and the geometric/Berry's phase $\phi(r)$ that appears in Corollary \ref{cor:phase-estimate}. In particular, we know from Corollary \ref{cor:phase-estimate} that for $H(\theta_x,\theta_y) = H(\theta_x,0,\theta_y,0)$ and small enough $r$, the following is true:
\begin{eqnarray}\label{adiabatic_phase_1}
1- |\braket{\Psi_0}{\Psi_{\circlearrowleft}(r)}| \le \left|\braket{\Psi_0}{\Psi_{\circlearrowleft}(r)} - e^{i\phi(r)}\right| \le 8 Q_{\max} J \,L \, \frac{e^{-\alpha^2\Delta^2}}{\Delta}\cdot r,\quad
\Delta = \gamma - 2 Q_{\max}\, JL\, r,
\end{eqnarray}
where we used (\ref{bound:rot-H_X}-\ref{bound:rot-H_Y}) to bound $D_H(r)$ in (\ref{phase-estimate}) and the gap estimate (\ref{bound:gap-estimate}) to compute the uniform lower bound for the spectral gap of the family of Hamiltonians $\{H(\theta_x,0,\theta_y,0)\}_{(\theta_x,\theta_y)\in\Lambda(r)}$, remembering that $H(0,0,0,0)=H_0$ has gap $\gamma > 0$, by assumption. From the definition of $\Delta$ above, we can see exactly how small $r$ needs to be in order for the gap to remain open during the quasi-adiabatic evolution. At this point, we note the following useful inequality:
\be\label{adiabatic_phase_2}
\left|\braket{\Psi_0}{\Psi_{\circlearrowleft}(r)}^{\left(\frac{2\pi}{r}\right)^2} - e^{4\pi^2\, i\phi(r)/r^2}\right| \le 24 Q_{\max} J \,L \, \frac{e^{-\alpha^2\Delta^2}}{r\cdot \Delta},\quad
\Delta = \gamma - 2 Q_{\max}\, JL\, r.
\ee
To show (\ref{adiabatic_phase_2}), we use the
general inequality for real $b$ with $0 \le b \le 1$ and $m\in \N$ derived in the Appendix:
\be\label{bound:power_difference}
\left |b - e^{i\theta}\right| \le \epsilon \le 1/2 \implies \left |b^m - e^{im\theta}\right| \le \sqrt{\frac{7}{3}}\, m\epsilon.
\ee
The above relation implies that $0\le b\le 1$:
\be
\left |be^{i\phi} - e^{i\theta}\right| = \left |b - e^{i(\theta-\phi)}\right| \le \epsilon \le 1/2 \implies \left |(be^{i\phi})^m - e^{im\theta}\right| = \left |b^m - e^{im(\theta-\phi)}\right| \le \sqrt{\frac{7}{3}}\, m\epsilon,
\ee
which combines with (\ref{adiabatic_phase_1}) to give (\ref{adiabatic_phase_2}).
Moreover, combining the simple inequality 
$$\left|e^{i\phi_1}-e^{i\phi_2}\right| = \left|e^{i(\phi_1-\phi_2)/2}-e^{-i(\phi_1-\phi_2)/2}\right| = 2\sin|(\phi_1-\phi_2)/2| \le |\phi_1-\phi_2|$$ with (\ref{bound:phase_partials}), (\ref{def:cond}) and (\ref{adiabatic_phase_2}), we get:
\begin{cor}\label{cor:adiabatic_phase_3}
For $0 < r \le \frac{\gamma}{4Q_{\max} J L}$, the following bound holds for a numeric constant $C >0$:
\begin{eqnarray}\label{adiabatic_phase_3}
\left|\braket{\Psi_0}{\Psi_{\circlearrowleft}(r)}^{\left(\frac{2\pi}{r}\right)^2} - e^{2\pi i\,\sigma_{xy}/(e^2/h)}\right| \le C \left(Q_{\max}\frac{J}{\gamma} L\right) \left(\frac{e^{-\frac{\alpha^2\gamma^2}{2}}}{r} + \left(Q_{\max} \frac{J}{\gamma} L\right)^2 r \right),
\end{eqnarray}
\end{cor}
where we used (\ref{bound:partial}-\ref{bound:partial_2}) and similar bounds for the norms of partials of the form $\partial^2_s, \partial_{s'} \partial_s, \partial_{s'} \partial^2_s, \partial^3_s$ acting on $\ket{\Psi_0(s',s)}$ and $\ket{\Psi_0(s,s')}$, in order to bound the partials $\partial_s g(s,s)$ and $\partial_s \, [g(s',s)+g(s,s')]$ in (\ref{bound:phase_partials}). In particular, these bounds are in effect bounds on the norms of the same set of partials acting on $H(s',0,s,0)$ and $H(s,0,s',0)$, for which we use arguments similar to those leading to (\ref{bound:rot-H_X}-\ref{bound:rot-H_Y}) to estimate.

The bound for the second term in (\ref{main_bound}) is given as Corollary \ref{cor:gs_evol} below. Before we derive the bound, we show certain trace inequalities and develop energy estimates that allow us to study the quasi-adiabatic evolution of $\Psi_0(0,0)$ under a closed path in parameter space, when we can no longer rely on the assumption of a uniform lower bound $\Delta >0$ for the spectral gap. Note here that we consider {\it each} of the paths 
\be\label{path_components}
\Lambda_1(2\pi): (0,0)\rightarrow (2\pi,0),\, \Lambda_2(2\pi): (2\pi,0)\rightarrow (2\pi,2\pi), \, \Lambda_3(2\pi):(2\pi,2\pi)\rightarrow (0,2\pi), \, \Lambda_4(2\pi):(0,2\pi)\rightarrow (0,0)
\ee
to be closed in parameter space, since $H_0=H(2\pi,0,0,0) = H(2\pi,0,2\pi,0)=H(0,0,2\pi,0)$, which follows directly from the Aharonov-Bohm $2\pi$-periodicity of the interactions $\Phi(Z;\theta_x,0,\theta_y,0)$.

\subsection{Partial trace approximation}
We now prove trace inequalities that describe the
partial trace of the quasi-adiabatic evolution of $\ket{\Psi_0}$, both near and far from the twists driving the evolution.   The estimates after tracing
out sites near from the twists driving the evolution are based on locality,
while the estimates after tracing out sites far from the twists driving
the evolution use the virtual flux idea of \cite{hast-lsm}, by comparing
the partial trace to a partial trace of a density matrix in which the
angles $\phi$ are twisted also.

First, we define the following useful states:
\begin{eqnarray}
\ket{\Psi_X(\theta)} &=& U_X(0,0,\theta)\,\ket{\Psi_0}, \quad \ket{\Psi_X(\theta,-\theta)} = U^{(2)}_X(0,0,\theta)\,\ket{\Psi_0}, \label{psi_X}\\ 
\ket{\Psi_Y(\theta)} &=& U_Y(0,0,\theta)\,\ket{\Psi_0}, \quad \ket{\Psi_Y(\theta,-\theta)} = U_Y^{(2)}(0,0,\theta)\,\ket{\Psi_0}\label{psi_Y},
\end{eqnarray}
where the unitaries $U_X(0,0,\theta)$, $U^{(2)}_X(0,0,\theta)$, $U_Y(0,0,\theta)$ and $U^{(2)}_Y(0,0,\theta)$ are defined in (\ref{def:unitary_X}-\ref{def:unitary_Y}), with dynamics based on the one-parameter families of Hamiltonians $\{H(\theta',0,0,0)\}_{\theta'\in [0,\theta]}$, $\{H(\theta',-\theta',0,0)\}_{\theta'\in [0,\theta]}$, $\{H(0,0,\theta',0)\}_{\theta'\in [0,\theta]}$ and $\{H(0,0,\theta',-\theta')\}_{\theta'\in [0,\theta]}$, respectively.
We define the density matrices corresponding to the above pure states by
\begin{eqnarray}
\rho_X(\theta) &=& \pure{\Psi_X(\theta)},\quad \rho_X(\theta,-\theta) = \pure{\Psi_X(\theta,-\theta)}\label{rho_X},\\
\rho_Y(\theta) &=& \pure{\Psi_X(\theta)},\quad \rho_Y(\theta,-\theta) = \pure{\Psi_Y(\theta,-\theta)}\label{rho_Y}.
\end{eqnarray}
Now, for the sets $\Omega_X$, $\Omega_Y$ defined in (\ref{sets_omega}),
we introduce the following Hamiltonian restrictions: 
\begin{eqnarray}
H_{\Omega_X}(\theta) &=& \sum_{Z \subset \Omega_X} \Phi(Z;\theta,0,0,0), \quad H_{\Omega_X^c}(\theta) = \sum_{Z \not \subset \Omega_X}\Phi(Z;0,-\theta,0,0),\\ 
H_{\Omega_Y}(\theta) &=& \sum_{Z \subset \Omega_Y} \Phi(Z;0,0,\theta,0), \quad H_{\Omega_Y^c}(\theta) = \sum_{Z \not \subset \Omega_Y}\Phi(Z;0,0,0,-\theta).
\end{eqnarray}
Using the above restrictions, we have the following useful decompositions:
\begin{eqnarray}\label{omega_X_decomposition}
H(\theta,0,0,0) &=& H_{\Omega_X}(\theta) + H_{\Omega_X^c}(0), \quad H(\theta,-\theta,0,0) = H_{\Omega_X}(\theta) + H_{\Omega_X^c}(\theta),\\
H(0,0,\theta,0) &=& H_{\Omega_Y}(\theta) + H_{\Omega_Y^c}(0), \quad H(0,0,\theta,-\theta) = H_{\Omega_Y}(\theta) + H_{\Omega_Y^c}(\theta).
\label{omega_Y_decomposition}
\end{eqnarray}

For any set $Z$, we define $\overline Z$ to be the complement of $Z$.
We define $\Omega_X^c$ to be the set of sites within distance $R$ of $\overline{\Omega_X}$ and we define $\Omega_Y^c$ to be the set of sites within distance $R$ of
$\overline{\Omega_Y}$.  Then, $H_{\Omega_X^c}$ is supported on $\Omega_X^c$ and
similarly $H_{\Omega_Y^c}$ is supported on $\Omega_Y^c$.  Thus, $H$ is a
sum of terms $\Phi(Z)$ such that each set $Z$ is either supported on $\Omega_X$, or
on $\Omega_X^c$.

\begin{lemma}[Partial trace approximations]\label{lem:partial_trace}
The following partial-trace norm inequalities hold for the states and sets defined above, for a numeric constant $C > 0$:
\begin{flalign*}
&\max\left\{\|\Tr_{\overline{\Omega_X^c}} (\rho_X(\theta) - \rho_X(0))\|_1,\, \|\Tr_{\overline{\Omega_Y^c}} (\rho_Y(\theta) - \rho_Y(0))\|_1\right\} \le C |\theta| \frac{\alpha \gamma}{\sqrt{2\pi}}\, \left(Q_{\max} \frac{J}{\gamma} L\right) \frac{L}{\sigma^2} e^{- \frac{L}{8R}}, \\
&\max\{\|\Tr_{\overline{\Omega_X}} \left(\rho_X(\theta) - R_X(\theta, \rho_X(0))\right)\|_1,\,
 \|\Tr_{\overline{\Omega_Y}} \left(\rho_Y(\theta) - R_Y(\theta,\rho_Y(0))\right)\|_1\} \le \\
 &C |\theta| \left(Q_{\max} \frac{J}{\gamma} L\right)\left(3 \frac{\alpha \gamma}{\sqrt{2\pi}} \frac{L}{\sigma^2} e^{- \frac{L}{8R}} + e^{-\alpha^2\gamma^2}\right)
 \end{flalign*}
 \begin{proof}
We only prove the bounds for $\rho_X$, since the bounds for $\rho_Y$ follow similarly.
Setting $H^{(1)}(\theta)=H(\theta,0,0,0)$ and $H^{(2)}(\theta)=H(\theta,-\theta,0,0)$, we note that
\begin{eqnarray*}
\partial_{\theta'} \rho_X(\theta') &=& i\left[\Sa\left(H^{(1)}(\theta'),\partial_{\theta'} H_{\Omega_X}(\theta')\right),\rho_X(\theta')\right], \\
\partial_{\theta'} \rho_X(\theta',-\theta') &=& i\left[\Sa\left(H^{(2)}(\theta'),\partial_{\theta'} H_{\Omega_X}(\theta')\right),\rho_X(\theta',-\theta')\right] + i\left[\Sa\left(H^{(2)}(\theta'),\partial_{\theta'} H_{\Omega_X^c}(\theta')\right),\rho_X(\theta',-\theta')\right].
\end{eqnarray*}
Furthermore, recalling the definition of $\Sa^{(M)}(H,A)$ given in (\ref{sa_decomposition}) and using the fact that $\partial_{\theta'} H_{\Omega_X}(\theta')$ has support on a strip of width $2R$ centered on the line $x=1$, we have:
\begin{eqnarray*}
\Sa^{(L/4-R)}(H^{(1)}(\theta'),\partial_{\theta'} H_{\Omega_X}(\theta')) = \Sa^{(L/4-R)}(H^{(2)}(\theta'),\partial_{\theta'} H_{\Omega_X}(\theta'))
\end{eqnarray*}
with support on $\overline{\Omega^c_X}$.
Moreover, from Lemma \ref{lem:sa_approx} with $ L/4-R \ge \sigma^2 R$, we get:
\begin{eqnarray*}
\left\|\Sa(H^{(1)}(\theta'),\partial_{\theta'} H_{\Omega_X}(\theta'))-\Sa^{(L/4-R)}(H^{(1)}(\theta'),\partial_{\theta'} H_{\Omega_X}(\theta'))\right\| &\le& \frac{16\alpha}{\sqrt{2\pi}}\, Q_{\max} J \frac{L^2}{\sigma^2} \, e^{- \frac{L}{8R}+\frac{1}{2}}\\
\left\|\Sa(H^{(2)}(\theta'),\partial_{\theta'} H_{\Omega_X}(\theta'))-\Sa^{(L/4-R)}(H^{(1)}(\theta'),\partial_{\theta'} H_{\Omega_X}(\theta'))\right\| &\le&  \frac{16\alpha}{\sqrt{2\pi}}\, Q_{\max} J \frac{L^2}{\sigma^2}\, e^{- \frac{L}{8R}+\frac{1}{2}}\\
\left\|\Sa(H^{(2)}(\theta'),\partial_{\theta'} H_{\Omega_X^c}(\theta'))-\Sa^{(L/4-R)}(H^{(2)}(\theta'),\partial_{\theta'} H_{\Omega_X^c}(\theta'))\right\| &\le& \frac{16\alpha}{\sqrt{2\pi}}\, Q_{\max} J \frac{L^2}{\sigma^2} \, e^{- \frac{L}{8R}+\frac{1}{2}}.
\end{eqnarray*}
We also note that $\Sa^{(L/4-R)}(H^{(2)}(\theta'),\partial_{\theta'} H_{\Omega_X^c}(\theta'))$ is supported on $\overline{\Omega_X}$, since $\partial_{\theta'} H_{\Omega_X^c}(\theta')$ has support on a strip of width $2R$ centered on the line $x=L/2+1$. Equipped with the above observations, we turn our attention to proving the upper bound on $\|\Tr_{\overline{\Omega_X^c}} (\rho_X(\theta) - \rho_X(0))\|_1$:
\begin{flalign*}
&\|\Tr_{\overline{\Omega_X^c}} (\rho_X(\theta) - \rho_X(0))\|_1 
= \left\|\int_0^\theta \Tr_{\overline{\Omega_X^c}} [\Sa(H^{(1)}(\theta'),\partial_{\theta'} H_{\Omega_X}(\theta')),\rho_X(\theta')]\, d\theta' \right\|_1 \\
&\le |\theta| \sup_{\theta'\in[0,\theta]} \left\|\Tr_{\overline{\Omega_X^c}} [\Sa(H^{(1)}(\theta'),\partial_{\theta'} H_{\Omega_X}(\theta')),\rho_X(\theta')] \right\|_1
= |\theta| \sup_{\theta'\in[0,\theta]} \sup_{\stackrel{A\in \A_{\Omega^c_X}}{\|A\|=1}} |\Tr(A\,[\Sa(H^{(1)}(\theta'),\partial_{\theta'} H_{\Omega_X}(\theta')),\rho_X(\theta')])|\\
&= |\theta| \sup_{\theta'\in[0,\theta]} \sup_{\stackrel{A\in \A_{\Omega^c_X}}{\|A\|=1}} |\Tr([A,\Sa(H^{(1)}(\theta'),\partial_{\theta'} H_{\Omega_X}(\theta'))] \,\rho_X(\theta'))|
\le  |\theta| \sup_{\theta'\in[0,\theta]} \sup_{\stackrel{A\in \A_{\Omega^c_X}}{\|A\|=1}} \|[A,\Sa(H^{(1)}(\theta'),\partial_{\theta'} H_{\Omega_X}(\theta'))]\|\\
&=  |\theta| \sup_{\theta'\in[0,\theta]} \sup_{\stackrel{A\in \A_{\Omega^c_X}}{\|A\|=1}} \|[A,\Sa(H^{(1)}(\theta'),\partial_{\theta'} H_{\Omega_X}(\theta'))-\Sa^{(L/4-R)}(H^{(1)}(\theta'),\partial_{\theta'} H_{\Omega_X}(\theta'))]\|\\
&\le 2\, |\theta|\, \sup_{\theta'\in[0,\theta]} \|\Sa(H^{(1)}(\theta'),\partial_{\theta'} H_{\Omega_X}(\theta'))-\Sa^{(L/4-R)}(H^{(1)}(\theta'),\partial_{\theta'} H_{\Omega_X}(\theta'))\|
\le |\theta| \frac{32\alpha}{\sqrt{2\pi}}\, Q_{\max} J \frac{L^2}{\sigma^2} \, e^{- \frac{L}{8R}+\frac{1}{2}}.
\end{flalign*}
To prove the bound on $ \|\Tr_{\overline{\Omega_X}} \left(\rho_X(\theta) - R_X(\theta, \rho_X(0))\right)\|_1$, we define the unitary $U_{{\Omega_X}}(\theta)$ with support on the set $\overline{\Omega_X^c}$ by the differential equation:
\be
\partial_{\theta'} \, U_{{\Omega_X}}(\theta') = i\, \Sa^{(L/4-R)}(H^{(1)}(\theta'),\partial_{\theta'} H_{\Omega_X}(\theta'))\,U_{{\Omega_X}}(\theta'),\quad U_{{\Omega_X}}(0)= \one.
\ee
Now, we have:
\begin{flalign*}
&\|\Tr_{\overline{\Omega_X}} (\rho_X(\theta) - \rho_X(\theta,-\theta))\|_1 = \left\|\Tr_{\overline{\Omega_X}} \left[U^{\dagger}_{{\Omega_X}}(\theta) \, (\rho_X(\theta) - \rho_X(\theta,-\theta))\,U_{{\Omega_X}}(\theta)\right]\right\|_1\\
 &\le \int_0^\theta \left\| \Tr_{\overline{\Omega_X}} \left[\Sa(H^{(1)}(\theta'),\partial_{\theta'} H_{\Omega_X}(\theta'))-\Sa^{(L/4-R)}(H^{(1)}(\theta'),\partial_{\theta'} H_{\Omega_X}(\theta')),\rho_X(\theta')\right] \right\|_1 \, d\theta' \nonumber \\
&+ \int_0^\theta \left\| \Tr_{\overline{\Omega_X}} \left[\Sa(H^{(2)}(\theta'),\partial_{\theta'} H_{\Omega_X}(\theta'))-\Sa^{(L/4-R)}(H^{(1)}(\theta'),\partial_{\theta'} H_{\Omega_X}(\theta')),\rho_X(\theta',-\theta')\right] \right\|_1 d\theta' \nonumber\\
 &+ \int_0^\theta \left\| \Tr_{\overline{\Omega_X}} \left[\Sa(H^{(2)}(\theta'),\partial_{\theta'} H_{\Omega^c_X}(\theta')),\rho_X(\theta',-\theta')\right] \right\|_1\, d\theta' \\
 &\le 2|\theta| \sup_{\theta'\in [0,\theta]} \left\|\Sa(H^{(1)}(\theta'),\partial_{\theta'} H_{\Omega_X}(\theta'))-\Sa^{(L/4-R)}(H^{(1)}(\theta'),\partial_{\theta'} H_{\Omega_X}(\theta')) \right\|\, \nonumber \\
 &+ 2|\theta| \sup_{\theta'\in [0,\theta]} \left\|\Sa(H^{(2)}(\theta'),\partial_{\theta'} H_{\Omega_X}(\theta'))-\Sa^{(L/4-R)}(H^{(1)}(\theta'),\partial_{\theta'} H_{\Omega_X}(\theta')) \right\|\, \nonumber \\
 &+ 2 |\theta| \sup_{\theta'\in[0,\theta]} \left\|\Sa(H^{(2)}(\theta'),\partial_{\theta'} H_{\Omega^c_X}(\theta'))-\Sa^{(L/4-R)}(H^{(2)}(\theta'),\partial_{\theta'} H_{\Omega^c_X}(\theta')) \right\| \\
 &\le |\theta|  \frac{96\alpha}{\sqrt{2\pi}}\, Q_{\max} J \frac{L^2}{\sigma^2} \, e^{- \frac{L}{8R}+\frac{1}{2}}
\end{flalign*}
Finally, noting that $\rho_X(0,0) = \rho_X(0)=P_0$ and that $\{R_X(\theta', \rho_X(0,0))\}_{\theta'\in[0,\theta]}$ is the family of ground states corresponding to the unitarily equivalent family of Hamiltonians $\{H(\theta',-\theta',0,0)\}_{\theta'\in[0,\theta]}$ with spectral gap $\gamma$, we bound:
\be
\|\Tr_{\overline{\Omega_X}} (\rho_X(\theta,-\theta) - R_X(\theta, \rho_X(0,0)))\|_1 \le \|\rho_X(\theta,-\theta) - R_X(\theta, \rho_X(0,0))\|_1 \le |\theta|\, Q_{\max} \frac{J}{\gamma}\, L\, e^{-\alpha^2\gamma^2}
\ee
by applying (\ref{delta-small}) to the quasi-adiabatic evolution of $\rho_X(0,0)$.
Using the triangle inequality 
$$\|\Tr_{\overline{\Omega_X}} \left(\rho_X(\theta) - R_X(\theta, \rho_X(0))\right)\|_1\le \|\Tr_{\overline{\Omega_X}} (\rho_X(\theta) - \rho_X(\theta,-\theta))\|_1+\|\Tr_{\overline{\Omega_X}} (\rho_X(\theta,-\theta) - R_X(\theta, \rho_X(0,0)))\|_1$$
with the above bounds, completes the proof.
\end{proof}
\end{lemma}
\subsection{Energy estimates}
We now prove the energy estimates.
\begin{prop}[Energy estimate]\label{prop:energy}
For the states $\ket{\Psi_X(\theta)}$ and $\ket{\Psi_Y(\theta)}$ defined in (\ref{psi_X}) and (\ref{psi_Y}), respectively, the following energy estimate is true, for a numeric constant $C >0$:
\begin{eqnarray*}
\left|\braket{\Psi_X(\theta)}{H(\theta,0,0,0)\, \Psi_X(\theta)}-E_0\right| \le C |\theta| \,Q_{\max}\frac{J}{\gamma} L^3 \left( \frac{\alpha \gamma}{\sqrt{2\pi}}\,\frac{L}{\sigma^2} \, e^{- \frac{L}{8R}} + e^{-\alpha^2\gamma^2}\right) J,
\end{eqnarray*}
with $E_0$ the ground state energy of $H_0$. The same estimate holds for $\left|\braket{\Psi_Y(\theta)}{H(0,0,\theta,0)\, \Psi_Y(\theta)}-E_0\right|$.
\begin{proof}
We will show the bound for $\ket{\Psi_X(\theta)}$, since the bound for $\ket{\Psi_Y(\theta)}$ follows from a symmetric argument. Noting that $R_X(\theta, \rho_X(0))$ is the ground state of $H(\theta,-\theta,0,0)$, we have:
\be
E_0(0) = \Tr (H_0 \rho_X(0)) = \Tr (H(\theta,-\theta,0,0)\, R_X(\theta, \rho_X(0))) = \Tr (H_{{\Omega_X}}(\theta) \,R_X(\theta, \rho_X(0)))+ \Tr (H_{\Omega^c_X}(0) \rho_X(0))\nonumber
\ee
where we used the unitary equivalence from (\ref{twist-anti-twist}) to get the second equality and also, 
$$\Tr (H_{\Omega^c_X}(\theta) R_X(\theta, \rho_X(0))) = \Tr (H_{\Omega^c_X}(0) \rho_X(0)).$$
Now, recalling the decompositions given in (\ref{omega_X_decomposition}) and noting that $\rho_X(0) = \pure{\Psi_0}$, we have the following estimates:
\begin{flalign*}
&\left|\braket{\Psi_X(\theta)}{H(\theta,0,0,0)\, \Psi_X(\theta)}-E_0(0)\right| = |\Tr \left(H_{\Omega_X}(\theta) (\rho_X(\theta) - R_X(\theta, \rho_X(0)))\right) + \Tr \left(H_{\Omega^c_X}(0) (\rho_X(\theta) - \rho_X(0))\right)|\\
&\le |\Tr \left(H_{\Omega_X}(\theta) (\rho_X(\theta) - R_X(\theta, \rho_X(0)))\right)| + |\Tr \left(H_{\Omega^c_X}(0) (\rho_X(\theta) - \rho_X(0))\right)| \\
&\le \|H_{\Omega_X}(\theta)\|\cdot  \left(\|\Tr_{\overline{\Omega_X}} (\rho_X(\theta) - R_X(\theta, \rho_X(0)))\|_1\right) 
+ \|H_{\Omega^c_X}(0)\|\cdot \|\Tr_{{\overline{\Omega_X^c}}} (\rho_X(\theta) - \rho_X(0))\|_1\\
&\le J\, L^2\, \left(\|\Tr_{\overline{\Omega_X}} (\rho_X(\theta) - R_X(\theta,\rho_X(0)))\|_1 + \|\Tr_{{\overline{\Omega_X^c}}} (\rho_X(\theta) - \rho_X(0))\|_1\right)
\end{flalign*}
and using Lemma \ref{lem:partial_trace}, completes the proof. \end{proof}
\end{prop}
\subsection{The phase around the big loop.}
In order to derive the next bound, we break the evolution of $\ket{\Psi_{\circlearrowleft}(2\pi)}$ around $\Lambda(2\pi)$ into its four individual components $\{\Lambda_i(2\pi)\}_{i=1}^4$ given in (\ref{path_components}) and apply Proposition \ref{prop:energy} to each one, which gives:
\begin{cor}\label{cor:gs_evol}
For a numeric constant $C > 0$, the following bound holds:
\be\label{gs_evol}
\left|\braket{\Psi_0}{\Psi_{\circlearrowleft}(2\pi)} - 1\right| \le C |\theta| \,Q_{\max} \left(\frac{J}{\gamma}\right)^2 L^3 \left( \frac{\alpha \gamma}{\sqrt{2\pi}}\,\frac{L}{\sigma^2}\, e^{- \frac{L}{8R}}+e^{-\alpha^2\gamma^2}\right),
\ee
\begin{proof}
We begin by noting that the $2\pi$ periodicity of $H(\theta_x,0,\theta_y,0)$ in each angle, implies:
\begin{flalign*}
&\braket{\Psi_0}{\Psi_{\circlearrowleft}(2\pi)} = \braket{\Psi_0}{U^{\dagger}_Y(0,0,2\pi)\, U^{\dagger}_X(0,0,2\pi)\,U_Y(0,0,2\pi)\,U_X(0,0,2\pi)\Psi_0}\\
&= \braket{\Psi_Y(2\pi)}{P_0\, U^{\dagger}_X(0,0,2\pi)\,P_0 \,U_Y(0,0,2\pi)\,P_0\,\Psi_X(2\pi)}
+ \braket{\Psi_Y(2\pi)}{P_0\, U^{\dagger}_X(0,0,2\pi)\,Q_0 \,U_Y(0,0,2\pi)\,P_0\,\Psi_X(2\pi)}\nonumber\\
&+ \braket{\Psi_Y(2\pi)}{P_0\, U^{\dagger}_X(0,0,2\pi)\,U_Y(0,0,2\pi)\,\delta_X(2\pi)}
+ \braket{\delta_Y(2\pi)}{U^{\dagger}_X(0,0,2\pi)\,U_Y(0,0,2\pi)\,\Psi_X(2\pi)}\\
&= |\braket{\Psi_Y(2\pi)}{\Psi_0}|^2\, |\braket{\Psi_0}{\Psi_X(2\pi)}|^2 + \braket{\Psi_Y(2\pi)}{\Psi_0}\, \braket{\delta_X(2\pi)}{\delta_Y(2\pi)}\,\braket{\Psi_0}{\Psi_X(2\pi)} \\
&+ \braket{\Psi_Y(2\pi)}{\Psi_0}\,\braket{\Psi_0}{\Psi'_X(2\pi)}\,\braket{\delta'_Y(2\pi)}{\delta_X(2\pi)} +
\braket{\Psi_Y(2\pi)}{\Psi_0}\,\braket{\delta_X(2\pi)}{U_Y(0,0,2\pi)\,\delta_X(2\pi)} 
\nonumber\\
&+ \braket{\delta_Y(2\pi)}{\delta'_X(2\pi)}\,\braket{\Psi_0}{U_Y(0,0,2\pi)\,\Psi_X(2\pi)} + \braket{\delta_Y(2\pi)}{U^{\dagger}_X(0,0,2\pi)\,\delta_Y(2\pi)}\,\braket{\Psi_0}{\Psi_X(2\pi)} \\
&+ \braket{\delta_Y(2\pi)}{U^{\dagger}_X(0,0,2\pi)\,(1-P_0)\, U_Y(0,0,2\pi)\,\delta_X(2\pi)},
\end{flalign*}
where we inserted $\one = P_0 + Q_0$ between unitaries and set:
\begin{flalign*}
&\ket{\Psi_X(2\pi)} = U_X(0,0,2\pi)\ket{\Psi_0},\quad \ket{\Psi_Y(2\pi)}=U_Y(0,0,2\pi)\ket{\Psi_0},\\
&\ket{\Psi'_X(2\pi)} = U^{\dagger}_X(0,0,2\pi)\ket{\Psi_0},\quad \ket{\Psi'_Y(2\pi)}=U^{\dagger}_Y(0,0,2\pi)\ket{\Psi_0},\\
&\delta_X(2\pi) = Q_0 \ket{\Psi_X(2\pi)} ,\quad  \delta_Y(2\pi) = Q_0 \ket{\Psi_Y(2\pi)},\quad
\delta'_X(2\pi) = Q_0 \ket{\Psi'_X(2\pi)},\quad \delta'_Y(2\pi) = Q_0 \ket{\Psi'_Y(2\pi)}.
\end{flalign*}
The above expression for $\braket{\Psi_0}{\Psi_{\circlearrowleft}(2\pi)}$ combined with the fact that 
$$|\braket{\Psi_Y(2\pi)}{\Psi_0}|^2 = 1 - \|\delta_Y(2\pi)\|^2,\quad |\braket{\Psi_X(2\pi)}{\Psi_0}|^2 = 1 - \|\delta_X(2\pi)\|^2$$ and 
$\|\delta_Y(2\pi)\|^2=\|\delta'_Y(2\pi)\|^2$, $\|\delta_X(2\pi)\|^2=\|\delta'_X(2\pi)\|^2$, gives the bound:
\be
|\braket{\Psi_0}{\Psi_{\circlearrowleft}(2\pi)}-1| \le 2(\|\delta_Y(2\pi)\|^2+\|\delta_X(2\pi)\|^2) + 4 \|\delta_Y(2\pi)\|\,\|\delta_X(2\pi)\| + \|\delta_Y(2\pi)\|^2\,\|\delta_X(2\pi)\|^2
\ee
Finally, noting that $H(2\pi,0,0,0)=H(0,0,2\pi,0)=H_0 \ge E_0 \one+ \gamma\, Q_0$, we have
\be
\|\delta_Y(2\pi)\|^2 \le \frac{1}{\gamma}\, |\braket{\Psi_Y(2\pi)}{H_0\,\Psi_Y(2\pi)}-E_0|, \quad 
\|\delta_X(2\pi)\|^2 \le \frac{1}{\gamma}\, |\braket{\Psi_X(2\pi)}{H_0\,\Psi_X(2\pi)}-E_0|
\ee
and applying Proposition \ref{prop:energy} with $\theta = 2\pi$ to the above inequalities completes the proof.
\end{proof}
\end{cor}

\section{Decomposing flux-space.}
Now that we have the estimates from Corollary \ref{cor:adiabatic_phase_3} and Corollary \ref{cor:gs_evol}, it remains to show that the following bound holds which is proven below:
\begin{cor}\label{cor:stokes} For a numeric constant $C > 0$, the following bound holds:
\be\label{eq:stokes}
\left|\braket{\Psi_0}{\Psi_{\circlearrowleft}(2\pi)}-\braket{\Psi_0}{\Psi_{\circlearrowleft}(r)}^{\left(\frac{2\pi}{r}\right)^2}\right| \le C  \left(Q_{\max} \alpha J\, L\right)^{3/2} \left(\left(Q_{\max} \alpha J\, L\right) r^{\frac{1}{2}}+ \frac{e^{-\frac{\alpha^2\gamma^2}{2}}}{r}\right) \le C \left(Q_{\max} \alpha J\, L\right)^{5/2} e^{-\frac{\alpha^2\gamma^2}{6}},
\ee
\be
\label{rchoice}
r=\frac{2\pi}{\lfloor  \left(Q_{\max} \alpha J\, L\right)^{2/3} e^{\frac{1}{3} \alpha^2\gamma^2}\rfloor}
\ee
\end{cor}
We proceed with the proof of this bound by
turning our focus on a decomposition process used to break the large evolution around flux space into evolutions around tiny loops on the $(2\pi) \times (2\pi)$ lattice. 
\begin{figure}
\centering
\includegraphics[width=300px]{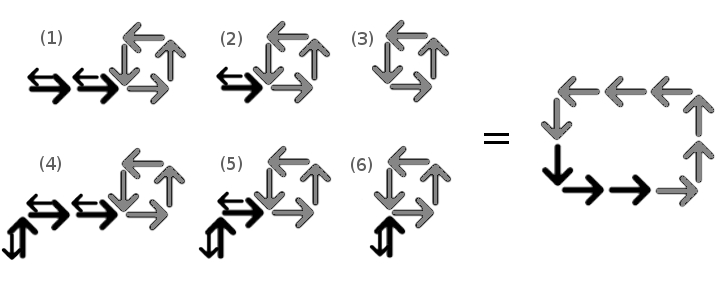}
\caption{{\small{The total evolution around a rectangle of dimension $2\times 3$ is decomposed into $6$ evolutions that all share the common features of evolving from the origin, first with $\theta_y$ and then with $\theta_x$, to reach $(\theta_x,\theta_y)$ in flux space, then following a tiny counter-clockwise loop and finally, reversing the path to go from $(\theta_x,\theta_y)$ to the origin. Any $m \times n$ evolution can be decomposed in this manner, by completing first the bottom row, as in steps $(1)-(3)$, and then stacking the remaining rows on top, as in steps $(4)-(6)$.  That is, the unitary corresponding to quasi-adiabatic evolution around the larger loops is exactly equal to the
product of unitaries corresponding to evolution around the smaller loops.
}}}
\label{fig:decomposition}
\end{figure}
Fig.~\ref{fig:decomposition} describes the process used to decompose the evolution around any $m \times n$ rectangle into $m\cdot n$ evolutions involving the $r\times r$ squares which form the lattice in flux space. This process is effectively a discrete version of Stokes' Theorem.

We define the family of states whose evolution follows the paths that appear in the decomposition process given by Fig.~\ref{fig:decomposition}:
\begin{equation}\label{def:states}
\ket{\Psi_{\circlearrowleft}(\theta_x,\theta_y,r)} = V^{\dagger}[(0,0)\rightarrow(\theta_x,\theta_y)]\,V_{\circlearrowleft}(\theta_x,\theta_y,r)\, V[(0,0)\rightarrow(\theta_x,\theta_y)]  \ket{\Psi_0}.
\end{equation}

To see why these states are important, note that projecting onto the ground-state after every individual cyclic evolution in the decomposition of $\braket{\Psi_0}{\Psi_{\circlearrowleft}(2\pi)}$ given in Fig.~\ref{fig:decomposition}, corresponds to the product of the following $(2\pi/r)^2$ terms:
\begin{equation}\label{eq:decomposition}
\braket{\Psi_0}{\Psi_{\circlearrowleft}(0,2\pi-r,r)}\braket{\Psi_0}{\Psi_{\circlearrowleft}(r,2\pi-r,r)}\cdots\braket{\Psi_0}{\Psi_{\circlearrowleft}(2\pi-r,0,r)}.
\end{equation}
The contribution from terms projecting off the ground-state is a bit more involved. It is shown in the Appendix that in order to prove the bound (\ref{eq:stokes}) in Corollary \ref{cor:stokes}, it suffices to combine the bound (\ref{adiabatic_phase_1}) with a bound of the form proven below:
\begin{lemma}\label{lem:translation}
For sufficiently large $L$, we have:
\begin{equation}\label{eq:translation}
|\braket{\Psi_0}{\Psi_{\circlearrowleft}(\theta_x,\theta_y,r)} -\braket{\Psi_0}{\Psi_{\circlearrowleft}(r)}| \le
C \left(Q_{\max} \alpha J\, L\right)^3 r^2 \left(\left(Q_{\max} \alpha J\, L\right)^2 r^3 + e^{-\alpha^2\gamma^2}\right) , \quad \forall \, \theta_x,\theta_y \in [0,2\pi]
\end{equation}
with $\alpha$ given in Lemma \ref{lem:local_approximation}.
\end{lemma}
Note that the above bound allows us to translate the coordinates $(\theta_x,\theta_y)$ to the origin $(0,0)$ in flux space, where we can make use of the gap $\gamma >0$, to get good bounds on the approximation of the evolution of the ground state $\ket{\Psi_0}$.
\subsection{The Translation Lemma}
We now focus our attention on proving (\ref{eq:translation}), but before we start, we develop the following important Lemma:
\begin{lemma}\label{lem:translation_unitary}
Let $A_{\Omega_0} \in \A_{\Omega_0}$, with $\Omega_0$ defined in $(\ref{omega_0})$. Then, the following bound holds for a numeric constant $C > 0$:
\begin{eqnarray*}
&&\left|\braket{\Psi_0}{A_{\Omega_0} \Psi_0}- \braket{\Psi_0}{V^{\dagger}[(0,0)\rightarrow(\theta_x,\theta_y)]\,R_Y(\theta_y,R_X(\theta_x,A_{\Omega_0}))\, V[(0,0)\rightarrow(\theta_x,\theta_y)]  \Psi_0}\right| \\
&&\le C (|\theta_x|+|\theta_y|) \,\|A_{\Omega_0}\| \left(Q_{\max}\, \frac{J}{\gamma}\, L\right) e^{-\alpha^2\gamma^2}
\end{eqnarray*}
with $\alpha$ given in Lemma \ref{lem:local_approximation}.
\begin{proof}
We recall the density matrices $\rho_X(\theta_x)$ and $\rho_Y(\theta_y)$, defined in $(\ref{rho_X}-\ref{rho_Y})$. Using the rotation identity (\ref{rotation_omega}) for the auxiliary unitary $U_{\Omega}(\theta_x,\theta_y,\theta_x)$ defined in $(\ref{unitary_omega})$ and recalling the definition of $V[(0,0)\rightarrow(\theta_x,\theta_y)]$ in (\ref{def:evol_ur}), we get the following equalities:
\begin{flalign}
&\braket{\Psi_0}{V^{\dagger}[(0,0)\rightarrow(\theta_x,\theta_y)]\,R_Y(\theta_y,R_X(\theta_x,A_{\Omega_0}))\, V[(0,0)\rightarrow(\theta_x,\theta_y)]  \Psi_0} - \braket{\Psi_0}{A_{\Omega_0} \Psi_0}\nonumber\\ 
&= \Tr\left(\rho_Y(\theta_y) \, U^{\dagger}_X(0,\theta_y,\theta_x) \,R_Y(\theta_y,R_X(\theta_x, A_{\Omega_0})) \,U_X(0,\theta_y,\theta_x)\right) - \Tr(P_0 A_{\Omega_0}) \nonumber\\
&= \Tr\left(\rho_Y(\theta_y) \, \left[U^{\dagger}_X(0,\theta_y,\theta_x) \,R_Y(\theta_y,R_X(\theta_x, A_{\Omega_0})) \,U_X(0,\theta_y,\theta_x) - U_\Omega(\theta_x,\theta_y,\theta_x) \,R_Y(\theta_y,R_X(\theta_x, A_{\Omega_0})) \,U^{\dagger}_\Omega(\theta_x,\theta_y,\theta_x)\right]\right) \nonumber\\
&+ \Tr\left(\left[\rho_Y(\theta_y) - R_Y(\theta_y,P_0)\right] \,U_\Omega(\theta_x,\theta_y,\theta_x) \,R_Y(\theta_y,R_X(\theta_x, A_{\Omega_0})) \,U^{\dagger}_\Omega(\theta_x,\theta_y,\theta_x)\right) \nonumber\\
&+ \Tr\left(R_Y(\theta_y,P_0) \, R_Y\left(\theta_y, \left[U_\Omega(\theta_x,0,\theta_x) \,R_X(\theta_x, A_{\Omega_0}) \,U^{\dagger}_\Omega(\theta_x,0,\theta_x)-U^{\dagger}_X(0,0,\theta_x) \,R_X(\theta_x, A_{\Omega_0}) \,U_X(0,0,\theta_x)\right]\right)\right) \nonumber\\
&+\Tr\left([\rho_X(\theta_x)- R_X(\theta_x,P_0)] R_X(\theta_x, A_{\Omega_0})\right)
\end{flalign}
Finally, recalling the sets $\Omega_X$ and $\Omega_Y$ defined in (\ref{sets_omega}) and the comment preceding their definition, we have:
\begin{flalign}
&\left|\braket{\Psi_0}{V^{\dagger}[(0,0)\rightarrow(\theta_x,\theta_y)]\,R_Y(\theta_y,R_X(\theta_x,A_{\Omega_0}))\, V[(0,0)\rightarrow(\theta_x,\theta_y)]  \Psi_0} - \braket{\Psi_0}{A_{\Omega_0} \Psi_0}\right|\nonumber\\ 
&\le \left\|U^{\dagger}_X(0,\theta_y,\theta_x) \,R_Y(\theta_y,R_X(\theta_x, A_{\Omega_0})) \,U_X(0,\theta_y,\theta_x) - U_\Omega(\theta_x,\theta_y,\theta_x) \,R_Y(\theta_y,R_X(\theta_x, A_{\Omega_0})) \,U^{\dagger}_\Omega(\theta_x,\theta_y,\theta_x)\right\| \nonumber\\
&+ \left\|U_\Omega(\theta_x,0,\theta_x) \,R_X(\theta_x, A_{\Omega_0}) \,U^{\dagger}_\Omega(\theta_x,0,\theta_x)-U^{\dagger}_X(0,0,\theta_x) \,R_X(\theta_x, A_{\Omega_0}) \,U_X(0,0,\theta_x)\right\| \nonumber\\
&+ \|A_{\Omega_0}\| \left(\|\Tr_{\overline{\Omega_Y}}(\rho_Y(\theta_y) - R_Y(\theta_y,P_0))\|_1+\|\Tr_{\overline{\Omega_X}}(\rho_X(\theta_x) - R_X(\theta_x,P_0))\|_1 \right) \nonumber\\
&\le \|A_{\Omega_0}\| \left(Q_{\max}\, \frac{J}{\gamma} \, L\right) \left(|\theta_x|\,4\frac{\alpha \gamma}{\sqrt{2\pi}} L^2\,  e^{- \left(\frac{L}{48 R}\right)^{4/5}} + C (|\theta_x|+|\theta_y|) \left(3 \frac{\alpha \gamma}{\sqrt{2\pi}}\,\frac{L}{\sigma^2} \, e^{- \frac{L}{8R}} + e^{-\alpha^2\gamma^2}\right) \right)\label{bound:long_path}
\end{flalign}
where we used Lemma \ref{lem:local_approximation} and Lemma \ref{lem:partial_trace}. Finally, the condition on $\alpha$ in Lemma \ref{lem:local_approximation} implies that the exponential term in (\ref{bound:long_path}) corresponding to:
\be
e^{-\alpha^2 \gamma^2} = e^{-\left(\frac{\gamma}{4v R}\right)^2 \left(\frac{L \xi}{48 \ln^3 L}\right)^{2/5}} \ge \alpha \gamma L^2\,  e^{- \left(\frac{L}{48 R}\right)^{4/5}},
\ee
dominates the other two terms for sufficiently large $L$. Putting everything together gives the desired bound.
\end{proof}
\end{lemma}

\subsection{Localizing the loop unitary $V_{\circlearrowleft}(\theta_x,\theta_y,r)$}
Now, we show that up to order $4$ in $r$, the operator $V_{\circlearrowleft}(0,0,r)$ defined in (\ref{def:unitary_loop_R}), can be approximated by the sum of operators with support on a cross of radius $L/8$ centered at the origin $x=y=0$. Moreover, we show how we may approximate $V_{\circlearrowleft}(\theta_x,\theta_y,r)$ by $R_Y(\theta_y, R_X(\theta_x, V_{\circlearrowleft}(0,0,r)))$, up to a small error in depending on $L$ and $r$.
\begin{lemma}\label{lem:local_loop_unitary}
For sufficiently large $L$, there exists a constant $C > 0$, such that for all $\theta_x,\theta_y \in [0,2\pi]$ we have the bound:
\be
\|V_{\circlearrowleft}(\theta_x,\theta_y,r)- R_Y(\theta_y, R_X(\theta_x, V_{\circlearrowleft}(0,0,r)))\| \le C\left(\frac{\alpha \gamma}{\sqrt{2\pi}}\,\left(Q_{\max} \frac{J}{\gamma} L\right) \frac{L}{\sigma^2}\, e^{- \frac{L}{96R}} \, r + \left(Q_{\max} \alpha J\, L\right)^5 r^5\right).
\ee
Moreover, there exists an operator $W(r) \in \A_{\Omega_0}$, such that $\|W(r)-\one\| \le C \left(Q_{\max} \alpha J L\right)^2 r^2$ and:
\be
\|V_{\circlearrowleft}(0,0,r) - W(r)\| \le C\left(\frac{\alpha \gamma}{\sqrt{2\pi}}\,\left(Q_{\max} \frac{J}{\gamma} L\right) \frac{L}{\sigma^2}\, e^{- \frac{L}{96R}} \, r + \left(Q_{\max} \alpha J\, L\right)^5 r^5\right).
\ee
\end{lemma}
Before proving the above Lemma, we introduce the unitaries $U_{X(M)}(\theta_x,\theta_y,r)$ and $U_{Y(M)}(\theta_x,\theta_y,r)$ defined by the following differential equations:
\begin{eqnarray*}
\partial_s U_{X(M)}(\theta_x,\theta_y,s)_{s=r} &=& i\, \Sa^{(M)}(H(\theta_x+r,0,\theta_y,0), \partial_\theta H(\theta,0,\theta_y,0)_{\theta=\theta_x+r}) \, U_{X(M)}(\theta_x,\theta_y,r), \quad U_{X(M)}(\theta_x,\theta_y,0) = \one \\
\partial_s U_{Y(M)}(\theta_x,\theta_y,s)_{s=r} &=& i\, \Sa^{(M)}(H(\theta_x,0,\theta_y+r,0), \partial_\theta H(\theta_x,0,\theta,0)_{\theta=\theta_y+r}) \, U_{Y(M)}(\theta_x,\theta_y,r), \quad U_{Y(M)}(\theta_x,\theta_y,0) = \one,
\end{eqnarray*}
with $\Sa^{(M)}(H,A)$ defined in (\ref{sa_decomposition}). The following Lemma gives us a bound on the error of approximating the unitaries $U_X(\theta_x,\theta_y,r)$ and $U_Y(\theta_x,\theta_y,r)$ with $U_{X(M)}(\theta_x,\theta_y,r)$ and $U_{Y(M)}(\theta_x,\theta_y,r)$, respectively.
\begin{lemma}\label{lem:local_unitaries}
The following bounds hold for all $\theta_x,\theta_y \in [0,2\pi]$ and $r \ge 0$, with $M \ge \sigma^2 R$:
\begin{eqnarray}
\|U_X(\theta_x,\theta_y,r) - U_{X(M)}(\theta_x,\theta_y,r)\| &\le& 64 \frac{\alpha \gamma}{\sqrt{2\pi}}\,\left(Q_{\max} \frac{J}{\gamma} L\right) \frac{M}{\sigma^2}\, e^{- \frac{M}{2R}} \, r, \\ 
\|U_Y(\theta_x,\theta_y,r) - U_{Y(M)}(\theta_x,\theta_y,r)\| &\le& 64 \frac{\alpha \gamma}{\sqrt{2\pi}}\,\left(Q_{\max} \frac{J}{\gamma} L\right) \frac{M}{\sigma^2}\, e^{- \frac{M}{2R}} \, r.
\end{eqnarray}
\begin{proof}
We only prove the bound for $U_X(\theta_x,\theta_y,r)$, since the bound for $U_Y(\theta_x,\theta_y,r)$ follows from a symmetric argument. First, note that $\|U_X(\theta_x,\theta_y,r) - U_{X(M)}(\theta_x,\theta_y,r)\| = \|U^{\dagger}_X(\theta_x,\theta_y,r) U_{X(M)}(\theta_x,\theta_y,r) - \one\|$.
Moreover, setting $$\Delta_M(\theta_x,\theta_y, r) = \Sa^{(M)}(H(\theta_x+r,0,\theta_y,0), \partial_\theta H(\theta,0,\theta_y,0)_{\theta=\theta_x+r})-\Sa(H(\theta_x+r,0,\theta_y,0), \partial_\theta H(\theta,0,\theta_y,0)_{\theta=\theta_x+r})$$
we get after differentiating with respect to $r$:
\be
\partial_r \left(U^{\dagger}_X(\theta_x,\theta_y,r) U_{X(M)}(\theta_x,\theta_y,r)\right)_{r=s} = i\, \left(U^{\dagger}_X(\theta_x,\theta_y,s) \, \Delta_M(\theta_x,\theta_y, s)\, U_{X(M)}(\theta_x,\theta_y,s)\right)
\ee
and hence using a triangle inequality we have
\be
\|U^{\dagger}_X(\theta_x,\theta_y,r) U_{X(M)}(\theta_x,\theta_y,r) - \one\| \le \int_0^r \|\Delta_M(\theta_x,\theta_y,s)\| \, ds \le 64 \frac{\alpha \gamma}{\sqrt{2\pi}}\,\left(Q_{\max} \frac{J}{\gamma} L\right) \frac{M}{\sigma^2}\, e^{- \frac{M}{2R}} \, r,
\ee
where we used Lemma \ref{lem:sa_approx} to get the final inequality, which completes the proof.
\end{proof}
\end{lemma}
We now return to the proof of Lemma \ref{lem:local_loop_unitary}:
\begin{proof}[Proof of Lemma \ref{lem:local_loop_unitary}]
Using Lemma \ref{lem:local_unitaries} and several triangle inequalities, we may approximate the unitary defined in (\ref{def:unitary_loop_R}):
$$V_{\circlearrowleft}(\theta_x,\theta_y,r)= U_Y^{\dagger}(\theta_x,\theta_y,r)\, U_X^{\dagger}(\theta_x,\theta_y+r,r) \, U_Y(\theta_x+r,\theta_y,r) \,U_X(\theta_x,\theta_y,r)$$
with the following version:
$$V^{(M)}_{\circlearrowleft}(\theta_x,\theta_y,r)= U_{Y(M)}^{\dagger}(\theta_x,\theta_y,r)\, U_{X(M)}^{\dagger}(\theta_x,\theta_y+r,r) \, U_{Y(M)}(\theta_x+r,\theta_y,r) \,U_{X(M)}(\theta_x,\theta_y,r)$$
up to an exponentially small error.
More formally, we have for all $\theta_x,\theta_y \in [0,2\pi]$:
\be\label{loop_approximation}
\left\|V_{\circlearrowleft}(\theta_x,\theta_y,r) - V^{(M)}_{\circlearrowleft}(\theta_x,\theta_y,r)\right\| \le 256 \frac{\alpha \gamma}{\sqrt{2\pi}}\,\left(Q_{\max} \frac{J}{\gamma} L\right) \frac{M}{\sigma^2}\, e^{- \frac{M}{2R}} \, r.
\ee
We introduce now the following useful unitaries that will aid us in the perturbative expansion of $V^{(M)}_{\circlearrowleft}(\theta_x,\theta_y,r)$:
\begin{eqnarray*}
F_r(\theta_x,\theta_y,s_1,s_2) &=& U_{Y(M)}^{\dagger}(\theta_x,\theta_y,s_1)\, U_{X(M)}^{\dagger}(\theta_x,\theta_y+r,s_2) \, U_{Y(M)}(\theta_x+r,\theta_y,s_1) \,U_{X(M)}(\theta_x,\theta_y,s_2)\\
G_r(\theta_x,\theta_y,s_1,s_2) &=& U_{Y(M)}^{\dagger}(\theta_x+r,\theta_y,s_1)\, U_{X(M)}^{\dagger}(\theta_x,\theta_y+r,s_2) \, U_{Y(M)}(\theta_x+r,\theta_y,s_1) \,U_{X(M)}(\theta_x,\theta_y+r,s_2)
\end{eqnarray*}
where $0\le s_1, s_2 \le r$.
Note that $V^{(M)}_{\circlearrowleft}(\theta_x,\theta_y,r) = F_r(\theta_x,\theta_y,r,r)$ and that the following identities hold:
\begin{eqnarray}
F_r(\theta_x,\theta_y,s_1,s_2) &=& F_r(\theta_x,\theta_y,s_1,0) \, G_r(\theta_x,\theta_y,s_1,s_2) \, F_r(\theta_x,\theta_y,0,s_2)\label{eq:F_r}\\ 
F_r(\theta_x,\theta_y,0,0) &=& G_r(\theta_x,\theta_y,0,s_2) = G_r(\theta_x,\theta_y,s_1,0) = \one \nonumber\\
\partial_{s} G_r(\theta_x,\theta_y,0,s) &=& \partial_s G_r(\theta_x,\theta_y,s,0) = 0.\nonumber
\end{eqnarray}
Finally, we define the following operators: 
\begin{eqnarray*}
D^{(M)}_Y(\theta_x,\theta_y+s) &\equiv& \Sa^{(M)}(H(\theta_x,0,\theta_y+s,0), \partial_\theta H(\theta_x,0,\theta,0)_{\theta=\theta_y+s})\\
\Delta_Y(\theta_x,\theta_y,s_1, r) &\equiv& i\DY^{(M)}(\theta_x,\theta_y+s_1) - i\DY^{(M)}(\theta_x+r,\theta_y+s_1) =-i\,\int_0^r \left(\partial_{s'} \DY^{(M)}(\theta_x+s',\theta_y+s_1)_{s'=s}\right) \, ds\\
D^{(M)}_X(\theta_x,\theta_y+s) &\equiv& \Sa^{(M)}(H(\theta_x,0,\theta_y+s,0), \partial_\theta H(\theta_x,0,\theta,0)_{\theta=\theta_y+s})\\
\Delta_X(\theta_x,\theta_y,s_2, r) &\equiv& i\DX^{(M)}(\theta_x+s_2,\theta_y) -i\DX^{(M)}(\theta_x+s_2,\theta_y+r) =-i\,\int_0^r \left(\partial_{s'} \DX^{(M)}(\theta_x+s_2,\theta_y+s')_{s'=s}\right) \, ds.
\end{eqnarray*}
Note that the support of $\Delta_X(\theta_x,\theta_y,s_2, r)$ lies strictly within a $2M \times 4M$ box, centered at the origin in the $x-y$ orientation, whereas  the support of $\Delta_Y(\theta_x,\theta_y,s_2, r)$ lies strictly within a $4M \times 2M$ box, centered at the origin. This is because each partial derivative eliminates terms that do not depend on the variables of interest and moreover, interaction terms composing the above operators have been trimmed to have a radius of support $M$. Our goal is to show that up to order $4$ in $r$, the unitaries $F_r(\theta_x,\theta_y,s_1,0)$, $G_r(\theta_x,\theta_y,s_1,s_2)$ and $F_r(\theta_x,\theta_y,0,s_2)$ are localized within the set $\Omega_0$, defined in (\ref{omega_0}), and that the localized versions are rotations of the same operators evaluated at $\theta_x=\theta_y=0$.

We begin by showing localization for $F_r(\theta_x,\theta_y,s_1,0)$. A similar argument applies to $F_r(\theta_x,\theta_y,0,s_2)$. One can easily check that:
\be
\partial_s F_r(\theta_x,\theta_y,s,0)_{s=s_1} = U_{Y(M)}^{\dagger}(\theta_x,\theta_y,s_1)\, \Delta_Y(\theta_x,\theta_y,s_1, r)\, U_{Y(M)} (\theta_x,\theta_y,s_1) \cdot F_r(\theta_x,\theta_y,s_1,0)\nonumber
\ee
Noting that we only need to consider partials up to order $3$ since $\|\Delta_Y(\theta_x,\theta_y,s_1, r)\| = \mathcal{O}(r)$ (a discussion on how one may bound $\|\partial_{s'} \DY^{(M)}(\theta_x+s',\theta_y+s_1)_{s'=s}\|$ is given in the Appendix), we have:
\begin{eqnarray*}
\partial_s F_r(\theta_x,\theta_y,s,0)_{s=0} &=& \Delta_Y(\theta_x,\theta_y,0, r)\\
\partial^2_s F_r(\theta_x,\theta_y,s,0)_{s=0} &=& i \left[\Delta_Y(\theta_x,\theta_y,0, r), D^{(M)}_Y(\theta_x,\theta_y)\right] + \partial_s \Delta_Y(\theta_x,\theta_y,s, r)_{s=0} + \Delta^2_Y(\theta_x,\theta_y,0, r)\\
\partial^3_s F_r(\theta_x,\theta_y,s,0)_{s=0} &=& i^2 \left[ \left[\Delta_Y(\theta_x,\theta_y,0, r), D^{(M)}_Y(\theta_x,\theta_y)\right], D^{(M)}_Y(\theta_x,\theta_y)\right] \\
&+& i\, \partial_s \left[\Delta_Y(\theta_x,\theta_y,s, r), D^{(M)}_Y(\theta_x,\theta_y+s)\right]_{s=0}
+ \partial^2_s \Delta_Y(\theta_x,\theta_y,s, r)_{s=0} + \mathcal{O}(r^2).
\end{eqnarray*}
Now, the crucial point is that the operators $\Delta_Y(\theta_x,\theta_y,s, r)$ and $D^{(M)}_Y(\theta_x,\theta_y+s)$ satisfy for $M \le L/24$:
\begin{eqnarray*}
&&\Delta_Y(\theta_x,\theta_y,s, r) = R_Y(\theta_y,R_X(\theta_x,\Delta_Y(0,0,s, r))), \\
&&\left[\Delta_Y(\theta_x,\theta_y,s, r), D^{(M)}_Y(\theta_x,\theta_y+s) -R_Y\left(\theta_y,R_X\left(\theta_x,D^{(M)}_Y(0,s)\right)\right)\right] = 0\\
&&\left[R_Y\left(\theta_y,R_X\left(\theta_x,\left[\Delta_Y(0,0,s, r), D^{(M)}_Y(0,s)\right)\right]\right), D^{(M)}_Y(\theta_x,\theta_y+s) -R_Y\left(\theta_y,R_X\left(\theta_x,D^{(M)}_Y(0,s)\right)\right)\right] = 0
\end{eqnarray*}
which may be verified by considering the action of the twists $R_Y(\theta_y,\cdot)$ and $R_X(\theta_x,\cdot)$ on interaction terms with support contained in $\Omega_0$. For example, note that for the truncated Hamiltonians defined in (\ref{def:H_M}) with $M \le L/2$, we have $R_Y(\theta_y, H_M(Z; 0,0,s,0)) = H_M(Z; 0,0,\theta_y+s,0)$ for $Z$ with support near $y=1$ and $R_X(\theta_x, H_M(Z; s,0,0,0)) = H_M(Z; \theta_x+s,0,0,0))$ for $Z$ with support near $x=1$. This follows from the trivial action of the truncated Hamiltonians $H_M(Z; 0,0,s,0)$ and $H_M(Z; s,0,0,0)$ near the twists at $y=L/2+1$ and $x=L/2+1$, respectively, and the fact that each Hamiltonian is the sum of interaction terms that commute with charge operators whose support covers the interaction. In general, the previous argument can be applied to any Hamiltonian that acts trivially on one of the two boundary lines of the charge operators $Q_X$ and/or $Q_Y$, effectively introducing a \emph{single} boundary twist by applying the corresponding global twist $R_X(\theta_x,\cdot)$ and/or $R_Y(\theta_y,\cdot)$, on the whole Hamiltonian.
\begin{figure}
\centering
\includegraphics[width=300px]{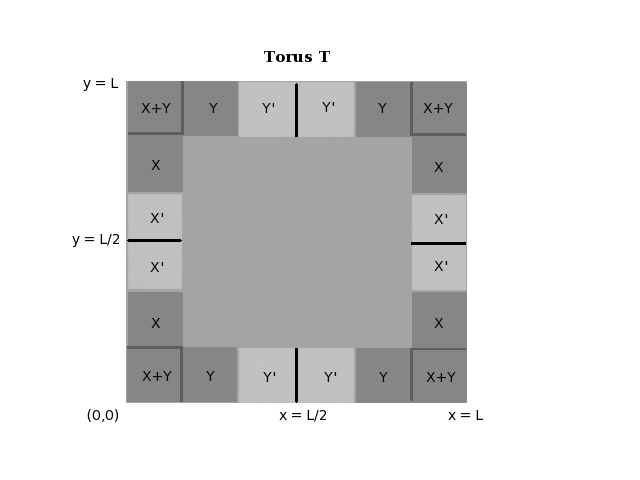}
\caption{\small{Decomposition of the torus $T$ into areas where the operators $\Sa^{(M)}(H(0,\theta_x,r,\theta_y),\partial_{\theta} H(\theta,0,r,\theta_y)|_{\theta=0})$ and $\Sa^{(M)}(H(r,\theta_x,0,\theta_y),\partial_{\theta} H(r,\theta_x,\theta,0)|_{\theta=0})$ act non-trivially. The dark shaded regions are defined as $X=\{s \in T: x(s) \in [-L/6,L/6] \, \wedge \, y(s) \in [-L/3-R, L/3+R] \}$ and $Y=\{s \in T: y(s) \in [-L/6,L/6] \, \wedge \, x(s) \in [-L/3-R, L/3+R] \}$ and correspond to interaction terms $\Sa^{(M)}(H(0,0,r,0),\partial_{\theta} \Phi(Z;\theta,0,r,0)|_{\theta=0})$ with $Z \subset X$ and $\Sa^{(M)}(H(r,0,0,0),\partial_{\theta} \Phi(Z;r,0,\theta,0)|_{\theta=0})$ with $Z \subset Y$, respectively. On the other hand, the light shaded regions $X'=\{s \in T: x(s) \in [-L/6,L/6] \, \wedge \, y(s) \in [2L/3, 4L/3] \}$ and $Y'=\{s \in T: y(s) \in [-L/6,L/6] \, \wedge \, x(s) \in [2L/3, 4L/3] \}$ correspond to terms $\Sa^{(M)}(H(0,0,0,\theta_y),\partial_{\theta} \Phi(Z;\theta,0,0,\theta_y)|_{\theta=0})$ with $Z \subset X'$ and $\Sa^{(M)}(H(0,\theta_x,0,0),\partial_{\theta} \Phi(Z;0,\theta_x,\theta,0)|_{\theta=0})$ with $Z \subset Y'$, respectively. Since each interaction has radius of support $M \le L/6$, interactions in $X'$ commute with interactions in $Y$ and $Y'$ and interactions in $Y'$ further commute with those in $X$.}}
\label{fig:space_decomp}
\end{figure}

Of course, in order for the previous argument to work, we need to show that each term we wish to rotate globally, has support away from the twists at $x=L/2+1$ and $y=L/2+1$. To study the support of the terms derived in the above partials, we choose one of these terms as an illustrative example. In particular, we note that taking commutators of $\Delta_Y(\theta_x,\theta_y,s, r)$ with $D^{(M)}_Y(\theta_x,\theta_y+s)$, can only extend the support of $\Delta_Y(\theta_x,\theta_y,s, r)$ (a $4M \times 2M$ box, centered at the origin) up to $2M$ in the $y$-direction, so terms affected by the twist on $x=L/2$ and $y=L/2$ will not appear in the commutator for $M\le L/12$, since then $(2+2)M\le L/2-2M$. Finally, since the support of $D^{(M)}_Y(\theta_x,\theta_y+s) -R_Y\left(\theta_y,R_X\left(\theta_x,D^{(M)}_Y(0,s)\right)\right)$ is away from the support of $\Delta_Y(\theta_x,\theta_y,s, r)$, the commutator of the two operators is trivially $0$.

The above observations, combined with the partials we computed and the assumption $0\le s_1, s_2 \le r$, imply that for a constant $C >0$, the following bounds hold:
\begin{eqnarray}\label{bound:F_r_1}
\|F_r(\theta_x,\theta_y,s_1,0) - R_Y(\theta_y,R_X(\theta_x, F_r(0,0,s_1,0)\| &\le& C \left(Q_{\max} \alpha J\, L\right)^5 r^5\\
\|F_r(\theta_x,\theta_y,0,s_2) - R_Y(\theta_y,R_X(\theta_x, F_r(0,0,0,s_2)\| &\le& C \left(Q_{\max} \alpha J\, L\right)^5 r^5,\label{bound:F_r_2}
\end{eqnarray}
with the constant power of $(Q_{\max} \alpha J\, L)$ coming from bounds on the norms of fifth-order commutators containing terms like $D^{(M)}_Y(\theta_x,\theta_y+s_1)$ and partials like $\partial_{s'} \DY^{(M)}(\theta_x+s',\theta_y+s_1)_{s'=s}$, whose norms may be bounded using (\ref{naive_bnd:sa}) and arguments similar to the one found in the Appendix, respectively.

Moreover, the supports of the partial derivatives  of $F_r(0,0,s_1,0)$ and $F_r(0,0,0,s_2)$ up to order $4$ in $r$, lie strictly within $\Omega_0$, for $M \le L/48$. By taking $M$ to be a smaller fraction of $L$, we could have continued the expansion to higher orders in $r$, while keeping the supports of $F_r(0,0,s_1,0)$ and $F_r(0,0,0,s_2)$ strictly within $\Omega_0$. For our purposes, it suffices to consider errors of order $\mathcal{O}(r^5)$.

We turn our focus, now, to showing the bound:
\begin{eqnarray}\label{bound:G_r}
\|G_r(\theta_x,\theta_y,s_1,s_2) - R_Y(\theta_y,R_X(\theta_x, G_r(0,0,s_1,s_2)))\| &\le& C \left(Q_{\max} \alpha J\, L\right)^5 r^5
\end{eqnarray}
for some constant $C > 0$, while keeping the support of $G_r(0,0,s_1,s_2)$ within $\Omega_0$, up to order $4$ in $r$ (recalling that $0 \le s_1,s_2 \le r$).
We begin by setting:
\begin{eqnarray*}
\Delta^{r}_{Y}(\theta_x,\theta_y,s_1, s_2) &\equiv& i\, U_{X(M)}^{\dagger}(\theta_x,\theta_y+r,s_2)\, D^{(M)}_Y(\theta_x+r,\theta_y+s_1) \, U_{X(M)} (\theta_x,\theta_y+r,s_2) - i\, D^{(M)}_Y(\theta_x+r,\theta_y+s_1)\\  
&=& \int_0^{s_2} U_{X(M)}^{\dagger}(\theta_x,\theta_y+r,s)\left[D^{(M)}_X(\theta_x+s,\theta_y+r),D^{(M)}_Y(\theta_x+r,\theta_y+s_1)\right] U_{X(M)} (\theta_x,\theta_y+r,s) \, ds
\\
\Delta^{r}_{X}(\theta_x,\theta_y,s_1, s_2) &\equiv& i\, D^{(M)}_X(\theta_x+s_2,\theta_y+r) -i\, U_{Y(M)}^{\dagger}(\theta_x+r,\theta_y,s_1)\, D^{(M)}_X(\theta_x+s_2,\theta_y+r) \, U_{Y(M)} (\theta_x+r,\theta_y,s_1)\\ 
&=& \int_0^{s_1} U_{Y(M)}^{\dagger}(\theta_x+r,\theta_y,s)\left[D^{(M)}_X(\theta_x+s_2,\theta_y+r),D^{(M)}_Y(\theta_x+r,\theta_y+s)\right] U_{Y(M)} (\theta_x+r,\theta_y,s) \, ds
\end{eqnarray*}
Taking partials with respect to the variables $s_1$ and $s_2$, one may verify:
\begin{eqnarray*}
\partial_{s_1'} G_r(\theta_x,\theta_y,s_1',s_2)_{s_1'=s_1} &=& G_r(\theta_x,\theta_y,s_1,s_2) \cdot \left(U_{Y(M)}^{\dagger}(\theta_x+r,\theta_y,s_1)\, \Delta^{r}_{Y}(\theta_x,\theta_y,s_1, s_2)\, U_{Y(M)} (\theta_x+r,\theta_y,s_1)\right) \\
\partial_{s_2'} G_r(\theta_x,\theta_y,s_1,s_2')_{s_2'=s_2} &=& \left(U_{X(M)}^{\dagger}(\theta_x,\theta_y+r,s_2)\, \Delta^{r}_{X}(\theta_x,\theta_y,s_1, s_2)\, U_{X(M)} (\theta_x,\theta_y+r,s_2)\right) \cdot G_r(\theta_x,\theta_y,s_1,s_2)
\end{eqnarray*}
Using the above relations, it is straightforward to check:
\begin{eqnarray*}
\partial_{s_1} G_r(\theta_x,\theta_y,s_1,s_2)_{s_1=0} &=& \Delta^{r}_{Y}(\theta_x,\theta_y,0, s_2)\\
\partial_{s_2} G_r(\theta_x,\theta_y,s_1,s_2)_{s_2=0} &=& \Delta^{r}_{X}(\theta_x,\theta_y,s_1, 0)\\
\partial^2_{s_1} G_r(\theta_x,\theta_y,s_1,s_2)_{s_1=0} &=&
i\left[\Delta^{r}_{Y}(\theta_x,\theta_y,0, s_2),D^{(M)}_Y(\theta_x+r,\theta_y)\right] + \left(\partial_{s_1} \Delta^{r}_{Y}(\theta_x,\theta_y,s_1, s_2)\right)_{s_1=0} + (\Delta^{r}_{Y}(\theta_x,\theta_y,0, s_2))^2\\
\partial^2_{s_2} G_r(\theta_x,\theta_y,s_1,s_2)_{s_2=0} &=&
i\left[\Delta^{r}_{X}(\theta_x,\theta_y,s_1, 0),D^{(M)}_X(\theta_x,\theta_y+r)\right] + \left(\partial_{s_2} \Delta^{r}_{X}(\theta_x,\theta_y,s_1, s_2)\right)_{s_2=0} + (\Delta^{r}_{X}(\theta_x,\theta_y,s_1,0))^2
\end{eqnarray*}
We note here that we only need to evaluate the following partials at $s_1=s_2=0$ in order to have a complete picture of $G_r(\theta_x,\theta_y,s_1,s_2)$ up to order $4$ in $r$:
\begin{eqnarray*}
\partial_{s_2} (\partial_{s_1} G_r(\theta_x,\theta_y,s_1,s_2)_{s_1=0})_{s_2=0} &=& \partial_{s_2} \Delta^{r}_{Y}(\theta_x,\theta_y,0, s_2)_{s_2=0} = \left[D^{(M)}_X(\theta_x,\theta_y+r),D^{(M)}_Y(\theta_x+r,\theta_y)\right]\\
\partial_{s_2} \left(\partial^2_{s_1} G_r(\theta_x,\theta_y,s_1,s_2)_{s_1=0}\right)_{s_2=0} &=& i \left[\left[D^{(M)}_X(\theta_x,\theta_y+r),D^{(M)}_Y(\theta_x+r,\theta_y)\right],D^{(M)}_Y(\theta_x+r,\theta_y)\right] \\
&+& \left[D^{(M)}_X(\theta_x,\theta_y+r), \left(\partial_{s_1} D^{(M)}_Y(\theta_x+r,\theta_y+s_1)\right)_{s_1=0}\right]\\
\partial_{s_1} \left(\partial^2_{s_2} G_r(\theta_x,\theta_y,s_1,s_2)_{s_2=0}\right)_{s_1=0} &=& i \left[\left[D^{(M)}_X(\theta_x,\theta_y+r),D^{(M)}_Y(\theta_x+r,\theta_y)\right],D^{(M)}_X(\theta_x,\theta_y+r)\right] \\
&+& \left[\left(\partial_{s_2} D^{(M)}_X(\theta_x+s_2,\theta_y+r)\right)_{s_2=0}, D^{(M)}_Y(\theta_x+r,\theta_y)\right]\\
\partial^2_{s_1} \left(\partial^2_{s_2} G_r(\theta_x,\theta_y,s_1,s_2)_{s_2=0}\right)_{s_1=0}  &=& i
\left[\left(\partial^2_{s_1} \Delta^{r}_{X}(\theta_x,\theta_y,s_1, 0)\right)_{s_1=0},D^{(M)}_X(\theta_x,\theta_y+r)\right] \\
&+& \left(\partial_{s_2} \left(\partial^2_{s_1} \Delta^{r}_{X}(\theta_x,\theta_y,s_1, 0)\right)_{s_1=0}\right)_{s_2=0}
+ 2 \left[D^{(M)}_X(\theta_x,\theta_y+r),D^{(M)}_Y(\theta_x+r,\theta_y)\right]^2
\end{eqnarray*}
Now, since $\left(\partial^2_{s_1} \Delta^{r}_{X}(\theta_x,\theta_y,s_1, 0)\right)_{s_1=0}= \partial_{s_2} \left(\partial^2_{s_1} G_r(\theta_x,\theta_y,s_1,s_2)_{s_1=0}\right)_{s_2=0},$
we can see from Figure \ref{fig:space_decomp} that all of the above commutators are rotated versions of the same commutators with $\theta_x=\theta_y=0$, where we apply $R_Y(\theta_y, R_X(\theta_x, \cdot))$ to perform the rotation as in the study of $F_r(\theta_x,\theta_y,s_1,0)$. For example, we have for one of the fourth order terms implicit in the first term of the last line of partials evaluated above:
\begin{eqnarray*}
&&\left[\left[\left[D^{(M)}_X(\theta_x,\theta_y+r),D^{(M)}_Y(\theta_x+r,\theta_y)\right],D^{(M)}_Y(\theta_x+r,\theta_y)\right], D^{(M)}_X(\theta_x,\theta_y+r)\right] = \\
&&R_Y\left(\theta_y,R_X\left(\theta_x, \left[\left[\left[D^{(M)}_X(0,r),D^{(M)}_Y(r,0)\right],D^{(M)}_Y(r,0)\right], D^{(M)}_X(0,r)\right]\right)\right)
\end{eqnarray*}
where we have assumed that $M \le L/12$ to make sure that all terms affected by the twists at $x=y=L/2$ vanish. Furthermore, choosing $M = L/48$, guarantees that up to fourth order all commutators have support within $\Omega_0$. To see this, note that each new commutator includes terms with support at most $2M$ away from the support of each term in the commutator. For example, the above commutator has, potentially, the largest support of all the terms up to fourth order. It is supported within a cross of radius $5M$, with two axis of width $2M$ each, centered at the origin.
Now, using (\ref{loop_approximation}) with $V^{(M)}_{\circlearrowleft}(\theta_x,\theta_y,r) = F_r(\theta_x,\theta_y,r,r)$, we get from (\ref{eq:F_r}-\ref{bound:G_r}) and several triangle inequalities:
\be
\|V_{\circlearrowleft}(\theta_x,\theta_y,r)- R_Y(\theta_y, R_X(\theta_x, V_{\circlearrowleft}(0,0,r)))\| \le C\left(\frac{\alpha \gamma}{\sqrt{2\pi}}\,\left(Q_{\max} \frac{J}{\gamma} L\right) \frac{L}{\sigma^2}\, e^{- \frac{L}{96R}} \, r + \left(Q_{\max} \alpha J\, L\right)^5 r^5\right)
\nonumber
\ee
Moreover, setting $W(r)$ to be the sum of terms up to order $4$ in $r$ in the Taylor expansion of the operator $V^{M}_{\circlearrowleft}(0,0,r)$ and noting that $W(0)= \one$ and $\partial_r W(r)_{r=0} = 0$, we have $\|W(r)-\one\| \le C \left(Q_{\max} \alpha J L\right)^2 r^2$, for some constant $C >0$, which completes the proof of this Lemma.
\end{proof}

Now, we come back to the proof of Lemma \ref{lem:translation}, which implies Corollary \ref{cor:stokes}. 
\begin{proof}[Proof of Lemma \ref{lem:translation}]
Using Lemmas \ref{lem:translation_unitary} and \ref{lem:local_loop_unitary} with $A_{\Omega_0} = W(r)-\one$, we get the bound (\ref{eq:stokes}) from the following estimate:
\begin{flalign*}
&|\braket{\Psi_0}{\Psi_{\circlearrowleft}(\theta_x,\theta_y,r)} -\braket{\Psi_0}{\Psi_{\circlearrowleft}(r)}| \le
|\braket{\Psi_0}{W(r) \Psi_0} - \braket{\Psi_0}{\Psi_{\circlearrowleft}(r)}|\\
&+\left|\braket{\Psi_0}{W(r) \Psi_0}- \braket{\Psi_0}{V^{\dagger}[(0,0)\rightarrow(\theta_x,\theta_y)]\,R_Y(\theta_y,R_X(\theta_x,W(r)))\, V[(0,0)\rightarrow(\theta_x,\theta_y)]  \Psi_0}\right|\\
&+\left|\braket{\Psi_0}{V^{\dagger}[(0,0)\rightarrow(\theta_x,\theta_y)]\,R_Y(\theta_y,R_X(\theta_x,[W(r)-V_{\circlearrowleft}(0,0,r)]))\, V[(0,0)\rightarrow(\theta_x,\theta_y)]  \Psi_0}\right|\\
&+\left|\braket{\Psi_0}{V^{\dagger}[(0,0)\rightarrow(\theta_x,\theta_y)]\,\left[V_{\circlearrowleft}(\theta_x,\theta_y,r)-R_Y(\theta_y,R_X(\theta_x,V_{\circlearrowleft}(0,0,r)))\right]\, V[(0,0)\rightarrow(\theta_x,\theta_y)]  \Psi_0}\right| \\
&\le 2\|W(r) -V_{\circlearrowleft}(0,0,r)\| + \|V_{\circlearrowleft}(\theta_x,\theta_y,r)-R_Y(\theta_y,R_X(\theta_x,V_{\circlearrowleft}(0,0,r)))\|\\
&+\left|\braket{\Psi_0}{(W(r)-\one) \Psi_0}- \braket{\Psi_0}{V^{\dagger}[(0,0)\rightarrow(\theta_x,\theta_y)]\,R_Y(\theta_y,R_X(\theta_x,(W(r)-\one)))\, V[(0,0)\rightarrow(\theta_x,\theta_y)]  \Psi_0}\right|,
\end{flalign*}
recalling the bound $\|W(r)-\one\| \le C \left(Q_{\max} \alpha J L\right)^2 r^2$, for some constant $C >0$.
\end{proof}

\section{Proof of main theorem}
At this point, we can put everything together to prove the main result, the Quantization of the Hall Conductance:
\begin{proof}[Proof of the Main Theorem]
Going back to (\ref{main_bound}), we use the bounds derived in Corollary \ref{cor:adiabatic_phase_3}, Corollary \ref{cor:gs_evol} and Corollary \ref{cor:stokes}, to complete the proof,
noting that terms which are exponentially small in $\alpha^2\gamma^2$ dominate terms which are
exponentially small in $L$. 
Finally, we set $G_{R,J,\gamma}(L) = \alpha \gamma$ with $\alpha$ given in the statement of Lemma \ref{lem:local_approximation}. 
\end{proof}
\section{Discussion and extensions}
\subsection{Hall conductance as an obstruction}

Given two different gapped Hamiltonians, $H_0$ and $H_1$, it is natural to ask whether there is a path in parameter space $\{H(s)\}_{s\in[0,1]}$ which
connects the two Hamiltonians such that $H(s)$ obeys the same upper bounds $R,J$ on the interaction range and strength, and given some lower bound on the spectral gap $\gamma(s)$ of $H(s)$, for $s \in [0,1]$.
As a corollary of our main theorem, we now show that there is an obstruction to
finding such a path connecting two Hamiltonians with different Hall conductance, unless the gap becomes polynomially small in $L$, for some $s \in [0,1]$.

Suppose that $H_0$, $H_1$ both have interaction range $R$ and an energy ratio $J/\gamma$ of order unity, with $L>>1$.  Then, the main theorem implies that the
Hall conductance is very close to an integer multiple of $e^2/h$ in both cases.  However, if the given integers, $n_0,n_1$ differ,
then for any path, for some $s\in [0,1]$, the Hall conductance must equal $(n_0 \pm 1/2)\frac{e^2}{h}
$.

However, as a corollary of the main theorem, we have:
\begin{cor}
For any given $R$ and $q_{max}$, there exists a constant $c$ such that for
$L\geq c (J/\gamma)^5 \ln^{11/2}(J/\gamma)$
\be
\left|\sigma_{xy} - n \cdot\frac{e^2}{h}\right| < \frac{e^2}{2h}
\ee
\begin{proof}
For the given $L$, $G_{R,J,\gamma}(L) \sim (\gamma/J) [(J/\gamma)^{1/5} \ln^{11/2}(J/\gamma)/\ln^{3}(J/\gamma)]^{1/5}\sim [\ln^{5/2}(J/\gamma)]^{1/5}\sim
\ln^{1/2}(J/\gamma)$.  Thus, $\exp(-G^2_{R,J,\gamma}/6)$ is inverse polynomial in $L$, and by appropriate choice of $c$ the desired result follows.
\end{proof}
\end{cor}
Thus, if $n_0\neq n_1$, for some $s$, the gap $\gamma(s)$ obeys
$\gamma(s) < c' J L^{-1/5}\cdot \ln^{11/10}(L)$, for some $c'$ which depends
on $R,J$.

\subsection{Fiber bundles}
We now relate our approach to Chern numbers.  First, recall how the Chern number
argument for the Hall conductance worked.  Under the assumption that the
gap is non-vanishing everywhere, the Hamiltonian has a unique ground
state density matrix $|\Psi_0(\theta_x,\theta_y)\rangle\langle\Psi_0(\theta_x,\theta_y)|$ for all $\theta_x,\theta_y$.  This ground state density
matrix is a rank-$1$ projection in a very high dimensional space (the
Hilbert space of the whole system has dimension exponentially large in $L^2$).

Given this family of projections,
they define a complex vector bundle, with the torus as
the base space and $\mathbb{C}$ as the fiber.  The connection given by
adiabatic continuation is the
canonical connection.

One could also consider a different bundle, with the larger Hilbert
space (that is, not just the ground state subspace, but all higher energy eigenspaces included) as the fiber.  This bundle is trivial, being the direct product
of the larger Hilbert space with the torus.
In this case, if there is a unique ground state for all $\theta_x,\theta_y$, then
adiabatic continuation gives a connection
that has the property of preserving the rank-$1$ projection onto the
ground state.

We now show that we can find a rank-$1$ projection that will be
approximately preserved by quasi-adiabatic continuation.  We do this using general error terms ($\err$ below)
so that the results can be applied to other quasi-adiabatic continuations.
Our assumptions below correspond to assumptions that different paths
starting at $(0,0)$ and ending at $(\theta_x,\theta_y)$ give approximately
the same state; the first assumption bounds the difference for two
different paths, while the second bounds the change in the state for
a slight change in the path, as shown in Figure \ref{paths}.

\begin{figure}
\centering
\includegraphics[width=300px]{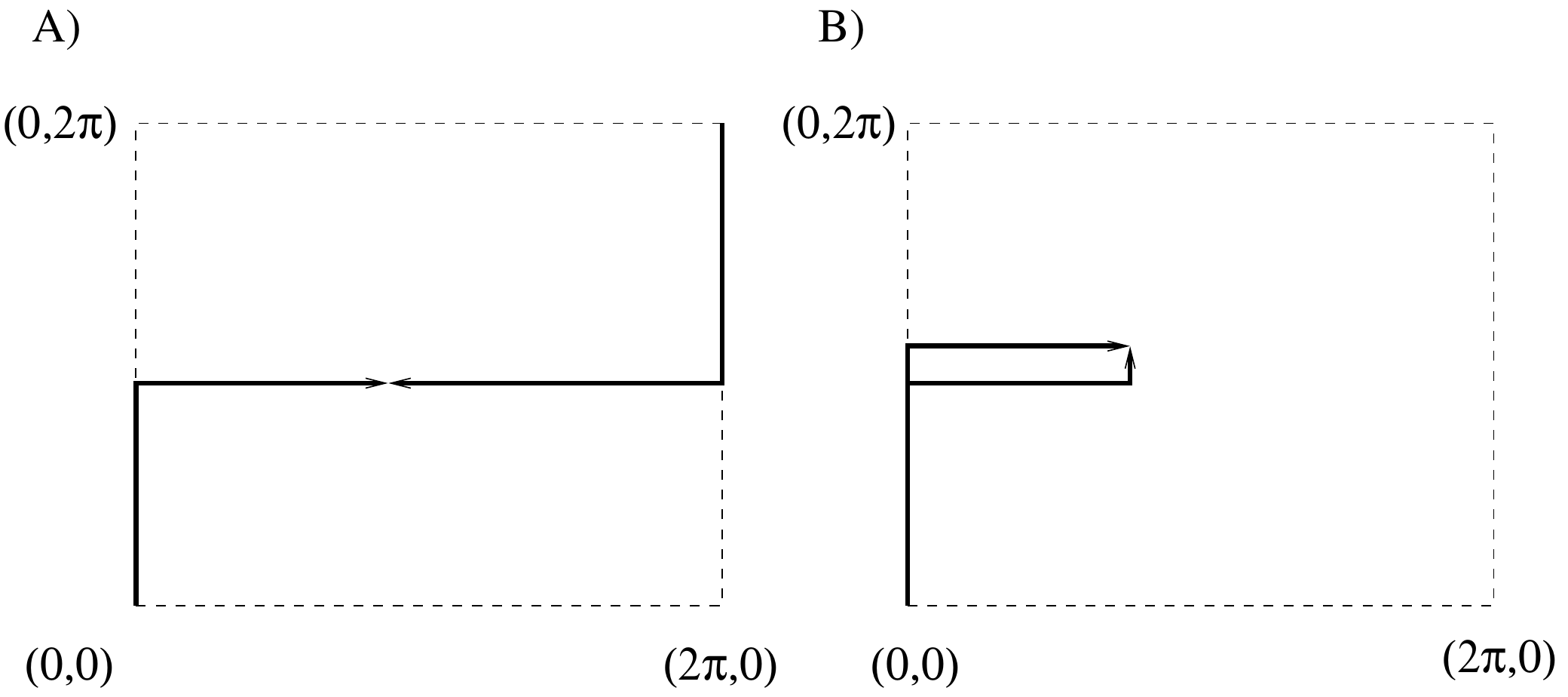}
\caption{\small{Paths on the flux torus.  Dashed lines are used for
the boundaries of the torus, and solid lines for paths.
a) An example of two different paths that are compared in
the first assumption. The first assumption states that the resulting projector
is almost the same for both paths.  b) Another example of two different paths,
one of which travels vertically and then horizontally, while the other
moves a slightly shorter distance vertically, then moves horizontally,
then moves a short distance vertically to arrive at the same point.
The second assumption describes a differential form of the assumption that
these two paths give almost the same projector.
}}
\label{paths}
\end{figure}

We obtain the desired projector by averaging over different paths; the reason this is necessary is that if we choose just one
particular path from $(0,0)$ to obtain the projector at $(\theta_x,\theta_y)$, for example the path first increasing $\theta_y$ and then
increasing $\theta_x$, this would give a discontinuity at $\theta_x=\theta_y=2\pi$.
\begin{lemma}
Suppose $\DX(\theta_x,\theta_y),\DY(\theta_x,\theta_y)$ are quasi-adiabatic continuation operators and $\Psi_0$ is a state 
such that the following
conditions hold:
\begin{enumerate}
\item For any $\theta_x,\theta_y \in [0,2\pi]$, let $\Lambda^a$, for $a=1,2,3,4$ denote four different paths on the flux torus from $(0,0)$ to
$(\theta_x,\theta_y)$ as follows.
We can move from $(0,0)$ to $(0,\theta_y)$ by
either increasing $\theta_y$ from zero or by decreasing it from $2\pi$, and then we move from $(0,\theta_y)$ to
$(\theta_x,\theta_y)$ again by either increasing $\theta_x$ from zero or by decreasing it from $2\pi$.
The paths with $a=1,2$ correspond to the paths in which
we increase $\theta_y$, while those with $a=3,4$ correspond to those with decreasing $\theta_y$, while
the paths with $a=1,3$ correspond to those with increasing $\theta_x$ and those with $a=2,4$ correspond to
those with decreasing $\theta_x$.
Let $U_{a}$, for $a=1,2,3,4$ denote the unitary evolution using quasi-adiabatic continuation along each of the corresponding paths.
Then, we require that
\be
|\langle \Psi_0,U_{a}^{\dagger} U_b \Psi_0 \rangle|^2\geq 1-\err^2 ,
\ee
for all $a,b$, for sufficiently small $\err$.

\item Next, let $U_Y^a(\theta_x,\theta_y)$ denote the unitary evolution corresponding to the vertical segment of path $U^a$ (the
evolution in $\theta_y$), while $U_X^a(\theta_x,\theta_y)$ denotes the unitary evolution corresponding to the horizontal segment
(the evolution in $\theta_x)$.  Note that $U_Y^a$ depends only
on $\theta_y$, not on $\theta_x$.  Then, we require that
\be
\|i\left\{\DY(\theta_x,\theta_y)U_X^a(\theta_x,\theta_y)-U_X^a(\theta_x,\theta_y)\DY(0,\theta_y) U_Y^a(\theta_x,\theta_y)\right\}|\Psi_0\rangle\langle \Psi_0|U^a(\theta_x,\theta_y)+h.c.\|\leq \err.
\ee

\item The operators $\DX(\theta_x,\theta_y),\DX(\theta_x,\theta_y)$
are differentiable with respect to $\theta_x,\theta_y$ everywhere and
are periodic in $\theta_x,\theta_y$ with period $2\pi$.  The norms
$\| \DX(\theta_x,\theta_y) \|$, $\| \DY(\theta_x,\theta_y) \|$ are bounded
by some constant $D$.
\end{enumerate}
Then,
there exists a rank-$1$ projector, $P(\theta_x,\theta_y)$ which is periodic in $\theta_x,\theta_y$ with
period $2\pi$ such that
\be
\| P(0,0)-|\Psi_0\rangle\langle\Psi_0| \| \leq C\err,
\ee
and such that
\be
\|i[\DX(\theta_x,\theta_y),P(\theta_x,\theta_y)]-\partial_{\theta_x} P(\theta_x,\theta_y)\| \leq
C(1+D)\err.
\ee
and
\be
\|i[\DY(\theta_x,\theta_y),P(\theta_x,\theta_y)]-\partial_{\theta_y} P(\theta_x,\theta_y)\| \leq
C(1+D)\err.
\ee
\begin{proof}
For $a=1,2,3,4$ let $P(\theta_x,\theta_y)^{a}=U_a|\Psi_0\rangle\langle\Psi_0|U_a^{\dagger}$.
By the first assumption,
\be
\label{by1}
{\rm Tr}(P(\theta_x,\theta_y)^a P(\theta_x,\theta_y)^b)\geq 1-\err^2.
\ee

For $a \in \{1,2,3,4\}$, and $\theta_x,\theta_y\in [0,2\pi]$, let $w^a(\theta_x,\theta_y)=w^a_x(\theta_x) w^a_y(\theta_y)$ denote the weight of the path $\Lambda_a$, where $w^a_y(\theta_y)$ is equal to $\theta_y/2\pi$ if the path initially moves by decreasing $\theta_y$ from $2\pi$ and is equal to
$1-\theta_y/2\pi$ if the path instead initially moves by increasing $\theta_y$ from $0$ and
$w^a_x(\theta_x)$ is equal to $\theta_x/2\pi$ if the path then moves by decreasing $\theta_x$ from $2\pi$ and is equal to $1-\theta_x/2\pi$ if the path instead moves by increasing $\theta_x$ from $0$.

For $\theta_x,\theta_y\in [0,2\pi]$,
let $O(\theta_x,\theta_y)=\sum_{a=1}^4 w^a(\theta_x,\theta_y) P(\theta_x,\theta_y)^a$.

Since $O(2\pi,\theta_y)=O(0,\theta_y)$ and $O(\theta_x,0)=O(\theta_x,2\pi)$,
we can extend the definition of $O(\theta_x,\theta_y)$ to a periodic function for
all $\theta_x,\theta_y$.  The operator $O$ is differentiable everywhere
except when $\theta_x$ or $\theta_y$ is an integer multiple of $2\pi$.
In the rest of the proof, whenever we evaluate a partial derivative of
$O(\theta_x,\theta_y)$, our bounds are valid for all $\theta_x,\theta_y$ except for
the lines where $\theta_x$ or $\theta_y$ is an integer multiple of $2\pi$.

First, we bound the norm of
$i[\DX(\theta_x,\theta_y),O(\theta_x,\theta_y)]-\partial_{\theta_x} O(\theta_x,\theta_y)$.
Note that for each $a$, $$i[\DX(\theta_x,\theta_y),P(\theta_x,\theta_y)^a]=\partial_{\theta_x} P(\theta_x,\theta_y)^a.$$
Thus,
\begin{eqnarray*}
i[\DX(\theta_x,\theta_y),O(\theta_x,\theta_y)]-\partial_{\theta_x} O(\theta_x,\theta_y)=
-\sum_{a=1}^4 (\partial_{\theta_x} w^a(\theta_x,\theta_y)) P(\theta_x,\theta_y)^a=\\
\left(\frac{1}{2\pi}-\frac{\theta_y}{4\pi^2}\right)\left(P(\theta_x,\theta_y)^{2}-P(\theta_x,\theta_y)^{1}\right)+\frac{\theta_y}{4\pi^2}\left(P(\theta_x,\theta_y)^{4}-P(\theta_x,\theta_y)^{3}\right).
\end{eqnarray*}
Using $\|A\| \le \|A\|_2 = \sqrt{\Tr(A^2)}$, true for any Hermitian operator $A$, we have:
\begin{flalign*}
&\left\|i[\DX(\theta_x,\theta_y),O(\theta_x,\theta_y)]-\partial_{\theta_x} O(\theta_x,\theta_y)\right\|
\le \left(\frac{1}{2\pi}-\frac{\theta_y}{4\pi^2}\right) \left\|P(\theta_x,\theta_y)^{1}-P(\theta_x,\theta_y)^{2}\right\| + \frac{\theta_y}{4\pi^2} \left\|P(\theta_x,\theta_y)^{3}-P(\theta_x,\theta_y)^{4}\right\|\\
&\le \frac{1}{2\pi} \max\left\{\left\|P(\theta_x,\theta_y)^{1}-P(\theta_x,\theta_y)^{2}\right\|_2,\left\|P(\theta_x,\theta_y)^{3}-P(\theta_x,\theta_y)^{4}\right\|_2\right\} \le \frac{\sqrt{2}}{2\pi}\,\err,
\end{flalign*}
where we used (\ref{by1}) to get $\max\left\{\left\|P(\theta_x,\theta_y)^{1}-P(\theta_x,\theta_y)^{2}\right\|_2^2,\left\|P(\theta_x,\theta_y)^{3}-P(\theta_x,\theta_y)^{4}\right\|_2^2\right\} \le 2\err^2$, for the last bound.

We next evaluate
$i[\DY(\theta_x,\theta_y),O(\theta_x,\theta_y)]-\partial_{\theta_y} O(\theta_x,\theta_y)$.
Now the difference between quasi-adiabatic continuation and partial derivative,
namely $i[\DY(\theta_x,\theta_y),P(\theta_x,\theta_y)^a]-\partial_{\theta_y} P(\theta_x,\theta_y)^a$, is non-vanishing.
The operator norm of this difference is equal to $\|\left\{i\DY(\theta_x,\theta_y)U^a(\theta_x,\theta_y)-\partial_{\theta_y}U^a(\theta_x,\theta_y)\right\}|\Psi_0\rangle\langle\Psi_0|U^a(\theta_x,\theta_y)+h.c.\|$.
This is equal to 
$\|i\left\{\DY(\theta_x,\theta_y)U_X^a(\theta_x,\theta_y)-U_X^a(\theta_x,\theta_y)\DY(0,\theta_y) U_Y^a(\theta_x,\theta_y)\right\}|\Psi_0\rangle\langle|U^a(\theta_x,\theta_y)+h.c.\|\leq C \err$.
So,
$\| i[\DY(\theta_x,\theta_y),O(\theta_x,\theta_y)]-\partial_{\theta_y} O(\theta_x,\theta_y)\| \le C\err
+(1/2\pi){\rm max}\{\|P(\theta_x,\theta_y)^{1}-P(\theta_x,\theta_y)^{2}\|,\|P(\theta_x,\theta_y)^{3}-P(\theta_x,\theta_y)^{4}\|\}$, which is
bounded by $C\err$.

Given the operator $O(\theta_x,\theta_y)$, we can define a smoothed
operator by convolving it:
\be
\tilde O(\theta_x,\theta_y)\equiv \int \int O(\phi_x,\phi_y) k(\phi_x-\theta_x,\phi_y-\theta_y) d\phi_x d\phi_y,
\ee
where $k(x,y)$ has support in a small square near $x=y=0$, $k$ is positive,
and the integral of $k$ is equal to unity.  For differentiable $\DX,\DY$,
by choosing the size of the
support of $k$ sufficiently small, we can construct an operator $\tilde O(\theta_x,\theta_y)$ that is everywhere differentiable, is periodic, and
obeys the bounds
\be
\| i[\DX(\theta_x,\theta_y),\tilde O(\theta_x,\theta_y)]-\partial_{\theta_x} \tilde O(\theta_x,\theta_y) \| \leq C\err,
\ee
\be
\| i[\DY(\theta_x,\theta_y),\tilde O(\theta_x,\theta_y)]-\partial_{\theta_y} \tilde O(\theta_x,\theta_y) \| \leq C\err,
\ee
\be
\| \tilde O(\theta_x,\theta_y)-O(\theta_x,\theta_y)\|\leq C\err,
\ee
for all $\theta_x,\theta_y$.

We now define $P(\theta_x,\theta_y)$ to be the projector onto the eigenvector of $\tilde O(\theta_x,\theta_y)$ with the largest
eigenvalue.  By (\ref{by1}), we have ${\rm Tr}(\tilde O(\theta_x,\theta_y) P(\theta_x,\theta_y)^a)\geq 1-C\err$
for $a \in \{1,2,3,4\}$.
Hence, by definition the eigenvector $P(\theta_x,\theta_y)$ has eigenvalue greater than or equal to $1-C\err$.  
This eigenvector is unique for sufficiently small $\err$ and
the next largest eigenvalue is less than or equal to $C\err$.
We now relate these bounds to bounds on
$\|i[\DY(\theta_x,\theta_y),P(\theta_x,\theta_y)]-\partial_{\theta_y} P(\theta_x,\theta_y)\|$.
For notational convenience, we write $\tilde O$ for $\tilde O(\theta_x,\theta_y)$ and $P$ for $P(\theta_x,\theta_y)$ when no derivative is taken.
Since $\|P-\tilde O\|\leq C\err$, we have
$$\|[\DY(\theta_x,\theta_y),P(\theta_x,\theta_y)]-[\DY(\theta_x,\theta_y),\tilde O(\theta_x,\theta_y)]\|\leq 2\|\DY(\theta_x,\theta_y)\|C \err \leq
DC\err.$$
Let 
$\lambda(\theta_x,\theta_y)$ be the largest eigenvalue of
$\tilde O(\theta_x,\theta_y)$ and $v(\theta_x,\theta_y)$ be the corresponding eigenvector so that $P=\lambda(\theta_x,\theta_y)^{-1}|v(\theta_x,\theta_y)\rangle\langle v(\theta_x,\theta_y)|$.  Then, we have
\begin{eqnarray}
(1-P) \partial_{\theta_y} v(\theta_x,\theta_y)&=& \frac{1}{\lambda(\theta_x,\theta_y)-\tilde O} (1-P) \left( \partial_{\theta_y} \tilde O(\theta_x,\theta_y) \right) v(\theta_x,\theta_y),
\end{eqnarray}
Then, since all eigenvalues of $\tilde O$ other than $\lambda$ are close
to zero, we have that $\|(\lambda(\theta_x,\theta_y)-\tilde O)^{-1} (1-P) - (1-P) \|
\leq C\err$.  Thus,
$$\|\partial_{\theta_y} P(\theta_x,\theta_y)-\lambda(\theta_x,\theta_y)^{-1}\left\{ (1-P) \partial_{\theta_y} \tilde O(\theta_x,\theta_y) P+P\partial_{\theta_y} \tilde O(\theta_x,\theta_y) (1-P)\right\}\| \leq C\err.$$
Since $\partial_{\theta_y} \tilde O(\theta_x,\theta_y)$ is close to $i[\DY(\theta_x,\theta_y),\tilde O(\theta_x,\theta_y)]$ which in
turn is close to $i[\DY(\theta_x,\theta_y),P(\theta_x,\theta_y)]$, and since $\lambda^{-1}$ is close to unity, we find that
\be
\|\partial_{\theta_y} P(\theta_x,\theta_y)-i\left\{ (1-P) [\DY(\theta_x,\theta_y),P] P+P[\DY(\theta_x,\theta_y), P] (1-P)\right\} \| \leq C(1+D)\err.
\ee
However, 
$(1-P) [\DY(\theta_x,\theta_y),P] P+P[\DY(\theta_x,\theta_y), P] (1-P)=[\DY(\theta_x,\theta_y),P]$ so
\be
\|\partial_{\theta_y} P(\theta_x,\theta_y)-i[\DY(\theta_x,\theta_y),P] \| \leq C(1+D)\err.
\ee

The proof for the bound on
$\|i[\DX(\theta_x,\theta_y),P(\theta_x,\theta_y)]-\partial_{\theta_x} P(\theta_x,\theta_y)\|$ is similar.
\end{proof}
\end{lemma}

Having defined the projector $P(\theta_x,\theta_y)$, we can define a fiber bundle with $\mathbb{C}$ as the fiber, and the projectors define a canonical
connection.
  One can show that the Chern number with this connection is close to the Hall conductance in units
of $e^2/h$.

\subsection{Fractional Hall Effect}
We now discuss the fractional Hall effect.  We do not claim to present
a complete proof of the fractional quantization at the same level of detail
as in the integer case; rather, our goal is to sketch a proof
to indicate where modifications need to be made in comparison to
the integer case.  We will present a detailed proof elsewhere.
In the case of the fractional quantum Hall effect, we consider the same class of
Hamiltonians $H_0$, but now we assume that there are $q$ different
eigenstates of $H_0$, $\Psi_0^a$ for $a=1...q$, which we refer to as ground states.  We assume that
the energy gap between the highest energy ground state and the rest of the
spectrum is at least $\gamma$.

The fractional quantum Hall effect can also be treated using our approach,
with one additional assumption about topological quantum order required and
with an appropriate definition of the Hall conductance, as discussed below.
We first need to define topological quantum order.
Following \cite{bhv}, we define that the $q$ ground states have topological quantum order with error $(l,\epsilon)$ if, for any operator $O_{loc}$ with
$\|O_{loc}\|=1$ supported on a set with diameter $l$ or less,
we have that the $q$-by-$q$ matrix $M$ with matrix elements
\be
M_{ab}=\langle \Psi_0^a,O_{loc} \Psi_0^b\rangle,
\ee
obeys $\| M_{ab}-c\delta_{ab}\|\leq \epsilon$ for some constant $c$.
That is, the projection of $O_{loc}$ into the ground state sector must give
an operator close in operator norm to a multiple of the identity operator.
Given this definition, we require that the system has
$(L/2-1,\epsilon)$ topological order for sufficiently small $\epsilon$.

Next, we need to correctly define the Hall conductance.
This problem was physically discussed by Thouless and Gefen\cite{tg}.
They pointed out that, at experimentally relevant temperatures and
voltages, the terms in the Kubo formula involving matrix
elements between the different ground states should be neglected.
Thus, we should define the Hall conductance in any normalized state
$\Psi_0$ which is a linear combination of ground states $\Psi_0^a$ by
\begin{equation}\label{def:kubofqhe}
\sigma_{xy}(\Psi_0) = 2 \Im \left\lbrace\braket{\frac{\partial}{\partial \theta_y}\Psi_0(0,\theta_y)_{\theta_y=0}}{(1-P_g) \frac{\partial}{\partial \theta_x}\Psi_0(\theta_x,0)_{\theta_x=0}}\right\rbrace \cdot \left(2\pi \, \frac{e^2}{h}\right),
\end{equation}
where $P_g$ is the projector onto the ground state subspace: $P_g=\sum_{a=1}^q |\Psi_0^a\rangle\langle \Psi_0^a|$.  
We will see below that quasi-adiabatic continuation, with an appropriate
choice of $\alpha$, provides a mathematical
realization of the idea of Thouless and Gefen that, in the fractional
Hall effect, experimental measurements are performed slowly enough to
avoid exciting states above some large energy gap, but quickly enough that
we ``shoot through" any avoided crossing of the low energy states \cite{tg,suffcon}.

We will see that this definition gives a Hall
conductance which is almost independent of the particular $\Psi_0$ we choose.
We will define $\sigma_{xy}$ to be the average of $\sigma_{xy}(\Psi_0^a)$
over $a=1,...,q$.
Using these assumptions, we will prove quantization, up to small error, of the Hall conductance in
units of $\frac{e^2}{h\,q}$, with $q$ the dimension of the ground state subspace.
This topological order assumption is certainly necessary to prove this.  To
see why the topological order assumption is necessary, let us
consider an example of a system lacking this assumption.
Imagine that
a system is close to a first order phase transition between a state with filling fraction $\nu=1/3$ and another one with $\nu=1/5$.  By
tuning the parameters in the Hamiltonian, we can construct a situation in which the two phases have the
same energy and the ground state is $8$-fold degenerate.  However, in $3$ of the ground states, we will
observe a conductance close to $1/3$ and in $5$ of them it will be close to $1/5$, neither of which is
a multiple of $1/8$.

We will use the same choices of $\alpha,r$ as in the integer case.
However, throughout this subsection we use $o(1)$ to denote quantities which vanish at fixed $R,J,\gamma$ in the
limit of large $L$ and vanishing $\epsilon$, while we use $O(\epsilon)$ to denote quantities which are bounded by a constant times
$\epsilon$ in the limit of small $\epsilon$ at fixed $R,J,L,\gamma,q,q_{max}$. 
Such quantities which are $O(\epsilon)$ will depend on $R,J/\gamma,L,q,q_{max}$ only polynomially, and hence if $\epsilon$ is exponentially small in
$L$ will be small.
 Due to the detail involved in the
estimates in the integer case, here we will be content to sketch the errors, although, for sufficiently
small $\epsilon$, the final error bounds on quantization (now up to a multiple of $e^2/hq$) will be the same.
We use ``close" to denote that two quantities are within $o(1)$ of each other.

The assumption of topological order implies that the energy difference between
the ground states is at most $\epsilon L^2 J$, since the Hamiltonian is
a sum of local operators, each with energy difference bounded by
$\epsilon J$.  Given this energy difference, it immediately
follows \cite{hast-qad} that the matrix elements of the quasi-adiabatic
evolution operators $\DX,\DY$ between the different ground
states at $\theta_x=\theta_y=0$ are also ${\cal O}(\epsilon)$.  An alternate way to bound the
matrix elements of $\DX,\DY$ between the different ground states is to
note that each one is a sum of operators $S_{\alpha}(H,A_Z)$ and each such operator, up to corrections which are exponentially small in $L$, is local up to length scale $L/2-1$.

One can now repeat the integer proof, with one important difference described below.
First, one can prove that evolution of a state $\Psi_0$ which is
a linear combination of ground states around the small loop of size $r$
near the origin gives a phase which approximates $e^{i\sigma_{xy}(\Psi_0) (r^2/2\pi)(e^2/h)}$.  There are two sources of error here.  The
first is that quasi-adiabatic continuation produces non-vanishing
matrix elements between the ground states, while we have defined
the Hall conductance with the projector $(1-P_g)$ which removes those
intermediate states.  However, quasi-adiabatic continuation
has only small matrix elements between the ground states and these terms
are ${\cal O}(\epsilon)$ for any fixed $R,J,L,\gamma,q,q_{max}$ which in turn implies
a fixed value of $\alpha$.
The second source of error is the same as in the integer case:
errors due to the non-vanishing value of $r$, and errors due to the
finite value of $\alpha$ meaning that we only approximate the
quantity $(1-P_g) \partial_{\theta_x} \Psi_0(\theta_x,0)$ by quasi-adiabatic continuation. 

Since the evolution around the small loop can be approximated by
a local operator, it is close to a constant times the identity operator
when projected into the ground state subspace, and hence that
$\sigma_{xy}(\Psi_0)$ is approximately independent of $\Psi_0$, as claimed.

The proof that the product of evolutions around the small loops is equal to
the evolution around the big loop is identical to the previous case.

The energy estimates and the translation lemma
both rely on the partial trace estimates.
The partial trace estimates can be repeated as before, with
the same error estimates as in the integer case, plus errors which
are ${\cal O}(\epsilon)$.
In proving the partial trace estimates in the fractional case, note that
the product $\exp(-iQ_X \theta) U_X^{(2)}(0,0,\theta)$ 
is close to a phase multiplied by the identity operator
when projected into the ground state subspace.
This can be shown by using the topological
order assumption: since the Hamiltonians
$H(\theta,-\theta,0,0)$ are unitarily equivalent for
all $\theta$, we need to show that
the partial derivative of
$\exp(-iQ_X \theta) U_X^{(2)}(0,0,\theta)$ is close to a multiple of identity operator
after projection into the ground state subspace;
however, this follows because this partial derivative is a sum of local
operators at $\theta=0$.

The translation lemma \ref{lem:translation_unitary} can be repeated as before, with the same error estimates as in the
integer case, plus errors which are ${\cal O}(\epsilon)$.
Thus, we can, with the same error estimates up to terms which are ${\cal O}(\epsilon)$, show that 
the analogue of corollary \ref{cor:stokes} holds, where now $\Psi_0$
in corollary \ref{cor:stokes} can be replaced with any linear combination
of the $\Psi_0^a$, $a=1,...,q$, and $\Psi_{\circlearrowleft}(r)$ is
defined by the quasi-adiabatic continuation of that state:
\begin{eqnarray}
\ket{\Psi^a_{\circlearrowleft}(r)} = V_{\circlearrowleft}(0, 0 , r) \ket{\Psi_0(0,0)} = U^{\dagger}_{Y}(0, 0 , r)\,U^{\dagger}_{X}(0, r , r)\, U_{Y}(r , 0 , r)\, U_{X}(0, 0 , r) \ket{\Psi_0^a(0,0)}.
\end{eqnarray}

The evolution around the big loop can again be exactly decomposed
into a product of evolutions $V^\dagger[(0,0]\rightarrow (\theta_x,\theta_y)]V_\circlearrowleft(\theta_x,\theta_y,r) V[(0,0)\rightarrow (\theta_x,\theta_y)]$
around small loops of size $r$.  
Since the evolution around each small loop is local and so is close to $e^{i \sigma_{xy} (r^2/2\pi) (e^2/h)}$ times
the identity operator in the ground state subspace, this implies
that the unitary $V_\circlearrowleft(0,0,2\pi)$ is close to
a multiple of identity in the ground state subspace.  That is,
$Z_{ab}\equiv \braket{\Psi^a_0}{\Psi_{\circlearrowleft}^b(2\pi)}$ satisfies $\|Z-z\one\| \le o(1)$, for some phase $z$,
with $z$ close to $e^{i \sigma_{xy} 2\pi h/e^2}$.

However, here is the difference.  Unlike the integer case, we do not
know that $Z$ is close to the identity matrix itself, for the following
reason.
The energy estimates given in Proposition \ref{prop:energy} for evolution around one side of the large loop
can again be repeated as before, with the same error estimates up to
terms which are ${\cal O}(\epsilon)$, to show that quasi-adiabatic evolution
of a ground state around that side gives a state which is close to a linear
combination of ground states.  However in this case the operator
$U_X(0,0,2\pi)$ is supported on a set of diameter $L$, so the
assumption of topological quantum order does not tell us that this operator
is close to constant times the identity operator in the ground state sector.

Since we no longer know that the phase cancels between evolution around
opposite sides of the big loop, we must determine what this phase is.
Let $u_x$ and $u_y$ denote $q \times q$ matrices defined as follows:
\be
(u_x)_{ab}=\braket{\Psi^a_0}{U_X(0,0,2\pi) \Psi^b_0}, \quad
(u_y)_{ab}=\braket{\Psi^a_0}{U_Y(0,0,2\pi) \Psi^b_0}.
\ee
Then, $Z=u_y^{\dagger}\, u_x^{\dagger}\, u_y\, u_x$ plus error terms which
are due to leakage out of the ground state sector. In particular, the product of operators
$u_y^{\dagger}\, u_x^{\dagger}\, u_y\, u_x$ corresponds to projecting back into the
ground state sector after each side of the big loop, hence the error terms which come from projecting
off the ground state are small due to Proposition \ref{prop:energy}, in a manner similar to the proof of Corrolary \ref{cor:gs_evol}.  
Since $\|Z-z\one\| \le o(1)$ for some
phase $z$, we have
$\|u_y^{\dagger}\, u_x^{\dagger}\, u_y\, u_x-z\one\| \le o(1)$.  Moreover, we have that
$\|u_x^{\dagger} u_x-\one\|\le o(1), \, \|u_y^{\dagger} u_y-\one\|\le o(1)$
since, for example,
$$(u_x^{\dagger} u_x)_{ab} = \braket{\Psi^a_0}{U^{\dagger}_X(0,0,2\pi) P_0 U_X(0,0,2\pi) \Psi^b_0} = \delta_{ab} - \braket{\Psi^a_0}{U^{\dagger}_X(0,0,2\pi) (1-P_0) U_X(0,0,2\pi) \Psi^b_0},$$
with $P_0$ the projection onto the ground state sector and the second term above is small due to the energy estimate in Proposition \ref{prop:energy} applied to the quasi-adiabatic continuation with $U_X(0,0,2\pi)$ of $\Psi^a_0$ and $\Psi^b_0$, separately.
Finally, we note that for $q \times q$ matrices $A, B$ satisfying $B = z\one + A$, with $\|A\| \leq o(1)$, we have 
$|\rm{det}(B) -z^q| \leq o(1)$.
Putting everything together and using the triangle inequality, we get from 
$\rm{det}(u_y^{\dagger}\, u_x^{\dagger}\, u_y\, u_x) = \rm{det}(u_x^{\dagger}\, u_x) \cdot \rm{det}(u_y^{\dagger}\, u_y)$ and $\max\{|\rm{det}(u_x^{\dagger}\, u_x)-1|, |\rm{det}(u_y^{\dagger}\, u_y)-1|\} \le o(1)$: 
$$|\rm{det}(u_y^{\dagger}\, u_x^{\dagger}\, u_y\, u_x) -z^q| \le o(1) \implies |\rm{det}(u_x^{\dagger}\, u_x) \cdot \rm{det}(u_y^{\dagger}\, u_y) -z^q| \le o(1) \implies |1 - z^q| \le o(1).$$
Hence, $z$ is close to a $q$-th root of unity and hence $\sigma_{xy}$ is close to an integer multiple of $(1/q) e^2/h$.

{\bf Acknowledgments:} MBH thanks M. Freedman, C. Nayak, and T. Osborne for useful discussions.
SM thanks B. Nachtergaele and R. Sims for useful
discussions on Lieb-Robinson bounds.  SM thanks the organizers of
the workshop on ``Quantum Information and Quantum Spin Systems" at
the Erwin Schr\"{o}dinger Institute and the organizers of the program on ``Quantum Information Science"
at the KITP at UC Santa Barbara, where parts of this work were completed.  SM was supported
by NSF Grant DMS-07-57581 and DOE Contract DE-AC52-06NA25396.

\appendix
\begin{center}
    {\bf APPENDIX}
  \end{center}
\section{Estimating the difference of powers}
In this section, we prove the implication (\ref{bound:power_difference}). First, we have $b \ge |e^{i\theta}| - |b-e^{i\theta}| \ge 1-\epsilon \ge 1/2$ and using induction it is easy to show that $b^m \ge 1- m\epsilon$. Moreover,
using Euler's formula $e^{i\theta}=\cos \theta + i\sin\theta$, we have that 
$$|b-e^{i\theta}|^2 = (1-b)^2 + 2b(1-\cos\theta) = (b-\cos\theta)^2+\sin^2\theta \le \epsilon^2 \implies \sin|\theta| \le \epsilon, \,\cos|\theta| \ge \sqrt{1-\epsilon^2} \ge \sqrt{3}/2,$$
which implies that $\cos \phi$ is decreasing for $\phi \in [0,|\theta|]$ and hence, $\sin|\theta| = \int_0^{|\theta|} \cos\phi\, d\phi \implies \sin|\theta| \ge |\theta| \cos|\theta| \implies |\theta|\le 2\epsilon/\sqrt{3}$. Moreover, 
$1-\cos(m\theta) = \int_0^{m\theta} d\phi \int_0^\phi \cos\eta \,d\eta \le \int_0^{m\theta} |\phi| \,d\phi \le m^2\,|\theta|^2/2\le \frac{2}{3} m^2\, \epsilon^2$.
Armed with these bounds, we get the inequality $\left|b^m-e^{im\theta}\right| \le \sqrt{\frac{7}{3}}\,m\,\epsilon$ as follows:
\begin{eqnarray*}
\left|b^m-e^{im\theta}\right|^2 = (1-b^{m})^2 + 2b^m\,(1-\cos(m\theta))\le \frac{7}{3}\,m^2\,\epsilon^2,
\end{eqnarray*}
and taking square roots completes the proof.

\subsection{Bounding the norm of higher partials}
The following general formula for a differentiable family of Hamiltonians $H(\theta)$, which can be verified by differentiating both sides with respect to $t$, will be instrumental in the analysis that follows:
\begin{equation}\label{eq:dH}
\left(\partial_{\phi} e^{itH(\phi)}\right)_{\phi=\theta} e^{-itH(\theta)} =  -e^{itH(\theta)} \left(\partial_{\phi} e^{-itH(\phi)}\right)_{\phi=\theta} =  i \int_0^t \tau_u^{H(\theta)} \left(\partial_{\theta} H(\theta)\right) du 
\end{equation}
Using the above identities, we derive:
\begin{eqnarray}
&&\partial_{\phi} \left\{\tau_u^{H(\phi,0,\theta_y,0)}(\partial_{\theta}H(\phi,0,\theta,0))_{\theta=\theta_y}\right\}_{\phi=\theta_x} =\tau_u^{H(\theta_x,0,\theta_y,0)}\left(\left(\partial_{\phi} \left(\partial_{\theta} \, H(\phi,0,\theta,0)\right)_{\theta=\theta_y}\right)_{\phi=\theta_x}\right)+ \nonumber\\
&& \left[\left\{\left(\partial_{\phi} e^{iuH(\phi,0,\theta_y,0)}\right)_{\phi=\theta_x} e^{-iuH(\theta_x,0,\theta_y,0)}\right\}, \tau_u^{H(\theta_x,0,\theta_y,0)}(\partial_{\theta}H(\theta_x,0,\theta,0)_{\theta=\theta_y})\right],\\
&&\partial_{\phi} \left\{\tau_u^{H(\theta_x,0,\phi,0)}(\partial_{\theta}H(\theta_x,0,\theta,0))_{\theta=\phi}\right\}_{\phi=\theta_y}=\tau_u^{H(\theta_x,0,\theta_y,0)}\left(\left(\partial^2_{\theta}\, H(\theta_x,0,\theta,0)\right)_{\theta=\theta_y}\right)+\nonumber\\
&& \left[\left\{\left(\partial_{\phi} e^{iuH(\theta_x,0,\phi,0)}\right)_{\phi=\theta_y} e^{-iuH(\theta_x,0,\theta_y,0)}\right\}, \tau_u^{H(\theta_x,0,\theta_y,0)}(\partial_{\theta}H(\theta_x,0,\theta,0)_{\theta=\theta_y})\right].
\end{eqnarray}
From (\ref{eq:dH}) we have:
\begin{eqnarray*}
\left(\partial_{\phi} e^{iuH(\phi,0,\theta_y,0)}\right)_{\phi=\theta_x} e^{-iuH(\theta_x,0,\theta_y,0)}
&=& i \int_0^u \tau_s^{H(\theta_x,0,\theta_y,0)}(\partial_{\phi}H(\phi,0,\theta_y,0)_{\phi=\theta_x}),\\
\left(\partial_{\phi} e^{iuH(\theta_x,0,\phi,0)}\right)_{\phi=\theta_y} e^{-iuH(\theta_x,0,\theta_y,0)}
&=& i \int_0^u \tau_s^{H(\theta_x,0,\theta_y,0)}(\partial_{\phi}H(\theta_x,0,\phi,0)_{\phi=\theta_y})
\end{eqnarray*}
Using the above formulas with the definition of $\DY(\theta_x,\theta_y)$, we get:
\begin{eqnarray*}
\left\|\left(\partial_\phi \DY(\phi,\theta_y)\right)_{\phi=\theta_x}\right\| &\le& \left\|\left(\partial_{\phi} \left(\partial_{\theta} \, H(\phi,0,\theta,0)\right)_{\theta=\theta_y}\right)_{\phi=\theta_x}\right\| 
\left(\int_{-\infty}^{\infty} |t| \,s_{\alpha}(t)\, dt\right)\nonumber\\
&+&\left\|\left(\partial_{\theta} \, H(\theta_x,0,\theta,0)\right)_{\theta=\theta_y}\right\| \left\|\left(\partial_{\phi} \, H(\phi,0,\theta_y,0)\right)_{\phi=\theta_x}\right\| 
\left(\int_{-\infty}^{\infty} t^2 \,s_{\alpha}(t) \, dt\right)
\\
\left\|\left(\partial_\phi \DY(\theta_x,\phi)\right)_{\phi=\theta_y}\right\| &\le& \left\|\left(\partial^2_{\phi}\, H(\theta_x,0,\phi,0)\right)_{\phi=\theta_y}\right\| \left(\int_{-\infty}^{\infty} |t| \,s_{\alpha}(t)\, dt\right)
+\left\|\left(\partial_{\phi} \, H(\theta_x,0,\phi,0)\right)_{\phi=\theta_y}\right\|^2 \left(\int_{-\infty}^{\infty} t^2 \,s_{\alpha}(t) \, dt\right)
\end{eqnarray*}
with similar bounds for partials involving $\DX(\theta_x,\theta_y)$.
Now, recalling that $s_{\alpha}(t)$ is a Gaussian, we have $\int_{-\infty}^{\infty} |t| \,s_{\alpha}(t)\, dt = \frac{2\alpha}{\sqrt{2\pi}}$ and $\int_{-\infty}^{\infty} t^2 \,s_{\alpha}(t)\, dt = \alpha^2$. Finally, using estimates similar to the ones derived in (\ref{bound:rot-H_X}-\ref{bound:rot-H_Y}), it should be clear that the above partials have norms bounded by $C\, (Q_{\max} \alpha J\, L)^2$, for some $C >0$. In general, one could show that $m$-th order partials of $\DY(\theta_x,\theta_y)$ and $\DX(\theta_x,\theta_y)$ would have norms bounded by $C\, (Q_{\max} \alpha J\, L)^{m+1}$. The same argument applies to partials of $\DY^{(M)}(\theta_x,\theta_y)$ and $\DX^{(M)}(\theta_x,\theta_y)$, after we decompose each into a sum of individual interactions.
 
\section{Decomposing the evolution, one loop at a time.}
Let $N= (2\pi)/r$ and define $U_{(N-m)+nN} = V^{\dagger}(m\cdot r,n\cdot r)V_{\circlearrowleft}(m\cdot r,n\cdot r,r) V(m\cdot r,n\cdot r),$ for $m,n\in[0,N-1]$ to be the evolution operator corresponding to the cyclic path leading to and around the square with lower-left corner at $(m\cdot r,n\cdot r)$ in flux-space (see Fig.~\ref{fig:decomposition}). Then, we have the following identity:
\begin{equation}
V_{\circlearrowleft}(0,0,2\pi) = U_{N^2}\, U_{N^2-1} \cdots U_3\, U_2\, U_1
\end{equation}
as can be verified by the decomposition described in Fig.~\ref{fig:decomposition}. To help us track the contribution from each term in the decomposition as we insert $P_0=\pure{\Psi_0}$ and $Q_0=\one-P_0$ after each cyclic evolution, we introduce the scalars 
$$p_{[s,r]} = \braket{\Psi_0}{U_r\, U_{r-1}\cdots U_s \Psi_0},\qquad \mbox{and}\qquad
q_{[s,r]} = \bra{\Psi_0}\, U_r\, U_{r-1}\cdots U_{s-1} Q_0 U_s \ket{\Psi_0}.$$
Moreover, we note that $p_r = \braket{\Psi_0}{U_r \Psi_0}$ and in particular, $p_N = \braket{\Psi_0}{\Psi_\circlearrowleft(0,0,r)} = \braket{\Psi_0}{\Psi_\circlearrowleft(r)}$.
The quantity we are trying to bound in (\ref{eq:stokes}) is $$\left|\braket{\Psi_0}{\Psi_{\circlearrowleft}(2\pi)}-\braket{\Psi_0}{\Psi_{\circlearrowleft}(r)}^{\left(\frac{2\pi}{r}\right)^2}\right| = \left|p_{[1,N^2]}-(p_N)^{N^2}\right|.$$
Using the triangle inequality, we get:
\be
\left|p_{[1,N^2]}-(p_N)^{N^2}\right| \le \left|p_{[1,N^2]}-p_1\, p_{2}\cdots p_{N^2-1}\, p_{N^2}\right| + \left|p_1\, p_{2}\cdots p_{N^2-1}\, p_{N^2}-(p_N)^{N^2}\right|.
\ee 
so it suffices to bound each term on the right.
We start with a bound for $\left|p_{[1,N^2]}-p_1\, p_{2}\cdots p_{N^2-1}\, p_{N^2}\right| $. Using the following recursive relationship:
\begin{eqnarray*}
p_{[1,N^2]} &=& p_1 p_{[2,N^2]} + q_{[1,N^2]}\\
p_{[2,N^2]} &=& p_2 p_{[3,N^2]} + q_{[2,N^2]}\\
&\vdots &
\end{eqnarray*}
we may write:
\begin{equation}\label{bound:p_0}
p_{[1,N^2]} - p_1\, p_{2}\cdots p_{N^2-1}\, p_{N^2} = \sum^{N^2-1}_{i=1} p_1\cdots p_{i-1} q_{[i,N^2]}
\end{equation}
We turn our attention to $|q_{[i,N^2]}|$. We begin by observing that the following series of bounds hold:
\begin{eqnarray}
|q_{[i,N^2]}| &=& |\braket{\Psi_0}{U_{N^2}\cdots U_{i+1} Q_0 U_i \Psi_0}| \le \|Q_0 U_i\ket{\Psi_0}\|
= \sqrt{1 - |p_i|^2} \le \sqrt{2(1 - |p_i|)} \le \sqrt{2(1-|p_N|)} + \sqrt{2|p_i-p_N|}.\nonumber
\end{eqnarray}
Since, we have $N^2-1$ terms in (\ref{bound:p_0}), the triangle inequality  and the above bound together with the fact that $|p_i|\le 1$ for all $i\ge 1$, imply:
\begin{equation}\label{off_gs}
\left|p_{[1,N^2]} - p_1\, p_{2}\cdots p_{N^2-1}\, p_{N^2}\right|\le (N^2-1) \left(\sqrt{2(1-|p_N|)} + \sup_{i\in [1,N^2]} \sqrt{2|p_i-p_N|}\right).
\end{equation}
Moreover, assuming $\sup_{i\in[1,N^2]} |p_i - p_N| = \delta/N^2$ for some $\delta \in [0,1]$, we get:
\begin{equation}
\left|(p_{N})^{N^2} - p_1\, p_{2}\cdots p_{N^2-1}\, p_{N^2}\right| \le \left(1+\delta/N^2\right)^{N^2}-1 \le e^{\delta}-1 \le \int_0^{\delta} e^y \, dy \le e^{\delta} \, \delta \le e\cdot \delta.
\end{equation}
where the first inequality follows from expanding the product 
\begin{equation*}
p_1\, p_{2}\cdots p_{N^2-1}\, p_{N^2} = (p_N + [p_1-p_N])\cdots (p_N + [p_{N^2-1}-p_N])\,(p_N + [p_{N^2}-p_N])
\end{equation*}
and then, after subtracting the term $(p_N)^{N^2}$, using the triangle inequality along with $|p_N|\le 1$. The second inequality follows from comparing the binomial expansion of $(1+x/m)^m$ term by term with the Taylor expansion of $e^x$ and using the simple inequality $\binom{m}{k} \le \frac{m^k}{k!}$.
Putting everything together, we get the final bound:
\be\label{bound:p_1}
\left|\braket{\Psi_0}{\Psi_{\circlearrowleft}(2\pi)}-\braket{\Psi_0}{\Psi_{\circlearrowleft}(r)}^{\left(\frac{2\pi}{r}\right)^2}\right| \le N^2 \left(\sqrt{2(1-|p_N|)} + \sqrt{2\, \sup_{i\in [1,N^2]} |p_i-p_N|} + e \sup_{i\in [1,N^2]} |p_i-p_N| \right)
\ee
recalling that $N = \frac{2\pi}{r}$ and noting that $1-|p_N|$ is bounded in (\ref{adiabatic_phase_1}), while $ \sup_{i\in [1,N^2]} |p_i-p_N|$ is bounded in (\ref{eq:translation}).



\begin{thebibliography}{99}
\bibitem{hast-lsm}
M. B. Hastings,
Phys. Rev. B {\bf 69}, 104431 (2004).

\bibitem{osborne} T. J. Osborne, 
Phys. Rev. A {\bf 75}, 032321 (2007).

\bibitem{laughlin} R. B. Laughlin, ``Quantized Hall conductivity in two dimensions", Phys. Rev. B {\bf 23}, 5632 (1981).

\bibitem{avron}
J.E. Avron and R. Seiler, 
Phys. Rev. Lett. {\bf 54} 259-262 (1985).

\bibitem{niu} Q. Niu, D. J. Thouless, and Y.-S. Wu, Phys. Rev. B {\bf 31}, 3372 (1985).

\bibitem{ncg} J. Bellisssard, A. van Elst, and H. Shultz-Baldes, ``The Noncommutative Geometry of the Quantum Hall effect", J. Math. Phys. {\bf 35}, 5373
(1994).

\bibitem{loc-estimates}
 B. Nachtergaele and R. Sims,
  \emph{New Trends in Mathematical Physics}.
 Springer, 2009

\bibitem{lr} E. H. Lieb and D. W. Robinson, Commun. Math. Phys. {\bf 28},
251 (1972).

\bibitem{hast-koma} M.  B. Hastings and T. Koma,
Commun. Math. Phys. {\bf 265}, 781 (2006).

\bibitem{ns} B. Nachtergaele and R. Sims, Commun. Math. Phys.
{\bf 265}, 119 (2006).

\bibitem{qad} M. B. Hastings and X.-G. Wen, Phys. Rev. B {\bf 72}, 045141
(2005).

\bibitem{hast-qad} M. B. Hastings,
JSTAT, P05010 (2007).

\bibitem{thouless} D. J. Thouless, ``Topological Quantum Numbers in Nonrelativistic Physics", World Scientific (1998).

\bibitem{bhv} S. Bravyi, M. B. Hastings, and F. Verstraete, Phys. Rev. Lett.
{\bf 97}, 050401 (2006).

\bibitem{tg} D. J. Thouless and Y. Gefen, Phys. Rev. Lett. {\bf 66}, 806 (1991).

\bibitem{suffcon} M. B. Hastings,
``Sufficient Conditions for Topological Order in Insulators",
Europhys. Lett. {\bf 70}, 824 (2005).

\bibitem{anh} B. Nachtergaele, H. Raz, B. Schlein, and R. Sims, Commun.
Math. Phys. {\bf 286}, 1073 (2009).
\end{thebibliography}
\end{document}